\documentclass{article}

\usepackage{arxiv}

\usepackage[utf8]{inputenc} 
\usepackage[T1]{fontenc}    
\usepackage[colorlinks,citecolor=blue,urlcolor=blue]{hyperref}       
\usepackage{url}            
\usepackage{booktabs}       
\usepackage{amsthm,amsmath,amsfonts,amssymb}       
\usepackage{nicefrac}       
\usepackage{microtype}      
\usepackage{cleveref}       
\usepackage{lipsum}         
\usepackage{graphicx}
\usepackage[authoryear,round]{natbib}
\usepackage{doi}
\usepackage{bm,bbm}
\usepackage{subfigure}
\usepackage{float}
\usepackage{algorithm}  
\usepackage{algorithmic}
\usepackage{textcomp}

\numberwithin{equation}{section}
\theoremstyle{plain}
\newtheorem{thm}{Theorem}[section]
\newtheorem{prop}{Proposition}[section]
\newtheorem{definition}{Definition}
\newtheorem{remark}{Remark}[section]
\newtheorem{lemma}{Lemma}[section]

\title{Statistical Depth for Point Process via the Isometric Log-Ratio Transformation}

\author{ Xinyu Zhou\\
	Department of Statistics\\
	Florida State University\\
	Tallahassee, FL 32306 \\
	\texttt{xz19c@my.fsu.edu} \\
	\And
	Yijia Ma \\
	Department of Statistics\\
	Florida State University\\
	Tallahassee, FL 32306 \\
	\texttt{ym19f@my.fsu.edu} \\
	\And
	Wei Wu \\
	Department of Statistics\\
	Florida State University\\
	Tallahassee, FL 32306 \\
	\texttt{wwu@stat.fsu.edu} \\
}

\renewcommand{\shorttitle}{ILR Depth for Point Process}

\hypersetup{
pdftitle={Statistical Depth for Point Process via the Isometric Log-Ratio Transformation},
pdfsubject={statistical depth},
pdfauthor={Xinyu Zhou, Yijia Ma, Wei Wu},
pdfkeywords={Statistical depth, Point process, Isometric Log-Ratio transformation, Poisson process, Time rescaling},
}

\begin{document}
\maketitle

\begin{abstract}
Statistical depth, a useful tool to measure the center-outward rank of multivariate and functional data, is still under-explored in temporal point processes.  Recent studies on point process depth proposed a weighted product of two terms - one indicates the depth of the cardinality of the process, and the other characterizes the conditional depth of the temporal events given the cardinality.  The second term is of great challenge because of the apparent nonlinear structure of event times, and so far only basic parametric representations such as Gaussian and Dirichlet densities were adopted in the definitions. However, these simplified forms ignore the underlying distribution of the process events, which makes the methods difficult to interpret and to apply to complicated patterns.  To deal with these problems, we in this paper propose a distribution-based approach to the conditional depth via the well-known Isometric Log-Ratio (ILR) transformation on the inter-event times.  The new depth, called the ILR depth, is at first defined for homogeneous Poisson process by using the density function on the transformed space.  The definition is then extended to any general point process via a time-rescaling transformation.  We illustrate the ILR depth using simulations of Poisson and non-Poisson processes and demonstrate its superiority over previous methods.  We also thoroughly examine its mathematical properties and asymptotics in large samples.  Finally, we apply the ILR depth in a real dataset and the result clearly shows the effectiveness of the new method.  
\end{abstract}


\section{Introduction} \label{sec:introduction}

In this paper, we study the center-outward rank in a set of temporal point process observations. 
Temporal point process (point process for short) is an important area in stochastic process, which has been extensively studied in both theory and applications. A point process is basically a list of event times in increasing order. Many real phenomena can produce data that can be represented as a point pattern. For instance, traffic accident times in a specific intersection, scoring times in a soccer match, or earthquake happening times in a geographical region. Point process  is a preferred tool to model and analyze such practical observations with random variability. The term ``point'' is used since an event can be thought as being instant and represented as a point on the time line.  In this study, we focus on orderly point process, where the process can be fully characterized with a conditional intensity function. 



Statistical depth has been a powerful tool to measure the center-outward rank of multivariate data and functional data. It was first studied on multivariate data by \citet{tukey1975mathematics} of a given data group in Euclidean space. Later, \citet{donoho1992breakdown} used hyperplanes idea to define depth for an arbitrary point in Euclidean space corresponding to the data group. From then on, various new depths on multivariate data were developed and the results have been fruitful. These methods include Oja depth \citep{oja1983descriptive}, simplicial depth \citep{liu1990notion}, Mahalanobis depth \citep{liu1993quality}, projection depth \citep{zuo2000note}, zonoid depth \citep{dyckerhoff1996zonoid}, and likelihood depth \citep{fraiman1999multivariate}.  Many of these methods are well-known and commonly used to build templates or identify outliers in given observations. Recent investigations on depth have focused more on functional observations. \citet{lopez2009concept} developed the notion of functional depth for the first time. Then, \citet{nieto2011properties} and \citet{mosler2012general} conducted researches on the general concepts and properties of functional depth. More applications of functional depth have been studied on outlier detection \citep{narisetty2016extremal} and data classification \citep{makinde2019classification}. 

The notion of depth has recently been introduced to non-Euclidean metric space  to deal with more complicated data structures in the era of big data \citep{geenens2021statistical,dai2021tukey}.  The methods were illustrated with examples in non-Euclidean function space, Riemannian manifold, and geodesic space.  It is clear that the point process is also a non-Euclidean metric space \citep{wu2011information} and the metric-based depth methods can be directly applied.  Moreover, the notion of depth for point process data has apparent benefits: 1. The center-outward ranking process can enhance the understanding of the temporal variability in the point process.  High depth-valued processes represent typical pattern in the data, and low depth-valued processes indicate outlier observations.  This provides a new and useful tool in addition to the traditional time-based frameworks.  2. These ranking result has broad applications.  For example, as we will show in this paper, the depth analysis can help understand the typical pattern in daily traffic accidents at a certain region and this information is important for proper management of social safety resources.  In addition, the depth tool in neuronal spike train study can help identify typical firing patterns and remove the outlier process to improve modeling and neural coding analysis.


Although important, the notion of depth in point process remains under-explored. The only two previous approaches on point process data are the generalized Mahalanobis depth \citep{liu2017generalized} and the Dirichlet depth \citep{qi2021dirichlet}. We point out both methods have clear limitations. At first, \citet{liu2017generalized} adopted the well-known Mahalanobis depth \citep{liu1993quality} to capture the structure of the point process given its cardinality. However, this depth ignores two important features in point process.  One is the finite time domain and the other is the ordered property of all events.  Any point process realization, even with unordered events times or events out of the given domain, will have a positive generalized Mahalanobis depth value.  This makes the proposed depth method inappropriate to use in practice.



In a recent study, \citet{qi2021dirichlet} introduced a new depth called Dirichlet depth for the point process. The process was equivalently represented by using inter-event times (IET) and the depth was defined based on the Dirichlet density function on the IET domain, which is a simplex.  Such approach appropriately addresses the issues on ordered events and bounded time domain. 
However, there are still other limitations. One is the choice of concentration parameter in the Dirichlet density, which was chosen for simplicity and lacked theoretical support. The other limit is that there is no clear symmetry for the deepest point, which is a common requirement in both multivariate \citep{zuo2000general} and functional depths \citep{nieto2011properties}. 

We point out that the limitations in the above depth methods are due to their constrained domains.  For the generalized Mahalanobis depth, its domain is all increasing events in a finite interval.  For the Dirichlet depth, its domain is a finite simplex (after the IET transformation).  In this paper, we aim to remove those constraints and define a new depth in an unconstrained Euclidean space.  Our approach is based on the well-known Isometric Log-Ratio (ILR) transformation on the IETs, and a new symmetry will be introduced to formally define a center with the largest depth value. The ILR transformation is a Log-Ratio analysis method and commonly used in compositional data analysis \citep{aitchison2000logratio}. It was formally defined by \citet{egozcue2003isometric} as an isometric isomorphism between the simplex $\mathcal{S}^D$ and Euclidean space $\mathbb{R}^{D-1}$, where $\mathcal{S}^D = \{(x_1, \cdots, x_D) \in \mathbb{R}^{D} \mid \sum_{i=1}^D x_i = C, x_i > 0, i = 1, \cdots, D\}$, $C$ is a positive constant, and $D>1$ is an integer. The key benefit of the ILR transformation is that it provides an isometric bijection between the constrained space $\mathcal{S}^D$ and the unconstrained space $\mathbb{R}^{D-1}$. For any point process with $D-1$ events in a finite time domain, its IETs can be equivalently transformed to a vector in $\mathbb{R}^{D-1}$.


In this paper, we will define a density-based depth in the unconstrained space $\mathbb{R}^{D-1}$ after the IET and ILR transformations.  As both transformations are invertible, we can map the depth values to the original point process observations.  This procedure can eliminate all limitations in the previous methods.  For homogeneous Poisson process, we will show that the density of the ILR transformed data is given in a closed form in the unconstrained Euclidean space. This density function is similar to a Gaussian density that it has one global maximum in the center and the function value consistently decreases from the center to boundary.  Inspired by the conventional Mahalanobis depth \citep{liu1993quality}, a new depth can be defined by this density function.  

We will show in Section \ref{sec:dep_method} that our proposed new depth method has the following benefits:  1) It is built under a rigorous mathematical framework and can maximally exploit the distribution pattern in the given data. 2)  For homogeneous Poisson process, the density after the ILR transformation is given in a closed form. 3) The density has a center-outward pattern with a clear center under orthogonal transformations.  4)  The density naturally leads to a new definition of depth, which satisfies all important properties for depth functions.   5) The new depth definition can be easily extended to any general point process with the well-known \textit{Time Rescaling Method} \citep{brown2002time}.

The rest of this paper is organized as follows. In Section \ref{sec:dep_method}, we will at first introduce the depth definition based on the ILR transformation for homogeneous Poisson process.  We will thoroughly examine its mathematical properties and provide simulations for illustration.  We will then provide definition on general point process and specifically focus on two cases: inhomogeneous Poisson process and inhomogeneous Markov interval process. Estimations of conditional intensity and simulation examples in both cases will be provided.  Comparisons with previous depth methods will also be conducted. In Section \ref{sec:asym}, we present the asymptotic theory on the sample depth when the underlying process is an inhomogeneous Poisson process. In Section \ref{sec:real_data_app}, we will apply the new depth to a real world dataset to demonstrate its effectiveness in characterization of typical patterns.  Finally, we will summarize our study and provide future work in Section \ref{sec:future}. All mathematical details are shown in appendices.

\section{Depth Methods} \label{sec:dep_method}
In this section, we will provide all details of the proposed methods.  We adopt the depth definition for point process in \citet{qi2021dirichlet} and will at first provide a review of the basic framework.


\subsection{Depth definition for point process} \label{sec:depdef}
Let $\mathbb{S}$ denote the set of all point processes in the time domain $[T_1,T_2]$ and $\mathbb{S}_k$ denote the set of all point processes with cardinality $k$ in the time domain $[T_1,T_2]$, e.g. $\mathbb{S}_k=\{(s_1,s_2,\dots,s_k)^T \in\mathbb{R}^k|T_1\leq s_1\leq s_2\leq \dots\leq s_k \leq T_2\}$ and $k$ is any non-negative integer. Therefore, $\mathbb{S}=\bigcup_{k=0}^{\infty}\mathbb{S}_k$.  For any point process $\bm{s}\in\mathbb{S}$, a depth function maps from $\mathbb{S}$ to $\mathbb{R}^+$. 
The boundary set for point process with cardinality $k$ is denoted as $\mathbb{B}_k=\{(s_1,s_2,\dots,s_k)^T \in\mathbb{S}_k|\text{at least one equality holds in: }$ $T_1\leq s_1\leq s_2\leq \dots\leq s_k \leq T_2\}$. 


In this paper, we adopt the overall framework in \citet{qi2021dirichlet}, where the depth is defined as the product of two terms.  For each process, the first term is a normalized one dimensional depth of the number of time events, and the second term is a conditional depth of point process given its cardinality. The formal depth definition is given below: 
\begin{definition} \label{def:wholedef}
For point process $S\in \mathbb{S}$ defined on $[T_1,T_2]$ with probability measure $P$, denote $P_{|S|}$ as a probability measure on the cardinality $|S|$ and $P_{S||S|}$ as the probability measure on the ordered events $S$ given $|S|$. For a realization $\bm{s}\in \mathbb{S}$, the depth $D(\bm{s};P)$ is defined as: 
\begin{eqnarray}
D(\bm{s};P)=w(|\bm{s}|;P_{|S|})^rD_c(\bm{s};P_{S||S|})
\end{eqnarray}
where $w(|\bm{s}|;P_{|S|})=\frac{D_1(|\bm{s}|;P_{|S|})}{\max_kD_1(|\bm{s}|=k;P_{|S|})}$ is the normalized one dimensional depth on the cardinality $|\bm{s}|$, $D_1(|\bm{s}|;P_{|S|})=\min\{P_{|S|}(|S|\leq |\bm{s}|),P_{|S|}(|S|\geq |\bm{s}|)\}$, $r>0$ is a hyper-parameter and $D_c(\bm{s};P_{S||S|})$ is the depth of $\bm{s}$ conditioned on $|\bm{s}|$. 
\end{definition}

\begin{remark}
There are many methods to estimate the one dimensional depth $D_1(|\bm{s}|;P_{|S|})$, in this paper we adopt the same approach in \citet{qi2021dirichlet}. In practice, $D_1(|\bm{s}|;P_{|S|})$ and $w(|\bm{s}|;P_{|S|})$ can be easily estimated by samples if the population result is unknown or difficult to obtain.
\end{remark}


The second term $D_c(\bm{s};P_{S||S|})$ is the main focus of this paper.  Note the conventional multivariate depth may not be directly used because the time events are in a non-Euclidean space (increasing sequence in the finite domain $[T_1,T_2]$). 
In this paper, we will propose to transform point process data to a Euclidean domain, and then utilize its density function to define the depth.  The detail is given in the next subsection.    

We often use various mathematical properties to evaluate the performance of a depth method.  A list of four properties is given in \citet{qi2021dirichlet} for the conditional depth in point process, which corresponds to the similar four properties in the multivariate case \citep{zuo2000general}.  These properties are listed below and we will evaluate them in our proposed depth:
\begin{enumerate}
\item \label{m1} $D_c(\bm{s};P_{S||S|=k})$ is a continuous mapping from $\mathbb{S}_k$ to $\mathbb{R}^+$ and $D_c(\bm{s};P_{S||S|=k})=0$ if $\bm{s}\in \mathbb B_k$. 
\item \label{m2} There exists unique $\bm{s}_c$ such that $D_c(\bm{s}_c;P_{S||S|=k})=\sup_{\bm{s}\in \mathbb{S}_k}D_c(\bm{s};P_{S||S|=k})$ for any $P_{S||S|=k}\in \mathcal{P}_k$, where $\bm{s}_c$ is the center point given a specific symmetry. 
\item \label{m3} If $\bm{s}_c$ is the center point, then $D_c(\bm{s};P_{S||S|=k})\leq D_c(\bm{s}_c+\alpha(\bm{s}-\bm{s}_c);P_{S||S|=k})$ for any $\bm{s}\in \mathbb{S}_k$ and $\alpha\in [0,1]$. 
\item \label{m4} For any scaling coefficient $a\in\mathbb{R}^+$ and translation coefficient $b\in \mathbb{R}$, $D_c(\bm{s};P_{S||S|=k})=D_c(a\bm{s}+b;P_{aS+b||S|=k})$
\end{enumerate}

\subsection{The ILR transformation on simplex} \label{sec:pdf}

Now we focus on point process $\bm{s}=(s_1,s_2,\dots,s_k)$ with given cardinality $k$ in the time domain $[T_1,T_2]$ and will provide an equivalent way to represent the process.

\subsubsection{Equivalent representation}

Denote $s_0=T_1$ and $s_{k+1}=T_2$. Then the process can be equivalently represented using a vector of the inter-event times (IET), obtained as $\bm{u}=(u_1,u_2,\dots,u_{k+1})^T=(s_1-s_0,s_2-s_1,\dots,s_{k+1}-s_k)^T$.  It is easy to see that $\sum_{i=1}^{k+1}u_{i}=T_2-T_1$ for $u_{i} \geq 0$, $i=1,\dots,k+1$.  That is, all IET vectors form a simplex $\mathcal{S}^{k+1}$ in $\mathbb R^{k+1}$, where
\[
\mathcal{S}^{k+1} = \{\bm{u}=(u_1,u_2,\dots,u_{k+1})^T \mid \sum_{i=1}^{k+1}u_{i}=T_2-T_1, u_{i} \geq 0, i = 1, \cdots, k+1\}. 
\]
 One previous approach is to assume the IET observations follow a Dirichlet distribution, which motives the definition of Dirichlet depth \citep{qi2021dirichlet}.  However, we point out that a Dirichlet model is only a simplified assumption and in general it is very challenging to model data in simplex. In this paper, we propose to adopt the well-known Isometric Log-Ratio (ILR) transformation to map the IET vectors to a conventional vector space. In this manner, we can examine the distribution of the IET vectors in the equivalent, unconstrained Euclidean space.

The ILR transformation is an isometric isomorphism mapping from simplex space $\mathcal{S}^{k+1}$ to Euclidean space $\mathbb{R}^{k}$.  Specifically, the transformation of any $\bm{u}=(u_1,u_2,\dots,u_{k+1})^T\in\mathcal{S}^{k+1}$ is given in the following form \citep{pawlowsky2007lecture}: 
\begin{eqnarray}
\bm{u}^*=ilr(\bm{u})  =   \Psi\cdot\Big[\log\frac{u_1}{g(\bm{u})},\log\frac{u_2}{g(\bm{u})},\dots,\log\frac{u_{k+1}}{g(\bm{u})}\Big]^T \label{eq:ilr}
\end{eqnarray}
where $g(\bm{u})$ is the geometric mean of $\bm{u}$.  $\Psi$ is a matrix in $\mathbb R^{k\times (k+1)}$ which satisfies $\Psi\Psi^T=I_k$ and $\Psi^T\Psi=I_{k+1}-\frac{1}{k+1}\bm{1}_{k+1}\bm{1}_{k+1}^T$, where $I_k$ is the identity matrix in $\mathbb R^{k \times k}$, $I_{k+1}$ is the identity matrix in $\mathbb R^{(k+1) \times (k+1)}$, and $\bm{1}_{k+1}$ is a column vector of ones in $\mathbb R^{k+1}$. Based on Eqn. \eqref{eq:ilr}, the inverse of ILR, i.e. recovering $\bm{u}$ from $\bm{u}^*$, takes the following form \citep{pawlowsky2007lecture}: 
\begin{eqnarray}
\bm{u}=ilr^{-1}(\bm{u}^*)=(T_2-T_1)\cdot\frac{\exp({\bm{u}^*}^T\Psi)^T}{\exp({\bm{u}^*}^T\Psi)\cdot\bm{1}_{k+1}} \label{eq:bmu}
\end{eqnarray}
Based on ILR transformation, we point out an important property of the matrix $\Psi$, which will be used in our newly defined depth:
\begin{prop} \label{prop:poly}
All $k+1$ columns of the matrix $\Psi \in \mathbb R^{k\times (k+1)}$  form a regular simplex in $\mathbb{R}^k$ centered at origin with edge length $\sqrt{2}$. 
\end{prop}
\begin{proof}
Since $\Psi\Psi^T=I_k$ and $\Psi^T\Psi=I_{k+1}-\frac{1}{k+1}\bm{1}_{k+1}^T\bm{1}_{k+1}$, it is easy to verify that $\Psi\bm{1}_{k+1}=\bm{0}$, thus, origin is the center of this polyhedron. Moreover, for each column $\Psi_{:,p}$, $p=1,2,\dots,k+1$, ${\lVert\Psi_{:,p}\rVert}^2=\frac{k}{k+1}$, where $\lVert\cdot\rVert$ denotes the Euclidean norm in $\mathbb{R}^{k}$, $\langle\Psi_{:,p},\Psi_{:,q}\rangle =-\frac{1}{k+1}$ for distinct $p,q=1,2,\dots,k+1$, then for any $i,j=1,2,\dots,k+1$ and $i\neq j$, ${\lVert \Psi_{:,i}-\Psi_{:,j}\rVert}^2={\lVert\Psi_{:,i}\rVert}^2+{\lVert\Psi_{:,j}\rVert}^2-2\langle\Psi_{:,i},\Psi_{:,j}\rangle=\frac{2k}{k+1}+\frac{2}{k+1}=2$. 
\end{proof}

Using the ILR transformation and its inverse, a bijective mapping for three spaces $\mathbb{S}_k$, $\mathcal{S}^{k+1}$ and $\mathbb{R}^k$ can be established. Whenever a point process $\bm{s}\in\mathbb{S}_k$ within a fixed time domain is given, its IET $\bm{u}\in\mathcal{S}^{k+1}$ as well as the ILR transformation $\bm{u}^*\in\mathbb{R}^k$ can be easily obtained. On the other hand, for any $\bm{u}^*\in\mathbb{R}^k$, the corresponding point process in $\mathbb{S}_k$ can be recovered by the ILR inverse. Therefore, the distribution of $\bm{u}^*$ can be derived from the distribution of $\bm{u}$, which will be illustrated in detail in the next subsection. 

\subsubsection{The ILR transformation on uniform distribution}
For simplicity, we will at first study the homogeneous Poisson process (HPP) \citep{stoyan2013stochastic}. As pointed out in \citet{qi2021dirichlet}, given the cardinality $k$ of point process, the IET of an HPP is uniformly distributed on simplex space $\mathcal{S}^{k+1}$. In this way, the joint density function of the IET vector $(u_{1}, u_{2}, ..., u_{k+1})^T$ is: 
\begin{eqnarray}
    f_{\bm{u}}(u_{1}, u_{2}, ..., u_{k+1})=k!\cdot \frac{1}{(T_2-T_1)^{k}}, \label{eq:pdf}
\end{eqnarray}
Thus, in order to obtain the density function of $\bm{u}^*$, the Jacobian matrix $J$ of $\bm{u}$ corresponding to $\bm{u}^*$ should be derived in closed form by taking derivative of $u_j$ to $u_i^*$ for $i,j=1,2,\dots,k$. Since $u_{k+1}$ is not random given $u_1,\dots,u_k$, it can be omitted when deriving the density. Then, the density function of $\bm{u}^*$ is: 
\begin{eqnarray}
f_{\bm{u}^{*}}(u_{1}^{*}, ..., u_{k}^{*}) = \frac{k!\cdot |\det(J)|}{(T_2-T_1)^k} \label{eq:ilrpdf}
\end{eqnarray}
The closed form of Eqn. \eqref{eq:ilrpdf} can be summarized in the next theorem (see proof in Appendix \ref{app:der}).

\begin{thm} \label{thm:ppdf}
Let $\bm{s}=(s_1,s_2,\dots,s_k)$ be a realization of a homogeneous Poisson point process with $k$ time events in $[T_1,T_2]$. Denote $s_0=T_1$, $s_{k+1}=T_2$ and $\bm{u}^*=(u_{1}^{*}, u_{2}^{*}, ..., u_{k}^{*})^T$ as the ILR transformation of the IET $\bm{u}=(u_1,u_2,\dots,u_{k+1})^T=(s_1-s_0,s_2-s_1,\dots,s_{k+1}-s_k)^T$. Then the probability density function of $\bm{u}^*$ conditioned on its cardinality is: 
\begin{eqnarray} 
f_{\bm{u^{*}}}(\bm{u}^*)=f_{\bm{u^{*}}}(u_{1}^{*},\dots u_{k}^{*})=\frac{c}{(\sum_{p=1}^{k+1}e^{\sum_{i=1}^{k}u_{i}^{*}\Psi_{i,p}})^{k+1}},\label{eq:density}
\end{eqnarray}
where $c$ is the normalizing constant (to make the integral of the density be 1). 
\end{thm}

To understand the density in Eqn. \eqref{eq:density}, we show two examples when $k =1$ and $2$: When $k = 1$, the density function can be simplified as $f_{\bm{u}^{*}}(u^{*})=\frac{\sqrt{2}}{(e^{u^{*}\Psi_{1,1}}+e^{-u^{*}\Psi_{1,1}})^{2}}$.  This function is shown in Fig. \ref{fig:special_case}(a).  We can see that it has a clear bell-shape, close to a standard normal distribution density. When $k=2$, the density function is  $f_{\bm{u}^{*}}(u_1^*,u_2^*)=\frac{6|\Psi_{1,1}\Psi_{2,2}-\Psi_{1,2}\Psi_{2,1}|}{(e^{u_1^*\Psi_{1,1}+u_2^*\Psi_{2,1}}+e^{u_1^*\Psi_{1,2}+u_2^*\Psi_{2,2}}+e^{u_1^*\Psi_{1,3}+u_2^*\Psi_{2,3}})^3}$, and the corresponding graphs of density as well as contours are shown in Fig. \ref{fig:special_case}(b). We also see a 2-dimensional bell shape in density.  However, we notice that the contours are not elliptical, but smoothed triangular. 

\begin{figure} [h!]
	\centering
	\subfigure[$k=1$]{
		\begin{minipage}[b]{0.47\textwidth}
			\centering
			\includegraphics[scale=0.3]{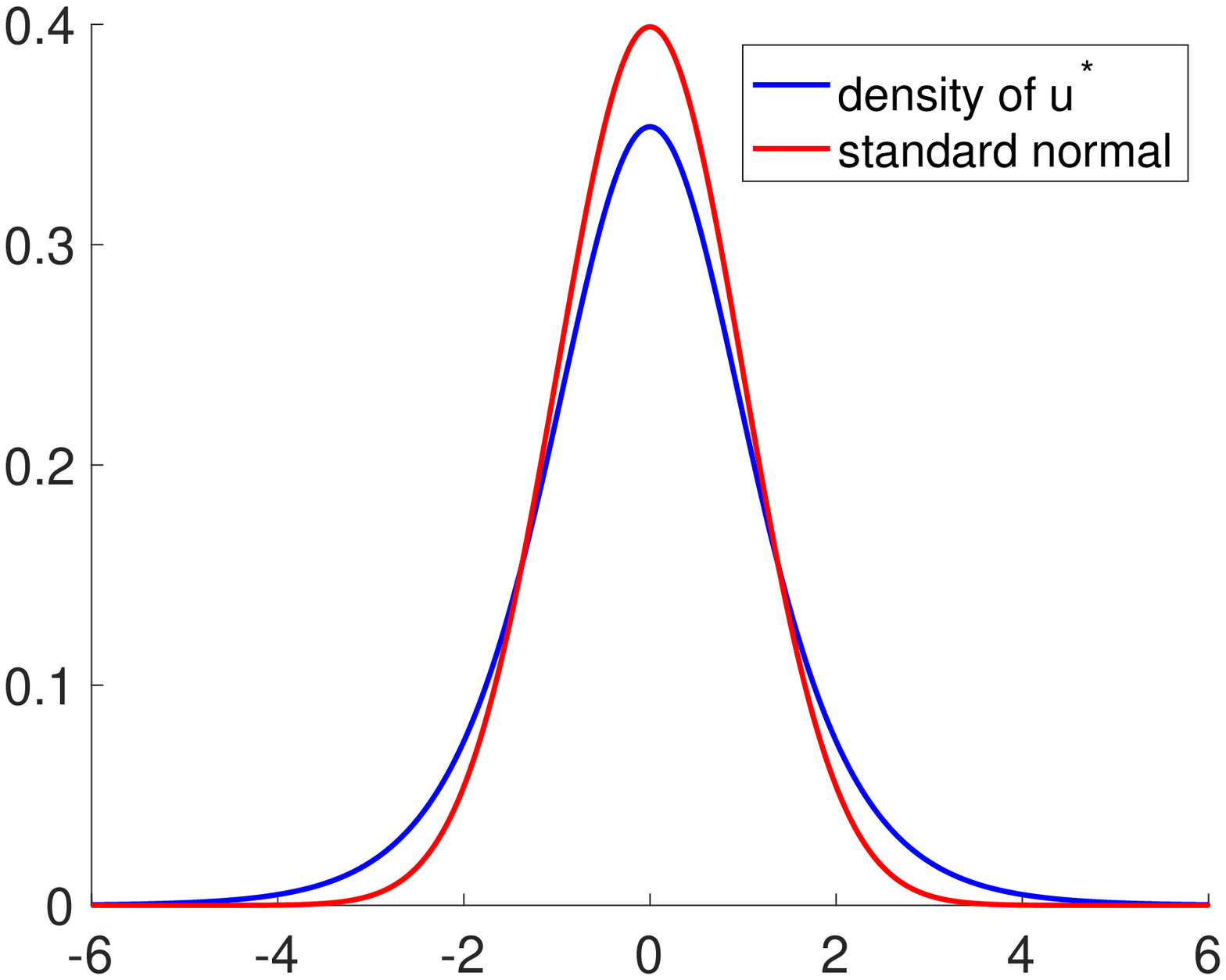}
		\end{minipage}
		\label{fig:1d_case}
	}
	\subfigure[$k=2$]{
		\begin{minipage}[b]{0.47\textwidth}
			\centering
			\includegraphics[scale=0.3]{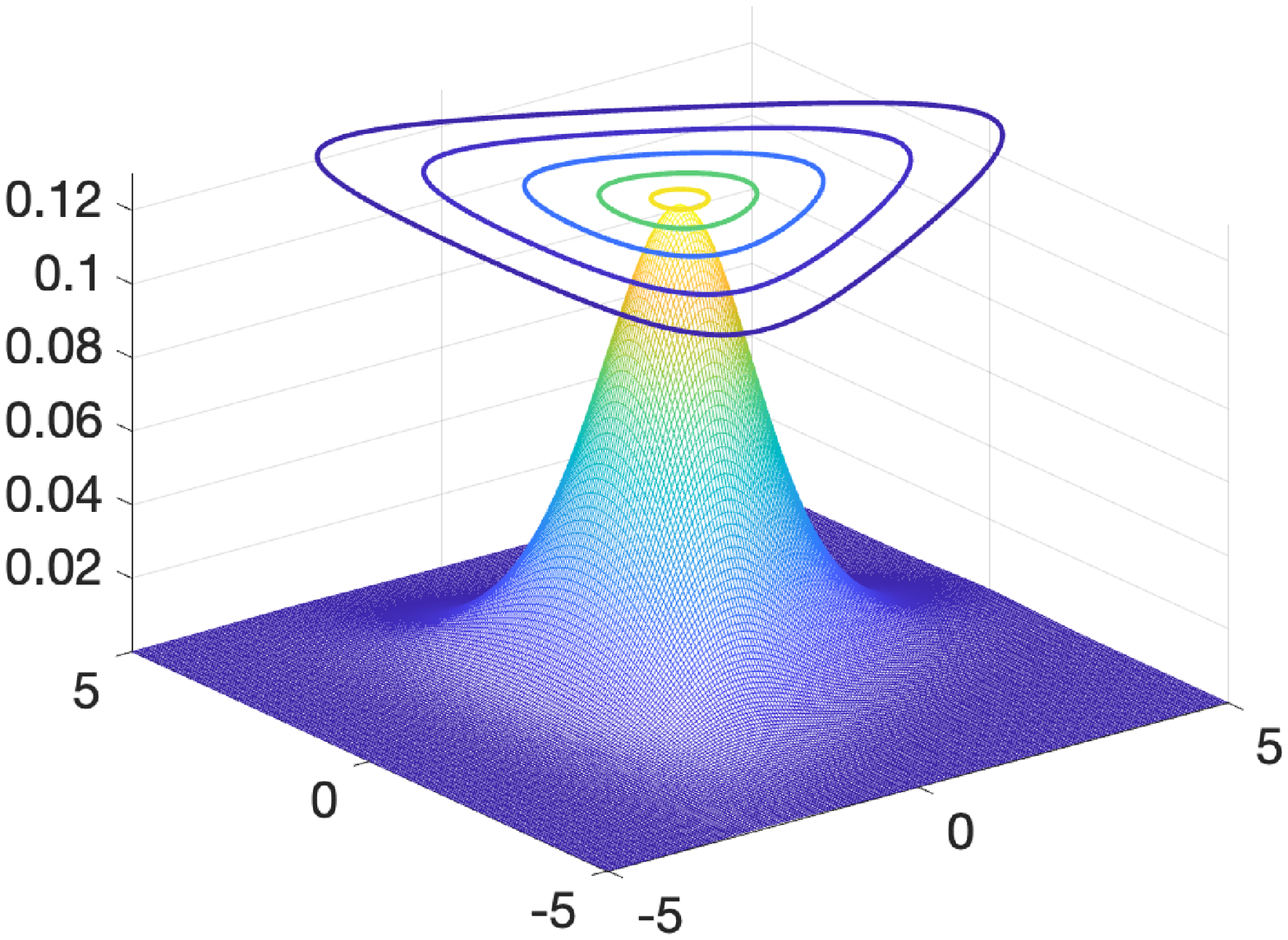}
		\end{minipage}
		\label{fig:2d_case}
	}
	\caption{Two special cases of the density in Eqn. (\ref{eq:density}). (a) The blue line is the density function when $k=1$, and the red line is the standard normal density function (b) Mesh plot of density when $k=2$, where the curves at the top are contours of the density function.}
	\label{fig:special_case}
\end{figure}

Based on the observations in these special cases, we have the following three important properties of the density in Eqn. \eqref{eq:density}:
\begin{enumerate}
\item The density is log-concave and uni-modal, i.e., it can be approximated using a normal distribution.
\item The approximated normal density has a standard form, i.e., the mean is 0 and the covariance is $I_k$. 
\item The density has a symmetry with respect to the origin in a simplex.  We will use this symmetry in the proposed depth function.   
\end{enumerate}
These properties are formally given in Propositions \ref{prop:logconcave} - \ref{prop:ortho} as follows:

\begin{prop} \label{prop:logconcativity}
The density function given in Eqn. \eqref{eq:density} is log-concave and uni-modal, and the global maximum point is the origin in the Euclidean space $\mathbb{R}^k$. 
\label{prop:logconcave}
\end{prop}

The proof of Proposition \ref{prop:logconcativity} is given in Appendix \ref{app:logcon}. Based on this result, the density function in Eqn. \eqref{eq:density} is centered at origin and the value decreases from origin in $\mathbb{R}^k$, which makes the density satisfy the second and third depth properties mentioned in Section \ref{sec:depdef}. Thus, if $\bm{u}^*$ is mapped back to $\mathbb{S}_k$, the density function in Eqn. \eqref{eq:density} provides an ideal candidate for point process depth conditioned on its cardinality, which will be formally defined in Section \ref{sec:depthdefinition}. 

The original density of $\bm{u}^*$ in Eqn. \eqref{eq:density} is in a complicated form and its mathematical properties are difficult to study, especially when $k$ is large.  Based on Proposition \ref{prop:logconcativity}, one can approximate the density using a normal distribution.  This approximation may simplify calculation on the given density.  For example, in Bayesian statistics, a method called Laplacian approximation estimates posterior distribution by using such approximated normal distribution.  
By looking at the contours in Figure \ref{fig:special_case}(b), we can see the contour curves near the center do have a circular shape.  This observation is formally given in the following proposition (see proof in Appendix \ref{app:contour}).
\begin{prop} \label{prop:sphere}
For the density function in Eqn. \eqref{eq:density}, if the Euclidean norm of $\bm{u}^*$ is small, the shape of contour is close to a hyper-sphere centered at origin. 
\end{prop}

In general, because of log-concavity, one may propose to use a normal distribution to approximate the density function in Eqn. \eqref{eq:density} for the purposes of simplification or efficiency. This normal approximation adopts the idea of the Laplacian approximation, and it will actually lead to an alternative approach of depth definition later in this paper.  In our framework, it is interesting to find that the approximated normal distribution has a standard form (proof is given in Appendix \ref{app:gau}): 
\begin{prop} \label{prop:lap}
The normal approximation of the density function in Eqn. \eqref{eq:density} is the $k$-dimensional standard multivariate normal distribution $N(0, I_{k})$.  
\end{prop}

Moreover, according to Fig. \ref{fig:special_case}(b), the density contours appear symmetric about the origin for orthogonal transformations on three corners of the smoothed triangular shape.  In general, we will show this symmetric property with respect to $(k+1)!$ orthogonal transformations in $k$-dimensional space.  We will at first define those transformations: Assume $r(\cdot)$ is any permutation operation on  $1,2,\dots,k+1$, and denote $\Psi_r$ as the matrix after permuting the column of $\Psi$ in the order $r(1),r(2),\dots,r(k+1)$.  That is, using column-wise representation, if $\Psi = (\Psi_{:,1}, \cdots, \Psi_{:,k+1})$, then $\Psi_r = (\Psi_{:,r(1)}, \cdots, \Psi_{:,r(k+1)})$.
Let 
\begin{equation}
A_r=\Psi\Psi_r^T \in \mathbb R^{k\times k}. 
\label{eq:ortho}
\end{equation}
It is easy to verify that $A_r$ is orthogonal and $A_r^T\Psi_{:,i}=\Psi_{:,r(i)}$, $i=1,2,\dots,k+1$. 

Using Eqn. \eqref{eq:density}, we have
\[
f_{\bm{u^{*}}}(\bm{u}^*)=\frac{c}{(\sum_{p=1}^{k+1}e^{\sum_{i=1}^{k}u_{i}^{*}\Psi_{i,p}})^{k+1}}
= \frac{c}{(\sum_{p=1}^{k+1}e^{({\bm{u}^{*}})^T\Psi_{:,p}})^{k+1}}
\]
Therefore,
\[
f_{\bm{u^{*}}}(A_r \bm{u}^*)
= \frac{c}{(\sum_{p=1}^{k+1}e^{{(A_r\bm{u}^{*}})^T\Psi_{:,p}})^{k+1}}
= \frac{c}{(\sum_{p=1}^{k+1}e^{{(\bm{u}^{*}})^T\Psi_{:,r(p)}})^{k+1}} = f_{\bm{u^{*}}}(\bm{u}^*).
\]
Therefore, the above analysis has shown the following proposition on orthogonal symmetry: 
\begin{prop} \label{prop:ortho} For any orthogonal transformation $A_r$ defined via a permutation $r(\cdot)$ in Eqn. \eqref{eq:ortho},  
 $f_{\bm{u^{*}}}(A_r\bm{u}^*)=f_{\bm{u^{*}}}(\bm{u}^*)$ for any column vector $\bm{u}^*\in\mathbb{R}^k$. 
\end{prop}

\begin{remark}
In Proposition \ref{prop:ortho}, if $k=2$, then there are $(2+1)! = 6$ different orthogonal matrices in total. Three of them are rotation matrices and the corresponding rotation angles are 0, $\frac{2\pi}{3}$, and $\frac{4\pi}{3}$, respectively.  The other three are reflection matrices.
\end{remark}

Combining Proposition $\ref{prop:poly}$ and Proposition $\ref{prop:ortho}$, the key part of density function of $\bm{u}^*$ can be considered as summing up the exponential of the inner product of $\bm{u}^*$ with the vertices of a $k+1$ regular simplex centered at origin in $\mathbb{R}^k$. Therefore, the orthogonal transformation of $\bm{u}^*$ introduced in Proposition $\ref{prop:ortho}$ can be viewed as orthogonal transformation of a regular simplex corresponding to the origin. In this way, the orthogonal symmetry of the density of $\bm{u}^*$ can be easily interpreted and the center of this density is the same as the center of the regular simplex.  

\subsection{ILR depth for homogeneous Poisson process}
\label{sec:ilrhpp}
In this subsection, we will formally define the depth for HPP, conditioned on its cardinality.  The definition is based on the density of the ILR transformed inter-event times in Eqn. \eqref{eq:density} and we call the new method the ILR depth. 

\subsubsection{Definition} \label{sec:depthdefinition}
The density function of the ILR transformed IET in an HPP is given in Eqn. \eqref{eq:density}. Based on the fact that the ILR transformation is isometric isomorphism between the simplex $\mathcal{S}^{k+1}$ and Euclidean space $\mathbb{R}^k$ \citep{pawlowsky2007lecture}, this density can capture temporal pattern in the original point process.  It is shown in Proposition \ref{prop:logconcave} that this density is log-concave and uni-modal, and therefore it provides an ideal form to define the depth.  

As the ILR transformation has a closed-form inverse, we can express the depth in terms of the original point process time events. Based on Eqn. \eqref{eq:density}, we have  
\begin{eqnarray*}
& & \Big(\sum_{p=1}^{k+1}e^{\sum_{i=1}^{k}u_i^*\Psi_{i,p}}\Big)^{k+1} = \Big(\sum_{p=1}^{k+1}e^{\log\frac{\bm{u}^T}{g(\bm{u})}\Psi^T\Psi_{:,p}}\Big)^{k+1} \\
&=& \Big(\sum_{p=1}^{k+1}e^{\log\frac{u_p}{g(\bm{u})}}\Big)^{k+1} 
= \Big(\frac{T_2-T_1}{g(\bm{u})}\Big)^{k+1} 
= \frac{(T_2-T_1)^{k+1}}{\prod_{i=1}^{k+1}(s_i-s_{i-1})}
\end{eqnarray*}
In this way, the formal depth definition conditioned on cardinality can be derived in terms of a point process $\bm{s}$, its IET $\bm{u}$, or the ILR transformation $\bm{u}^*$ of IET. In other words, the depth can be defined on three equivalent spaces, e.g. $\mathbb{S}_k$, $\mathcal{S}^{k+1}$ and $\mathbb{R}^k$. 
Similar to the commonly used Mahalanobis depth in \citet{liu1993quality}, the density in Eqn. \eqref{eq:density} will not be directly used to define depth because the value of density will decrease sharply when the data point deviates from the global maximum point. This phenomenon will become more evident when the dimension $k$ is large. Instead, similar to the definition of the Mahalanobis depth, a logarithm-based increasing function $f(\cdot)=\frac{1}{1-\log(\cdot)}$ can be used to the kernel part of density to alleviate the decreasing rate. Therefore, the ILR depth can be formally defined in the following form:  
\begin{definition} \label{def:formal}
Let $\bm{s}=(s_1,s_2,\dots,s_k)\in\mathbb{S}_k$ be a realization of an HPP in the time domain $[T_1,T_2]$ with $T_1< s_1< s_2 < \cdots < s_k < T_2$, denote $s_0=T_1$, $s_{k+1}=T_2$, $\bm{u}=(u_1,u_2,\dots,u_{k+1})^T=(s_1-s_0,s_2-s_1,\dots,s_{k+1}-s_k)^T$ as the IET and $\bm{u}^*=(u_1^*,u_2^*,\dots,u_k^*)^T$ as the ILR transform of $\bm{u}$. Then, the ILR depth of $\bm{s}$ conditioned on $|\bm{s}|=k$ is defined as: 
\begin{eqnarray}
    D_{c}(\bm{s};P_{S||S|=k}) &=& \frac{1}{1-\log\Big(\frac{c}{(\sum_{p=1}^{k+1}e^{\sum_{i=1}^{k}u_{i}^{*}\Psi_{i,p}})^{k+1}}\Big)} \label{eq:ilrdepth_rk} \\
    &=& \frac{1}{1-\log\Big(\frac{c}{(T_2-T_1)^{k+1}}\prod_{i=1}^{k+1}(s_i-s_{i-1})\Big)}  \label{eq:ilrdepth} 
\end{eqnarray}
where $c$ is a positive constant within interval $\big(0,e(k+1)^{k+1}\big)$in order to make $ D_{c}(\bm{s};P_{S||S|=k})$ positive. If the maximum value of $D_{c}(\bm{s};P_{S||S|=k})$ is constrained to be $1$, then $c=(k+1)^{k+1}$. If $\bm{s}\in\mathbb{B}_k$, then $D_{c}(\bm{s};P_{S||S|=k})=0$. 
\end{definition}
\begin{remark}
In Theorem \ref{thm:ppdf}, there is an important assumption that the point process $\bm{s}$ belongs to the interior of $\mathbb{S}_k$, i.e., $\bm{s} \notin \mathbb{B}_k$.  Otherwise, the ILR transformation cannot be properly conducted. However, in Definition \ref{def:formal}, if $\bm{s}\in\mathbb{B}_k$, the depth value is defined to be $0$ since in this case the ILR depth is continuous at boundary set $\mathbb{B}_k$ based on Eqn. \eqref{eq:ilrdepth}. 
\end{remark}

In the remaining part of this paper, the constant $c$ in Definition \ref{def:formal} is fixed as $(k+1)^{k+1}$ to normalize the maximum value of the ILR depth being $1$.

\subsubsection{Illustrations} \label{sec:illustration_ILR}


Assume the time domain is $[0,2]$ and the intensity function is a constant value $1$.  Then $1000$ HPP realizations with cardinality being 2 are generated. 
For each realization, the IET is a three dimensional vector in the simplex $\mathcal{S}^3$.  Thus, a 2-dimensional ternary plot \citep{pawlowsky2007lecture} can be conducted together with the ILR depth value as the contour according to Eqn. \eqref{eq:ilrdepth}.  The result is shown in Fig. \ref{fig:ILR_contour}(a). The depth values and their contours in the transformed Euclidean space based on Eqn. \eqref{eq:ilrdepth_rk} are also shown in Fig. \ref{fig:ILR_contour}(b). 
\begin{figure} [h!]
	\centering
	\subfigure[data points in Simplex]{
		\begin{minipage}[b]{0.47\textwidth}
			\centering
			\includegraphics[scale=0.3]{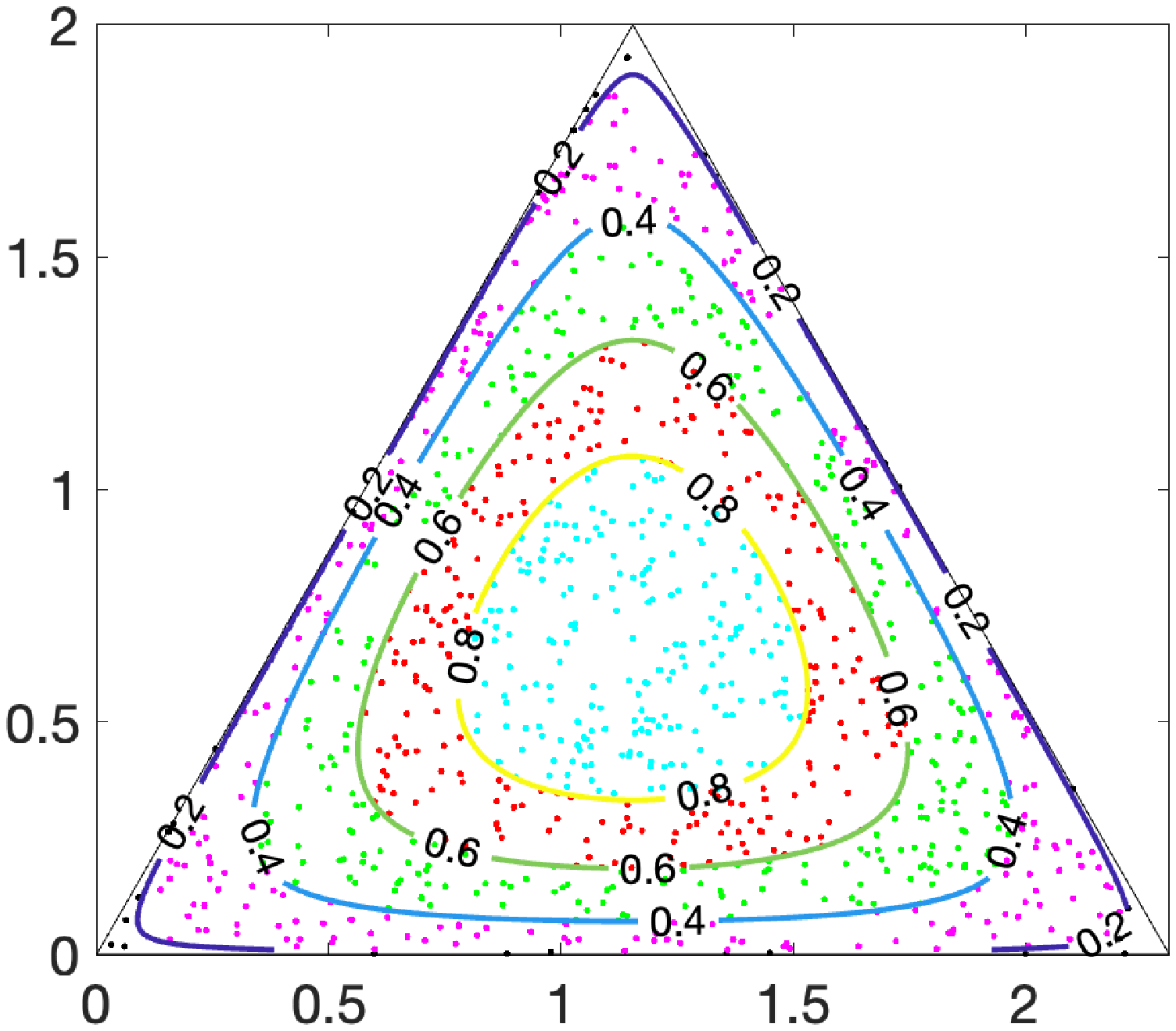}
		\end{minipage}
		\label{fig:ILR_ternary}
	}
	\subfigure[data points after ILR]{
		\begin{minipage}[b]{0.47\textwidth}
			\centering
			\includegraphics[scale=0.3]{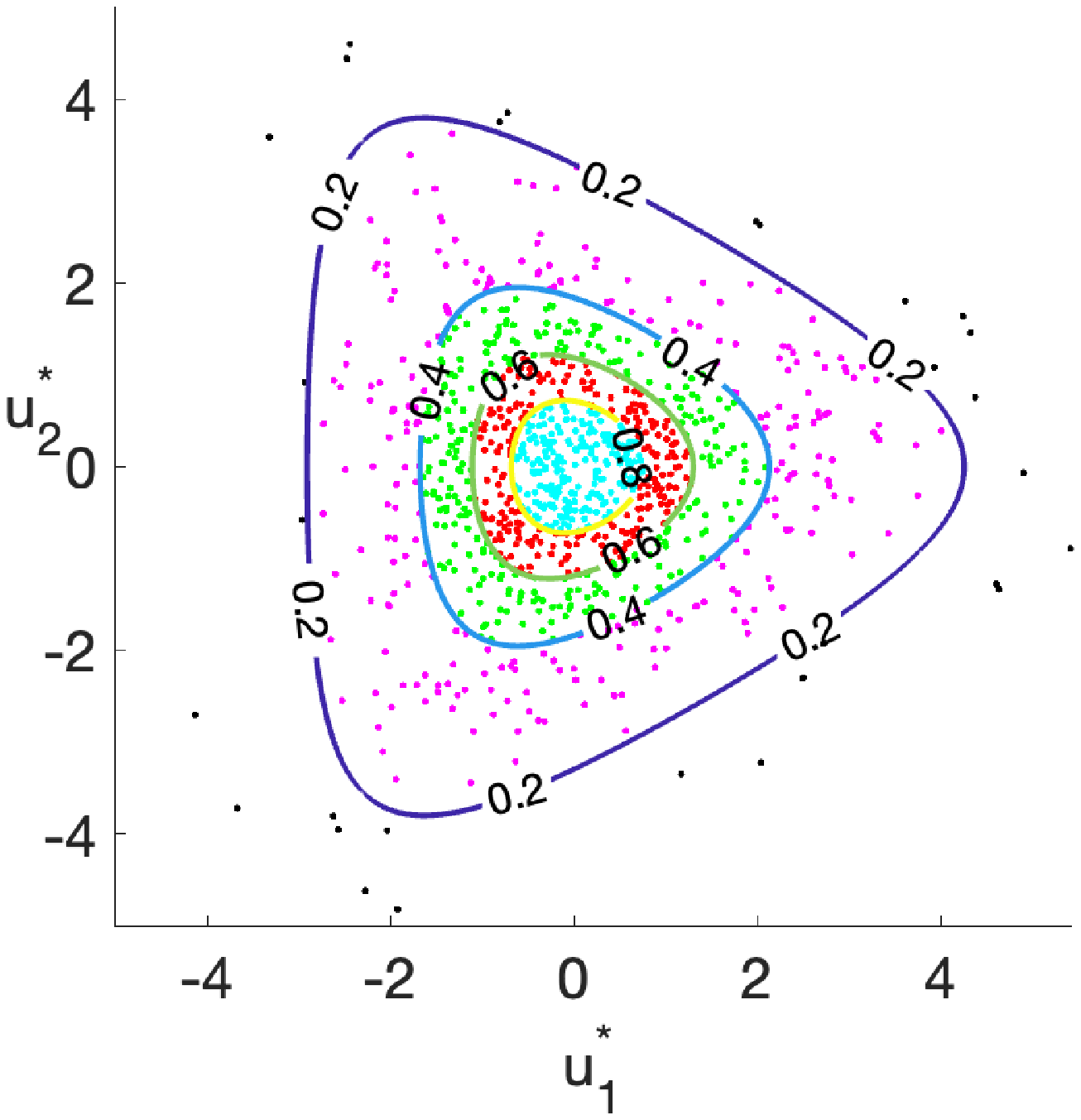}
		\end{minipage}
		\label{fig:ILR_ILR}
	}
	\caption{ILR depth result in HPP. (a) Data points in simplex shown by ternary plot, where the points are colored with respect to ranges of the depth values.  The solid lines indicate the depth contours with specific values.  (b) Same as (a) except for data points in the ILR-transformed space $\mathbb R^2$.}
	\label{fig:ILR_contour}
\end{figure}
From Fig. \ref{fig:ILR_contour}(a), the contour value decreases from the center to edges, and the depth value approaches $0$ if the IET vector approaches the boundary. In addition, based on Fig. \ref{fig:ILR_contour}(b), the shape of the ILR transformed data look similar to a regular triangle centered at origin. 
The inner contours of the ILR depth in $\mathbb{R}^2$ are closer to circular shapes, whereas the outside ones are similar to regular, smoothed triangles, which coincide with the shape of the contour of density in Eqn. \eqref{eq:density} except the scale. Based on these observations, mathematical properties of this new depth will be examined in Section \ref{sec:math_prop}. 

\begin{figure} [h!]
	\centering
	\subfigure[Dirichlet depth]{
		\begin{minipage}[b]{0.47\textwidth}
			\centering
			\includegraphics[scale=0.3]{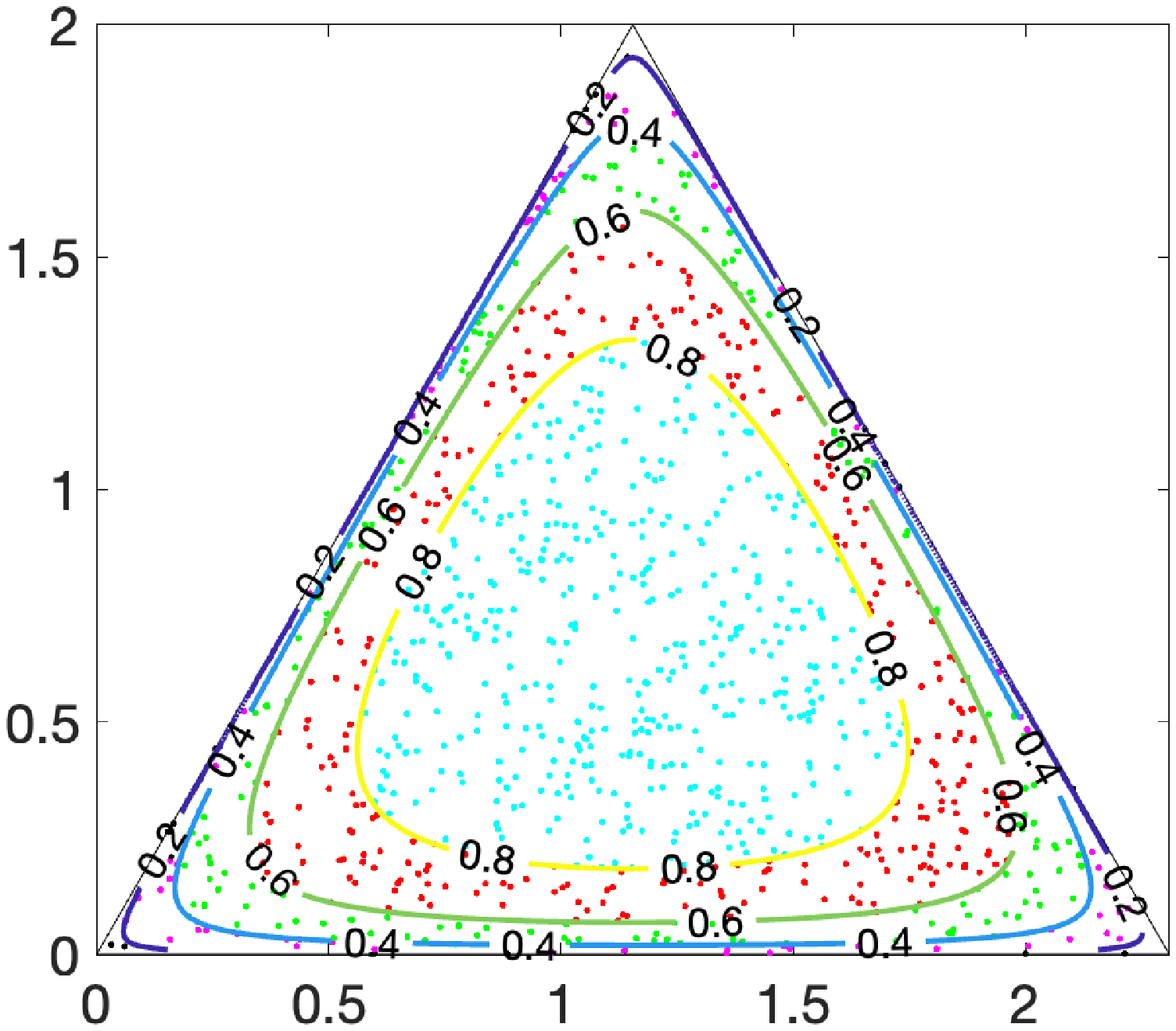}
		\end{minipage}
		\label{fig:dirichlet}
	}
	\subfigure[generalized Mahalanobis depth]{
		\begin{minipage}[b]{0.47\textwidth}
			\centering
			\includegraphics[scale=0.3]{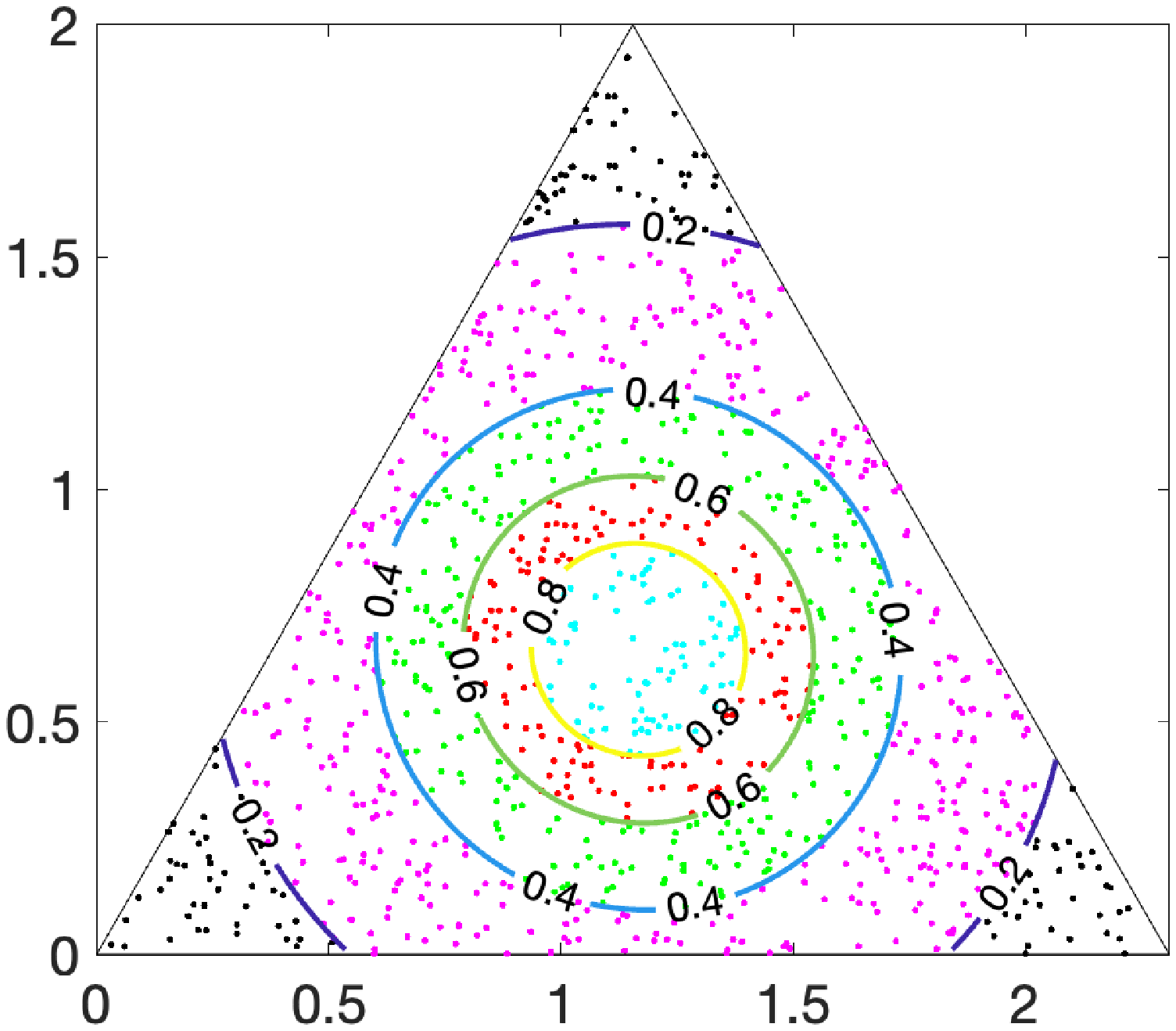}
		\end{minipage}
		\label{fig:mahalanobis}
	}
	\caption{Comparison with the Dirichlet depth and generalized Mahalanobis depth. (a) Data points in simplex shown by ternary plot, where the points are colored with respect to ranges of the Dirichlet depth values, the solid lines indicate the depth contours with specific values. (b) Same as (a) except for the generalized Mahalanobis depth.}
	\label{fig:two_comparison}
\end{figure}
To compare with the proposed ILR depth, we also show result using the Dirichlet depth \citep{qi2021dirichlet} and the generalized Mahalanobis depth \citep{liu2017generalized} on the same dataset in Fig. \ref{fig:two_comparison}.
Comparing Fig. \ref{fig:ILR_contour}(a) and \ref{fig:two_comparison}(a), we can see that the decreasing rate of the Dirichlet depth near the triangular center is much slower than that in the ILR depth. For the Dirichlet depth, many realizations in the middle part have large depth value, whereas the depth values decrease sharply when the data points approach the boundary. On the other hand, the decreasing rate of the ILR depth changes more slowly when the data points vary from center to boundary. Finally, the generalized Mahalanobis depth is not appropriate for HPP since it gives positive value to process at boundary and does not fit the triangular domain. 

In addition to the illustration of conditional depth, we will use another simulation to show the rank of HPP realizations using the depth in Definition \ref{def:wholedef}, where we use the ILR depth as the conditional depth given cardinality.  Suppose the intensity function is $\lambda=1$ on the time domain $[0,5]$. From the definition of HPP, the expected number of time events is $5$. This makes the normalized one dimensional depth $w(|\bm{s}|;P_{|S|})$ in Definition \ref{def:wholedef} obtains the maximum value when the cardinality of point process is $5$.  This result is clearly shown in Fig. \ref{fig:top10}(a).  

Based on the overall depth in Definition \ref{def:wholedef}, top-ranked processes are expect to have 5 evenly distributed events.  
In Fig. \ref{fig:top10}(b), we show the top 10 processes with largest overall depth value when $r=1$.  In this case, the one-dimensional depth plays an important role and all these processes have 5 events and their distirbutions are nearly uniform on [0, 5].  In contrast, we also show the top 10 processes with largest overall depth value when $r=0.1$ in Fig. \ref{fig:top10}(c).  With a smaller weight coefficeint, the one-dimensional depth contributes less to the overall depth value and we can see several top-10 processes have 3, 4, or 5 events.  Nevertheless, the event distribution in each top process is still homogeneous on [0, 5].

\begin{figure} [h!]
	\centering
	\subfigure[one dimensional depth]{
		\begin{minipage}[b]{0.3\textwidth}
			\centering
			\includegraphics[scale=0.2]{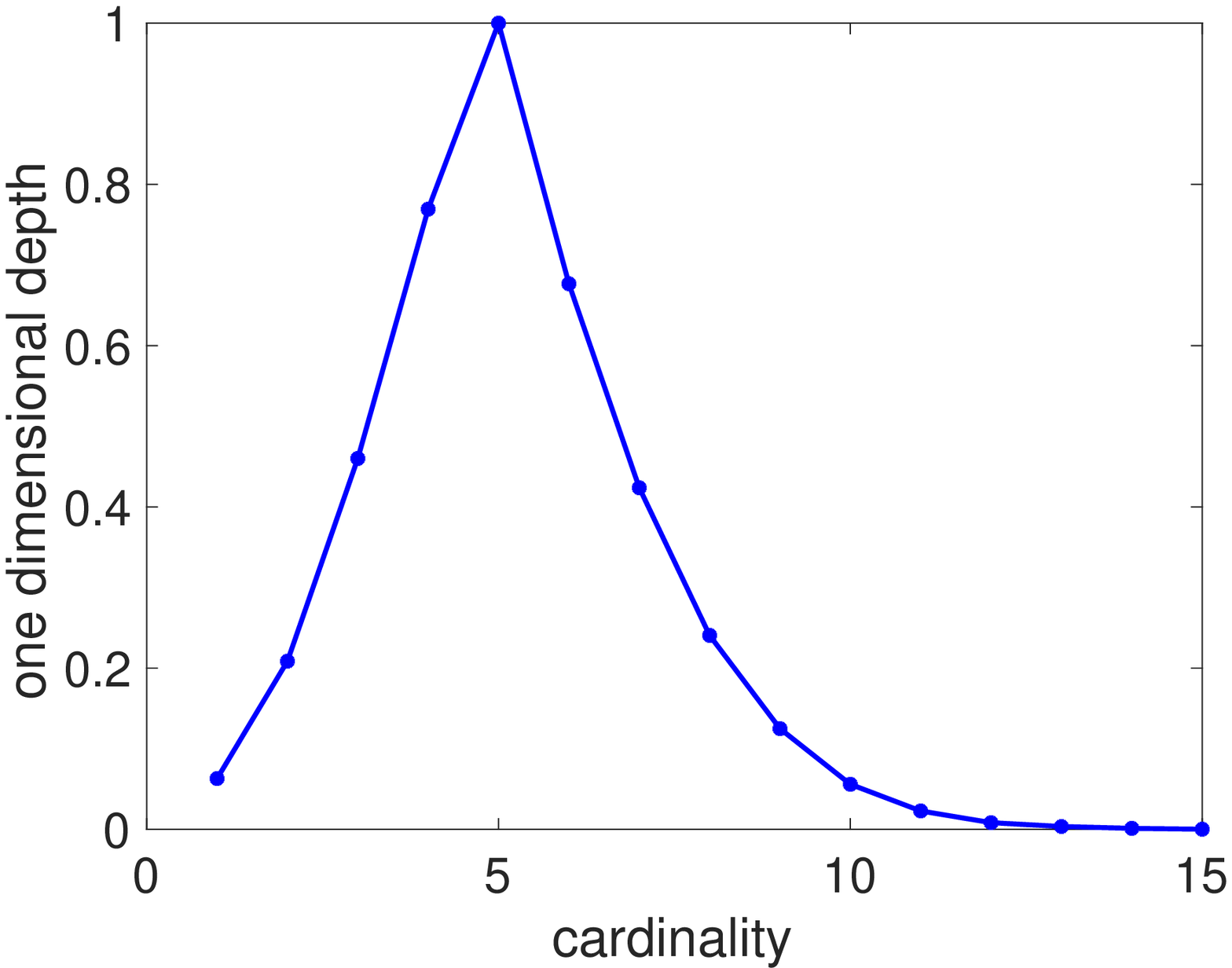}
		\end{minipage}
	}
	\subfigure[depth ranking, $r=1$]{
		\begin{minipage}[b]{0.3\textwidth}
			\centering
			\includegraphics[scale=0.2]{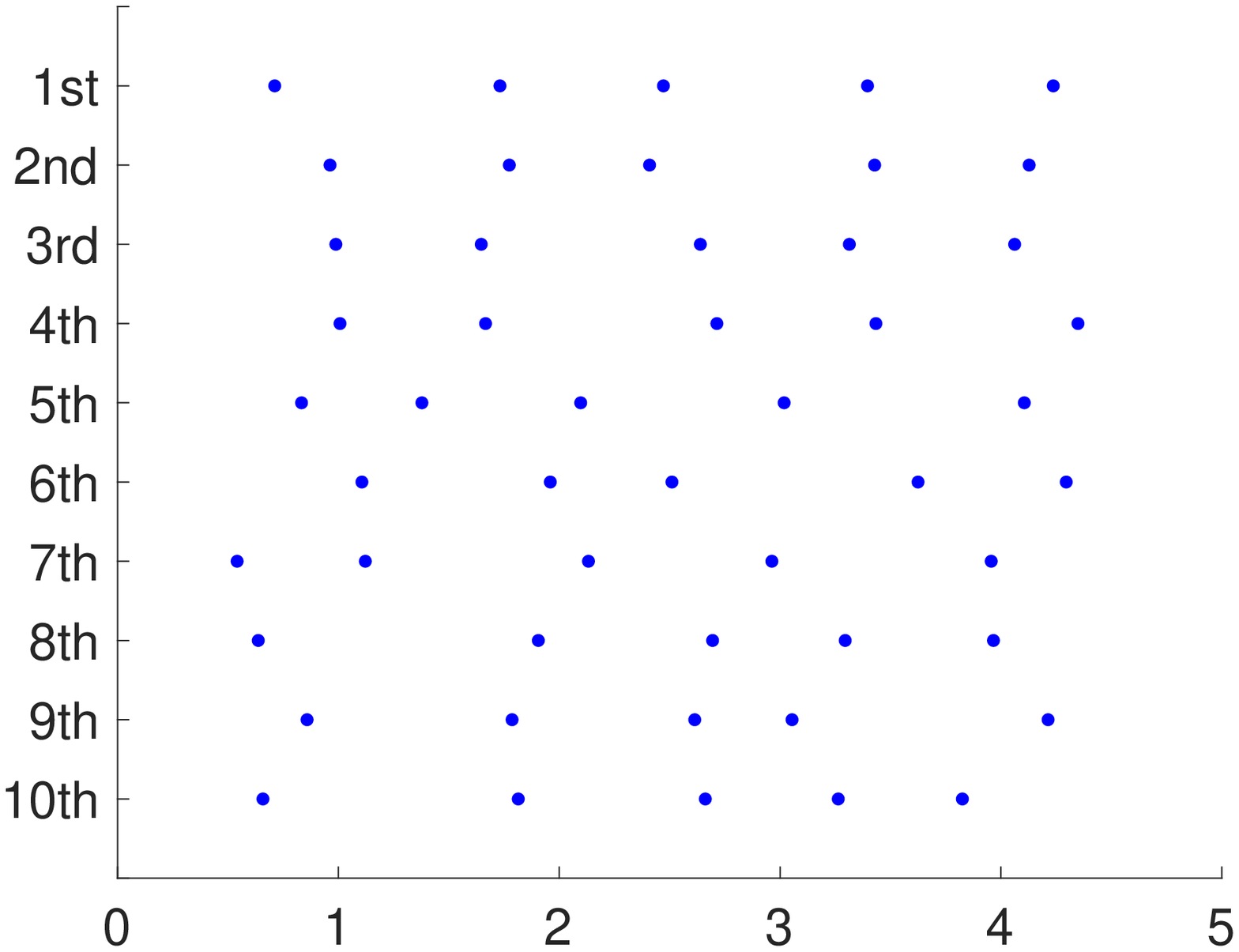}
		\end{minipage}
	}
	\subfigure[depth ranking, $r=0.1$]{
    		\begin{minipage}[b]{0.3\textwidth}
			\centering
   		 	\includegraphics[scale=0.2]{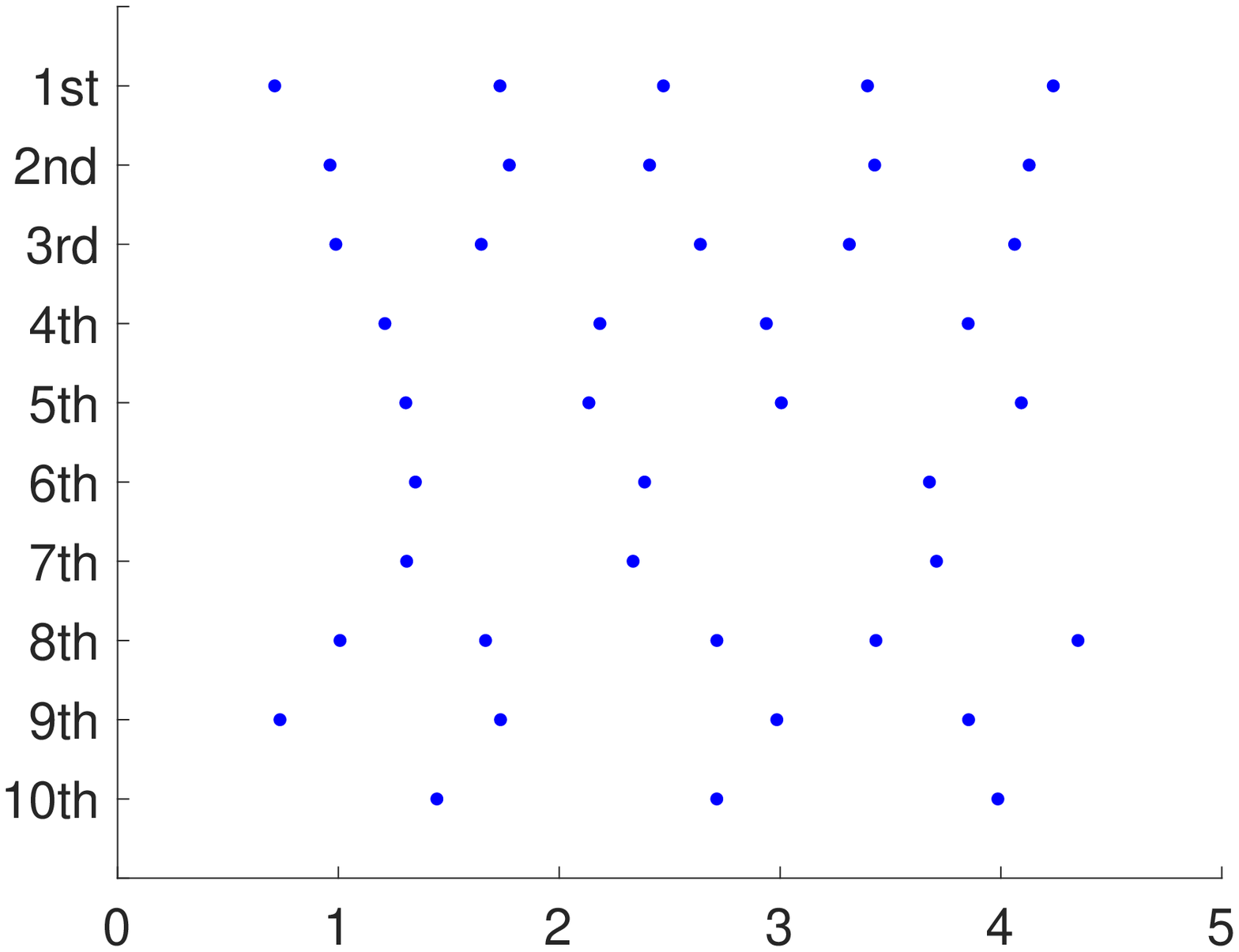}
    		\end{minipage}
    	}
	\caption{One dimensional depth and point processes with top $10$ overall depths. (a) One dimensional depth for different cardinalities in the simulated sample. (b) Each row is a realization of simulated HPP with the ranking of depth value shown in vertical axis. The depth is computed by Definition \ref{def:wholedef}, where the ILR depth is the conditional depth and the hyper-parameter $r=1$. (c) Same as (b) except that $r=0.1$. }
	\label{fig:top10}
\end{figure}

\subsubsection{Mathematical properties} \label{sec:math_prop}
In Section \ref{sec:depdef}, we have listed four basic mathematical properties that a conditional depth is expected to satisfy. All properties are easy to verify for the newly defined ILR depth except for a notion of center with respect to a specific symmetry in Property 2.  In order to make the ILR depth a proper conditional depth, a center needs to be defined corresponding to an appropriate symmetry. Previous studies \citep{zuo2000general} have summarized different types of depth symmetries for multivariate data, e.g. A-symmetry, C-symmetry and H-symmetry. However, none of them can be directly applied to the ILR depth given in Definition \ref{def:formal}. 
 In the definition of the Dirichlet depth, the center is taken as the the conditional mean of IET $\bm{u}=(\frac{T_2-T_1}{k+1},\frac{T_2-T_1}{k+1},\dots,\frac{T_2-T_1}{k+1})^T$ since the depth value reaches the maximum value at this point \citep{qi2021dirichlet}.  This ``center", however, lacks geometric interpretation with respect to a symmetry in the simplex space $\mathcal{S}^{k+1}$.
 
By Definition \ref{def:formal}, the ILR depth is defined in  the unconstrained Euclidean space $\mathbb{R}^k$ transformed from the simplicial domain $\mathcal S^{k+1}$.  In this case, we can examine possible symmetries with respect to probability distribution of the transformed IET. According to Proposition \ref{prop:ortho}, the origin is the center of the ILR depth contours in $\mathbb{R}^k$ with respect to $(k+1)!$ orthogonal transformations defined using all vertices in a regular simplex in Eqn. \eqref{eq:ortho}. Given the bijective mappings between $\mathbb{S}_k$, $\mathcal{S}^{k+1}$, and $\mathbb{R}^k$, we can map the origin back to $\mathcal{S}^{k+1}$ and $\mathbb{S}_k$, and call the two corresponding points as ``simplicial centers'' with respect to the ``orthogonal symmetry''. This result is formally given in the following proposition. 

\begin{prop} \label{prop:center}
For any HPP with $k$ events in a time domain $[T_1,T_2]$, the simplicial center with respect to the orthogonal symmetry in $\mathbb{S}_k$ is $$\bm{s}=(T_1+\frac{T_2-T_1}{k+1},T_1+\frac{2(T_2-T_1)}{k+1},\dots,T_1+\frac{k(T_2-T_1)}{k+1}),$$ and in $\mathcal{S}^{k+1}$ is $$\bm{u}=(\frac{T_2-T_1}{k+1},\frac{T_2-T_1}{k+1},\dots,\frac{T_2-T_1}{k+1})^T.$$
\end{prop}

Once the notion of simplicial center is given, the mathematical properties of the ILR depth can be easily verified.  This is given in the following proposition, where the detailed proof is given in Appendix \ref{app:mat}. 

\begin{prop} \label{prop:ilrdepmat}
All of the four mathematical properties in Section \ref{sec:depdef} hold in the Euclidean space $\mathbb{R}^k$ for the ILR depth in Definition \ref{def:formal}.
\end{prop}

\begin{remark}
It is easy to verify that Properties \ref{m1}, \ref{m2} and \ref{m4} in Section \ref{sec:depdef} also hold in $\mathbb{S}_k$, when mapping back from $\mathbb{R}^k$. However, Property \ref{m3} cannot hold as there is no linear structure in $\mathbb{S}_k$. In the remaining part of this paper, we will only discuss Property \ref{m3} in the Euclidean space $\mathbb{R}^k$.
\end{remark}

\subsubsection{A simplified ILR depth} \label{sec:opt_definition}
Based on Proposition \ref{prop:lap}, the density function of the ILR transformed data in Eqn. \eqref{eq:density} can be approximated by the standard multivariate Gaussian density, which can lead to another type of depth definition by Gaussian density function. This simplified version of ILR depth is actually given in the form of a Mahalanobis depth \citep{liu1993quality} as follows: 

\begin{definition} \label{def:opt}
Let $\bm{s}=(s_1,s_2,\dots,s_k)\in\mathbb{S}_k$ be a realization of an HPP in the time domain $[T_1,T_2]$ with $T_1< s_1< s_2 < \cdots < s_k < T_2$. Denote $s_0=T_1$, $s_{k+1}=T_2$, $g_s=(\prod_{i=1}^{k+1}(s_i-s_{i-1}))^{\frac{1}{k+1}}$, $\bm{u}=(u_1,u_2,\dots,u_{k+1})^T=(s_1-s_0,s_2-s_1,\dots,s_{k+1}-s_k)^T$ as the IET and $\bm{u}^*=(u_1^*,u_2^*,\dots,u_k^*)^T$ as the ILR transformation of $\bm{u}$. Then, a simplified version of the ILR depth of $\bm{s}$ conditioned on $|\bm{s}|=k$ is defined as:
\begin{eqnarray*}
    D_{c}(\bm{s};P_{S||S|=k})&=& \frac{1}{1+\frac{1}{2}\lVert\bm{u}^*\rVert^2} \\
    &=&\frac{1}{1+\frac{1}{2}\sum_{i=1}^{k+1}(\log\frac{s_i-s_{i-1}}{g_s})^2}
\end{eqnarray*}
If $\bm{s}\in\mathbb{B}_k$, i.e., at least two of $s_0, s_1, \dots,s_{k+1}$ are identical, then $D_{c}(\bm{s};P_{S||S|=k})=0$. 
\end{definition}

Definition \ref{def:opt} is a simplified conditional depth definition for homogeneous Poisson process.  It cannot capture the real data pattern as accurately as the ILR depth because of the approximation. However, the contours of this new depth are all hyper-spheres centered at origin.
All classical symmetries can be satisfied in this simplified depth and there is no need to introduce a simplicial center as in the ILR depth. This simplified depth represents the transformed IET using conventional Mahalanobis depth \citep{liu1993quality} and satisfies all necessary important properties.  This result is given in the following proposition (see proof in Appendix \ref{app:gaumat}).
\begin{prop}
All of the four mathematical properties in Section \ref{sec:depdef} hold for the simplified ILR depth in Definition \ref{def:opt}. 
\end{prop}

We also illustrate the simplified ILR depth using simulations.  Using the same simulated realizations of HPP in Section \ref{sec:illustration_ILR}, the result is shown in Fig. \ref{fig:gaussian_depth_contour}. 
\begin{figure} [h!]
	\centering
	\subfigure[data points in simplex]{
		\begin{minipage}[b]{0.47\textwidth}
			\centering
			\includegraphics[scale=0.3]{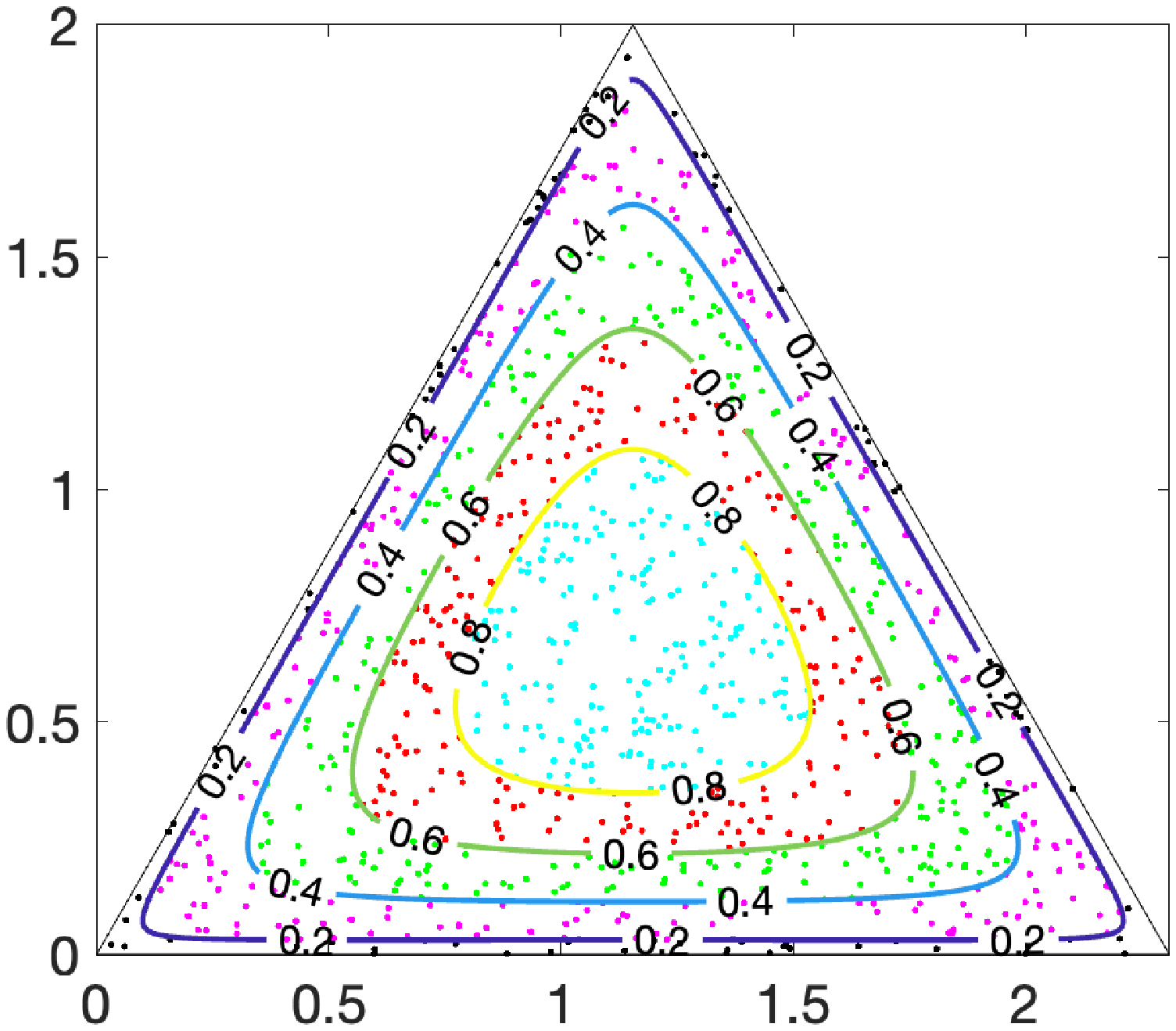}
		\end{minipage}
		\label{fig:gaussian_ternary}
	}
	\subfigure[data points after ILR]{
		\begin{minipage}[b]{0.47\textwidth}
			\centering
			\includegraphics[scale=0.3]{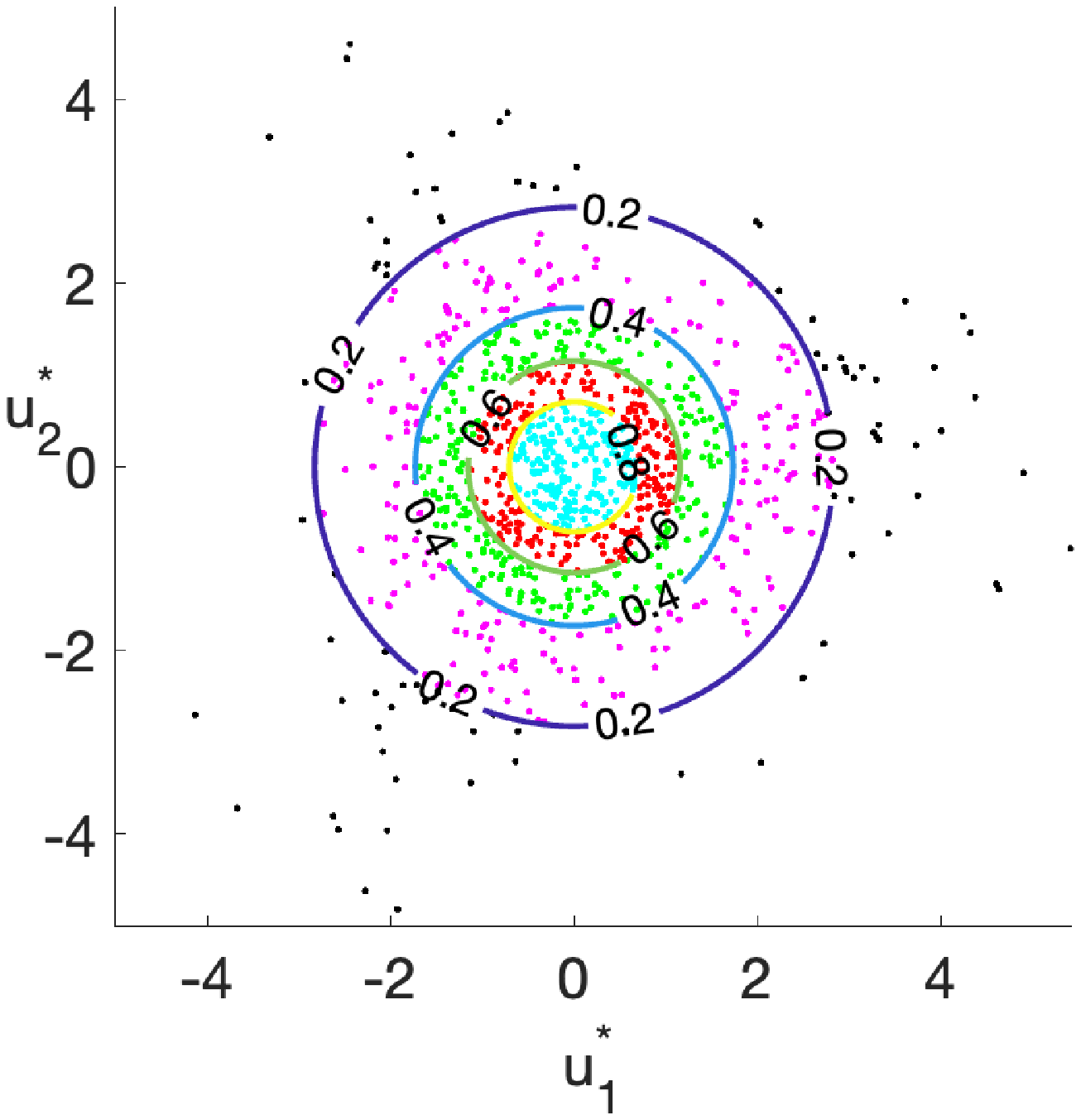}
		\end{minipage}
		\label{fig:gaussian_ILR}
	}
	\caption{Simplified ILR depth result in HPP. (a) Data points in simplex shown by ternary plot, where the points are colored with respect to ranges of the simplified ILR depth values.  The solid lines indicate the depth contours with specific values.  (b) Same as (a) except for data points in the ILR-transformed space $\mathbb R^2$.}	
	\label{fig:gaussian_depth_contour}
\end{figure}
Comparing Fig. \ref{fig:ILR_contour}(a) and Fig. \ref{fig:gaussian_depth_contour}(a), there is no obvious difference between contours in ternary plot for two depths. However, when comparing corresponding depths in Euclidean space in Fig. \ref{fig:ILR_contour}(b) and Fig. \ref{fig:gaussian_depth_contour}(b), the contours of ILR depth match the data pattern more accurately. This is reasonable as the model-based depth should rely on the accuracy of the model -- better models are expected to result in better ranks. In the remaining part of this paper, we will only use the original ILR depth to measure center-outward patterns in point processes.

\subsection{ILR depth for general point process}

\subsubsection{Definition and mathematical properties}
To derive the ILR depth for general point process, one approach is to derive the density function of the ILR transformation of the IET of a general point process.  This, in general, is highly challenging because the conditional intensity function is not a constant, but vary with respect to the event history. Consequently, we propose to adopt an alternative approach where the general point process is transformed to a homogeneous Poisson process, so that  all results in Section \ref{sec:ilrhpp} can be fully exploited. A commonly used transformation is the \textit{Time-Rescaling (TR) method}  \citep{brown2002time}. Denote $T_1 < s_1 < s_2 < \dots < s_k <  T_2$ as a realization from a point process with a conditional intensity function $\lambda(t|H_t)>0$ for all $t\in[T_1,T_2]$. Then, the sequence $\Lambda_S(s_i)=\int_{T_1}^{s_i}\lambda(t|H_t)dt$, $i=1,\dots,k$ is a Poisson process with the unit rate in $(0,\Lambda_S(T_2)]$. In this way, any point process can be transformed to an HPP as long as the conditional intensity function is known or can be estimated. Based on this result, we provide a definition of the ILR depth for general point process as follows. 

\begin{definition} \label{def:ilrgen}
For a general point process $\bm{s}\in\mathbb{S}_k$ in time domain $[T_1,T_2]$ with cardinality $k$ satisfying $T_1<s_1< s_2<\dots< s_k< T_2$, assume the conditional intensity function $\lambda(t|H_t)>0$ and denote $\Lambda_S(t)=\int_{T_1}^{t}\lambda(u|H_u)du$.  Let $s_0=T_1$ and $s_{k+1}=T_2$.  The ILR depth of $\bm{s}$ is defined as: 
\begin{eqnarray*}
    D_{c-TR}(\bm{s};P_{S||S|=k},\Lambda_S)=\frac{1}{1-\log\Big(\frac{c}{\Lambda_S(T_2)^{k+1}}\prod_{i=1}^{k+1}(\Lambda_S(s_i)-\Lambda_S(s_{i-1}))\Big)}
\end{eqnarray*}
where $c$ is a positive constant in $(0, e(k+1)^{k+1})$. We can let $c=(k+1)^{k+1}$ to have  the maximum depth value being $1$.  If $\bm{s}\in\mathbb{B}_k$, we define $D_{c-TR}(\bm{s};P_{S||S|=k})=0$. 
\end{definition}

\begin{remark}
Definition \ref{def:ilrgen} is a generalized version of Definition \ref{def:formal}. If $\bm{s}=(s_1,s_2,\dots,s_k)$ is a homogeneous Poisson process, then the conditional intensity function is a positive constant number $\lambda$. In this case, $\Lambda_S(s_i)-\Lambda_S(s_{i-1})=\lambda(s_i-s_{i-1})$ for $i=1,2,\dots,k+1$ and $\Lambda_S(T_2)=\lambda(T_2-T_1)$. With some simple algebra, one can show that the depth is the same with HPP case. In the remaining part of this paper, we fix $c=(k+1)^{k+1}$ so that the maximum depth value is $1$. 
\end{remark}

In the case of (inhomogeneous) Poisson process, the four properties listed in Section \ref{sec:depdef} can be verified for Definition \ref{def:ilrgen}.  This is summarized in the following proposition, where the proof is given in Appendix \ref{app:matipp}. 

\begin{prop}
All of the four mathematical properties in Section \ref{sec:depdef} hold for Definition \ref{def:ilrgen} for inhomogeneous Poisson process.  
\end{prop}

\begin{remark}
If the process is not a Poisson process, then the conditional intensity function depends on the history of time events. In this case, based on Appendix \ref{app:matipp}, Property \ref{m1} and Property \ref{m4} in Section \ref{sec:depdef} still hold. However, it was pointed out in \citet{qi2021dirichlet} that there exist counter-examples such that different point process realizations will be mapped to the same HPP after the time-rescaling method.  This implies the geometric center may not be unique. Therefore, Properties \ref{m2} and \ref{m3}, in general, do not hold for definition \ref{def:ilrgen}. 
\end{remark}

Based on Definition \ref{def:ilrgen}, the conditional intensity function must be used when applying time rescaling method. However, the true conditional intensity function is often unknown and needs an estimation. In the following subsections, we will examine the estimation procedures according to two commonly used processes. One is the inhomogeneous Poisson process (IPP), and the other is a non-Poisson process with Markovian property on inter-event times.

\subsubsection{ILR depth for inhomogeneous Poisson process} \label{sec:histogram_method}

For Poisson process, the process events are independent of each other and the intensity function is deterministic.  In this subsection, we propose to use a conventional histogram or binning method to estimate the intensity. Such method has been commonly used in density estimation \citep{silverman2018density}.  Based on the definition, the intensity function in an IPP is given as 
$$\lambda(t) =  \lim_{\Delta t \rightarrow 0} \frac{\mathbb{E}[N(t+ \Delta t)-N(t)]}{\Delta t},$$ 
where $N(t)$ is the counting measure on $[0, t)$.  The intensity can be approximated when the bin size $\Delta t $ is small. 
Therefore, the intensity as well as the depth value of an IPP can be estimated by binning method.  We have shown the estimation in Algorithm \ref{alg:hist} and the large sample theory on the estimation is discussed in detail in Section \ref{sec:asym}.  

\begin{algorithm}[ht!]
\caption{ILR depth estimation for inhomogeneous Poisson process}
\begin{algorithmic}
\STATE{\textbf{Input}: $n$ independent realizations of IPP; The time domain $[T_1,T_2]$; Number of bins $M$. }
\STATE{- Evenly divide $[T_1,T_2]$ into $M$ bins with equal width;}
\FOR{each $i=1,2,\dots,n$}
	\STATE{- Denote $n_i$ as the number of events of the $i$-th realization;}
	\STATE{- Denote $\bm{s_i}=(s_{i1},s_{i2},\dots,s_{in_i})$ as the $i$-th realization;}
	\STATE{- Denote $s_{i0}=T_1$, $s_{i(n_i+1)}=T_2$;}
\ENDFOR
\FOR{each $j=1,2,\dots,M$}
	\STATE{- Denote the $j$-th bin as $B_j$;}
	\STATE{- The intensity estimator is: $\hat{\lambda}(t)=\frac{M}{n(T_2-T_1)}\sum_{i=1}^{n}\sum_{r=1}^{n_i}I(s_{ir}\in B_j)$ if $t\in B_j$;}
\ENDFOR
\FOR{each $i=1,2,\dots,n$}
	\FOR{each $j=0,1,\dots,n_i+1$}
		\STATE{- Compute $\hat{\Lambda}_S^{(n)}(s_{ij})=\int_{T_1}^{s_{ij}}\hat{\lambda}(t)dt$;}
	\ENDFOR
	\STATE{- Depth of the $i$-th realization is: 
	\begin{eqnarray*}
	\hat{D}_i=\frac{1}{1-\log\Big(\frac{(n_i+1)^{n_i+1}}{\hat{\Lambda}_S^{(n)}(T_2)^{n_i+1}}\prod_{j=1}^{n_i+1}(\hat{\Lambda}_S^{(n)}(s_{ij})-\hat{\Lambda}_S^{(n)}(s_{i(j-1)}))\Big)}
	\end{eqnarray*}
	}
\ENDFOR
\STATE{\textbf{Output}: Depth values $\hat{D}_1,\dots,\hat{D}_n.$}
\end{algorithmic}
\label{alg:hist}
\end{algorithm}

We will then use two examples to illustrate the conditional ILR depth using ternary plots in IPP based on Definition \ref{def:ilrgen}. \\

{\bf Example 1:} Let the intensity function $\lambda(t)=\cos(4t)+1$ in the time domain $[0,\frac{\pi}{2}]$.  This function decreases from $0$ to $\frac{\pi}{4}$ and increases from $\frac{\pi}{4}$ to $\frac{\pi}{2}$ with $t=\frac{\pi}{4}$ as the global minimum point.  One can generate $1000$ independent realizations with cardinality $2$ and make the ternary plot to show the contours of the ILR depth. The ILR depth values in the IET domain and Euclidean space are shown in Fig. \ref{fig:contour_different}(a) and (b), respectively.  We can clearly see the non-triangularly shaped contours in both plots, which demonstrates the inhomogeneity of the Poisson processes.  

\begin{figure} [h!]
	\centering
	\subfigure[data points in simplex]{
		\begin{minipage}[b]{0.47\textwidth}
			\centering
			\includegraphics[scale=0.3]{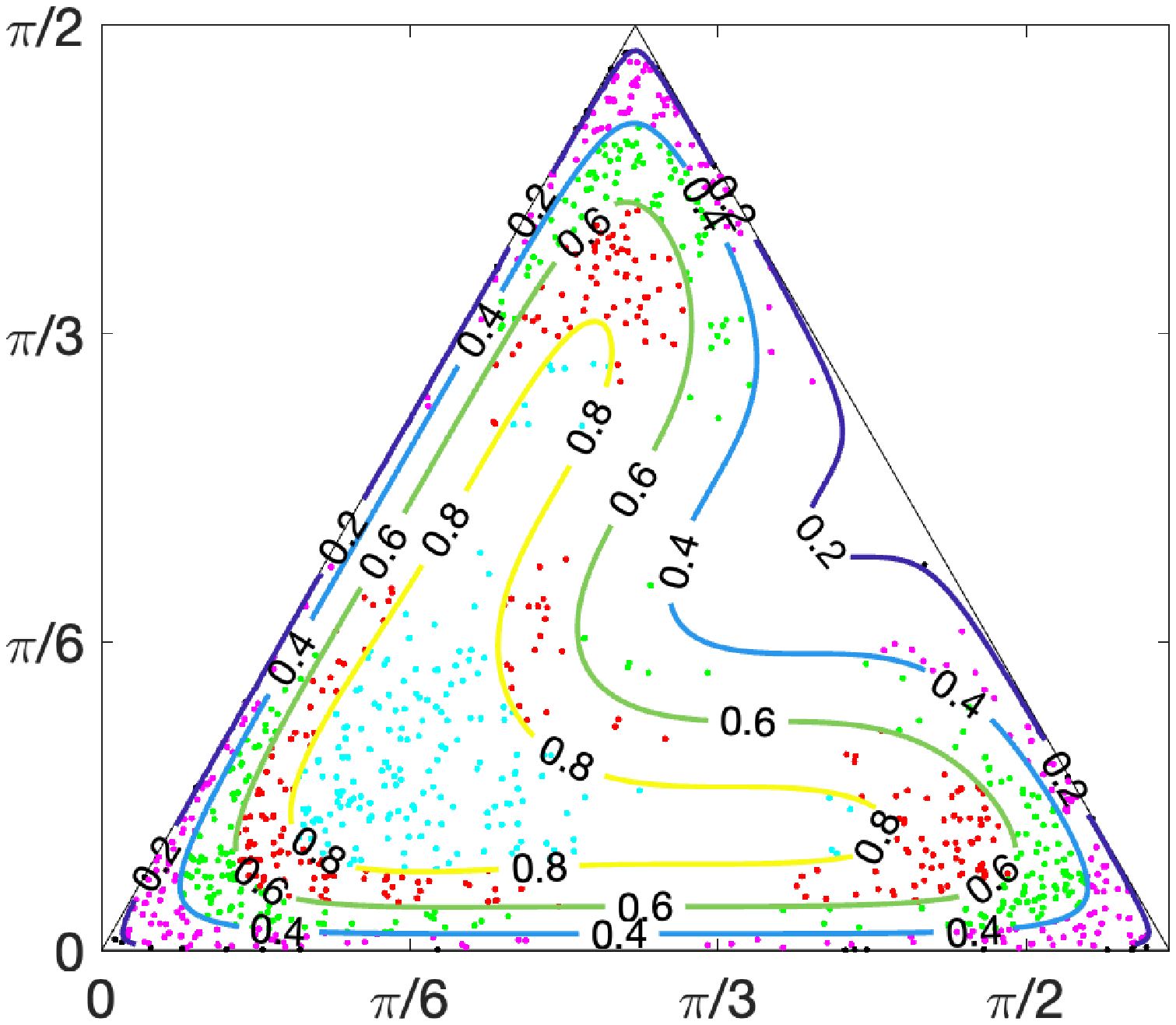}
		\end{minipage}
		\label{fig:ilrcos}
	}
	\subfigure[data points after ILR transformation]{
		\begin{minipage}[b]{0.47\textwidth}
			\centering
			\includegraphics[scale=0.3]{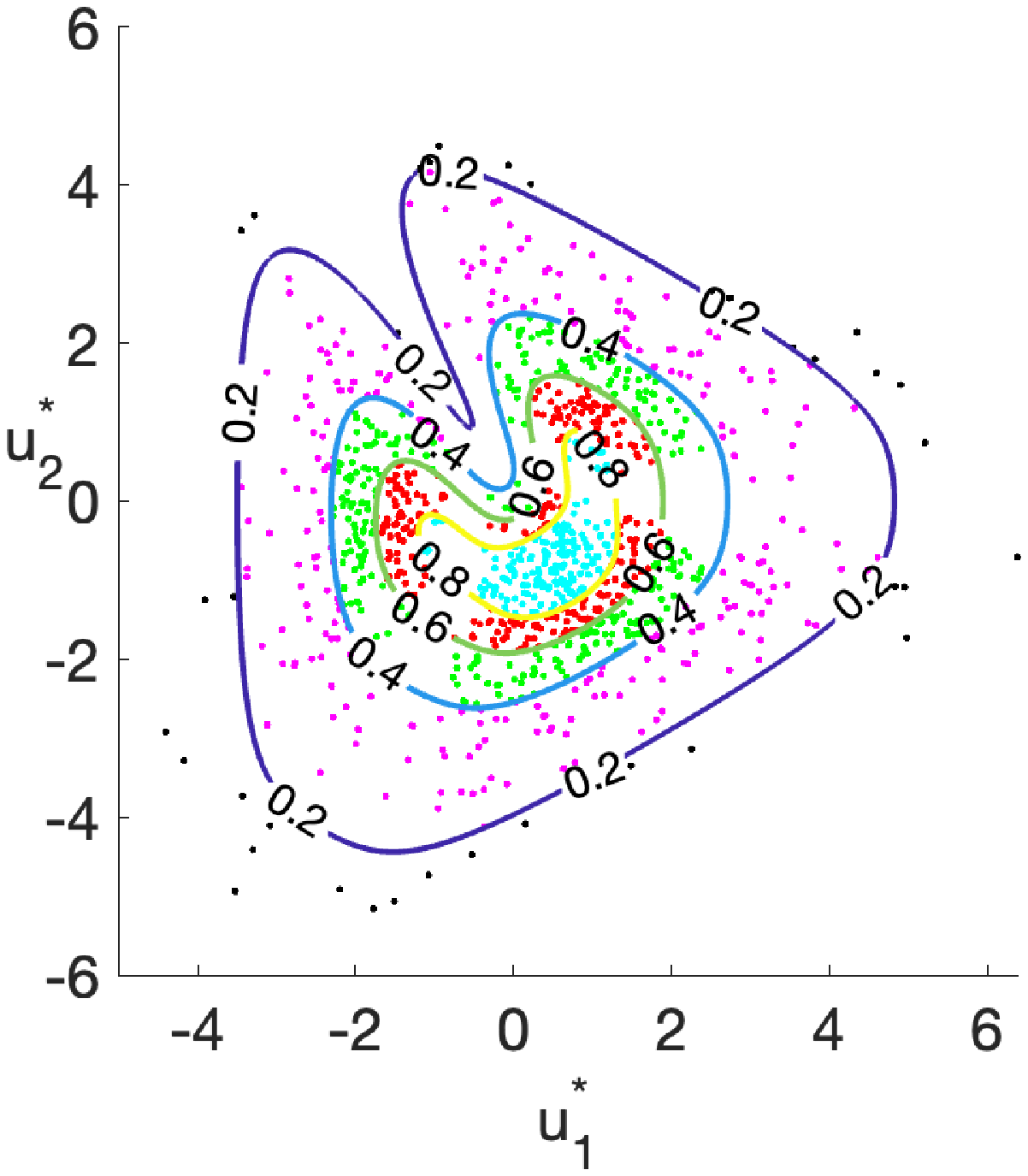}
		\end{minipage}
		\label{fig:cosilrmap}
	}
	\\
	\subfigure[data points in simplex]{
		\begin{minipage}[b]{0.47\textwidth}
			\centering
			\includegraphics[scale=0.3]{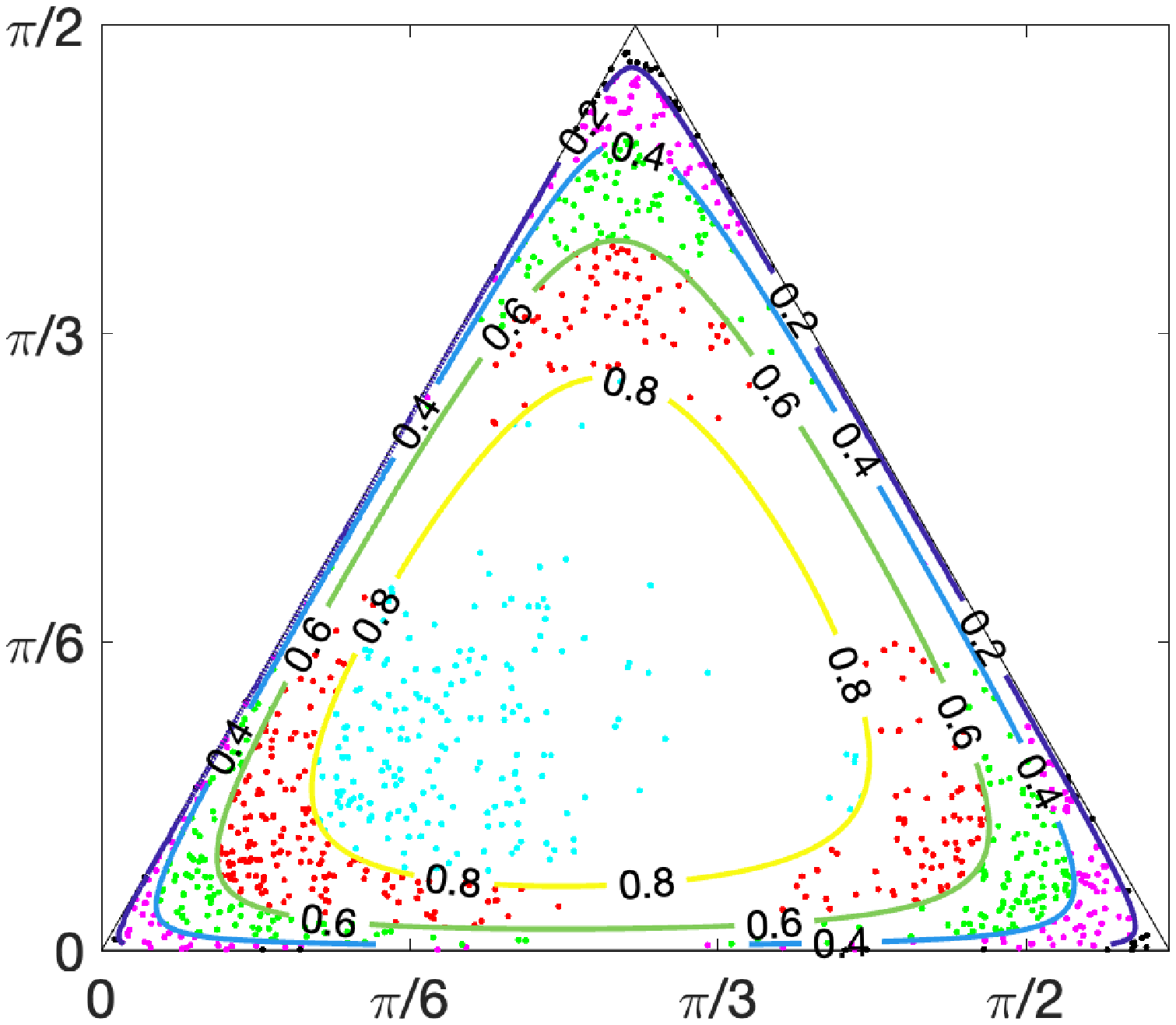}
		\end{minipage}
		\label{fig:sample_dir}
	}	
	\subfigure[data points in simplex]{
		\begin{minipage}[b]{0.47\textwidth}
			\centering
			\includegraphics[scale=0.3]{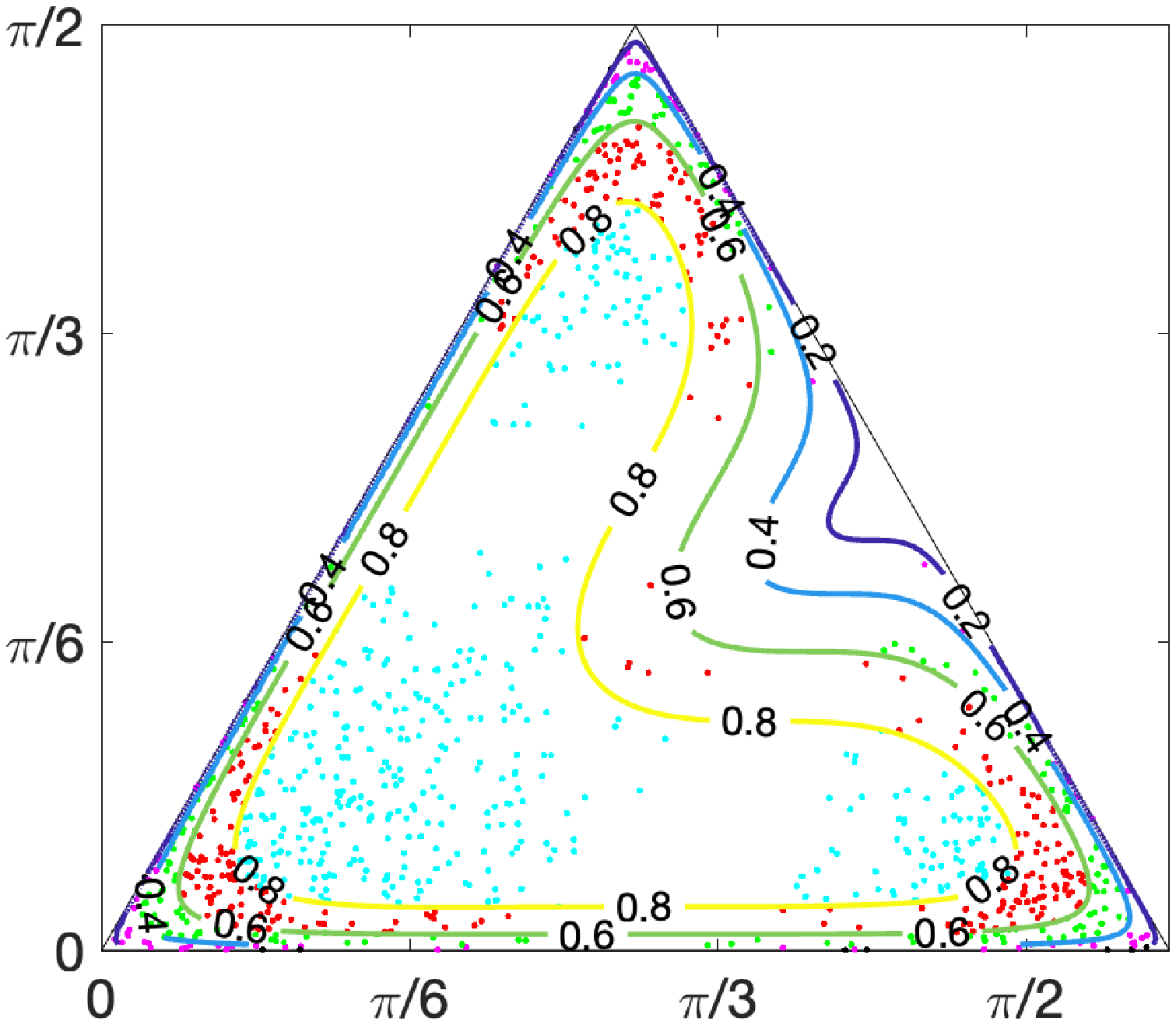}
		\end{minipage}
		\label{fig:dir_tr}
	}	
	\caption{Depth results of IPP with $\lambda(t)=\cos(4t)+1$. (a) Data points in simplex shown by ternary plot, where the points are colored with respect to ranges of the ILR depth values given in Definition \ref{def:ilrgen}. The solid lines indicate the depth contours with specific values. (b) Same as (a) except for data points in the ILR transformed space $\mathbb{R}^k$. (c) Same as (a) except by using the sample Dirichlet depth in \citet{qi2021dirichlet}. (d) Same as (a) except by using the Dirichlet depth after time-rescaling in \citet{qi2021dirichlet}.}
	\label{fig:contour_different}
\end{figure}

Next, the sample Dirichlet depth and the Dirichlet depth with time rescaling method in \citet{qi2021dirichlet} can be applied to the same IPP sample to compute conditional depth with cardinality $2$. The results are shown in Fig. \ref{fig:contour_different}(c) and (d). Given the fact that the value of intensity function is close to $0$ around the middle range, it is not likely to have a process realization such that the second IET is small, meanwhile the first and third IET are almost the same. However, in Fig. \ref{fig:contour_different}(c), this characteristic cannot be reflected. Thus, the sample Dirichlet depth \citep{qi2021dirichlet} is not a good depth choice for IPP. Next, comparing Fig. \ref{fig:contour_different}(a) with \ref{fig:contour_different}(d), the shape of the two contours are identical. Nevertheless, there are fewer realizations located around the middle range of the right side in the ternary plot, e.g. there are fewer realizations with the two time events occurred around the middle part within the time domain. A reasonable depth should capture this feature and give much lower depth value to these realizations. Hence, the ILR depth is more reasonable since it satisfies this property better than the Dirichlet depth with time rescaling. 

{\bf Example 2:} We have shown the ILR depth given that the cardinality is 2.  Now we will check the depth ranking performance for the overall depth in Definition \ref{def:wholedef}. $1000$ independent realizations are simulated with intensity function $\lambda(t)=\cos(t)+1$ in the time domain $[0,2\pi]$. Then the expected value of the number of time events of a realization is $\mathbb{E}[N(2\pi)-N(0)]=\int_{0}^{2\pi}(\cos(t)+1)dt=2\pi\approx 6$. Thus, a processes with $6$ events is expected to have a large 
one dimensional depth in Definition \ref{def:wholedef}.  In Fig. \ref{fig:costop10}, we display the one-dimensional depth with respect to different cardinalities. We also show the top 10 processes with the largest overall depths when the weight coefficient $r$ = 1 and 0.1, respectively.  

\begin{figure} [h!]
	\centering
%
	\subfigure[one dimensional depth]{
		\begin{minipage}[b]{0.3\textwidth}
			\centering
			\includegraphics[scale=0.2]{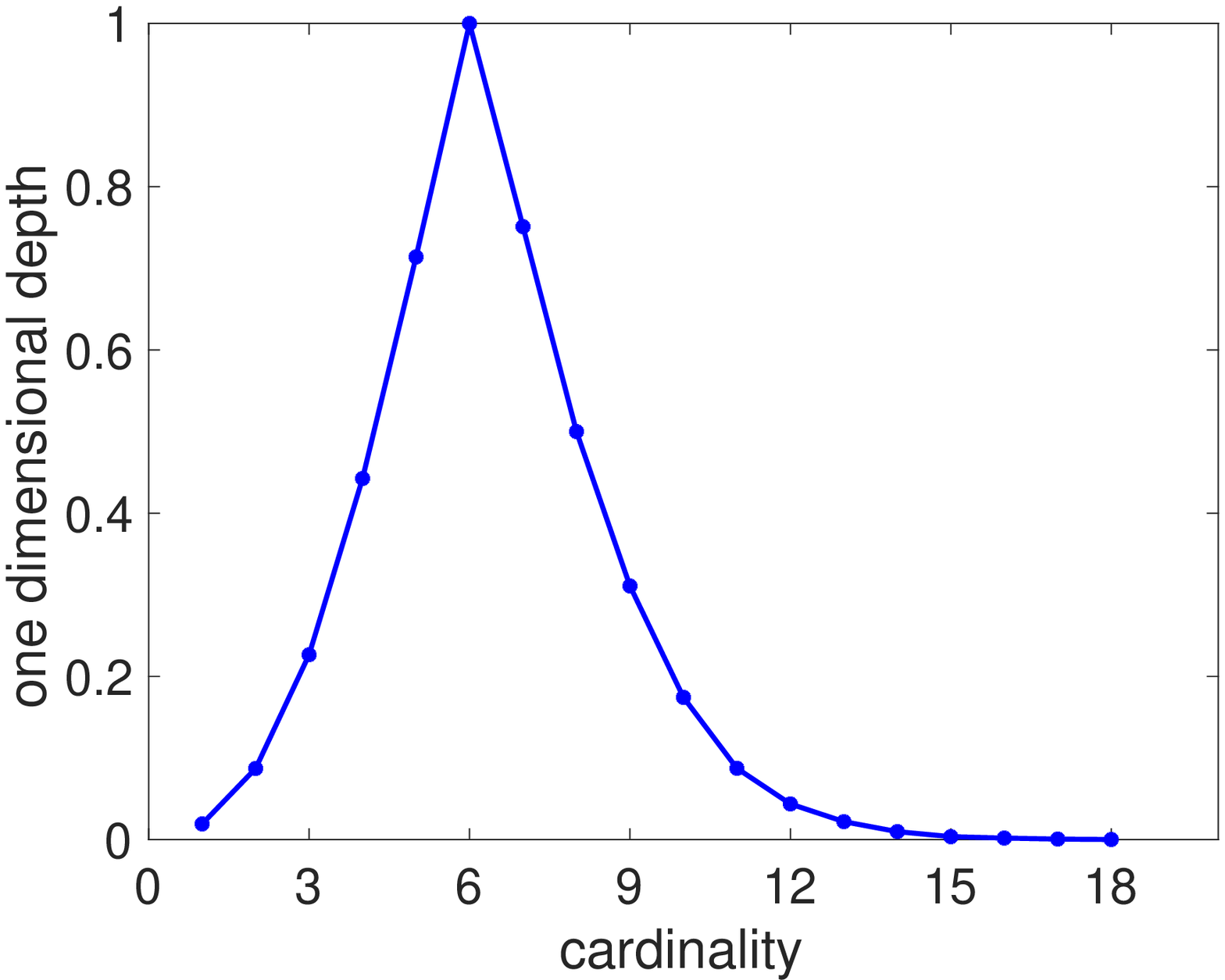}
		\end{minipage}
	}
	\subfigure[depth ranking if $r=1$]{
		\begin{minipage}[b]{0.3\textwidth}
			\centering
			\includegraphics[scale=0.2]{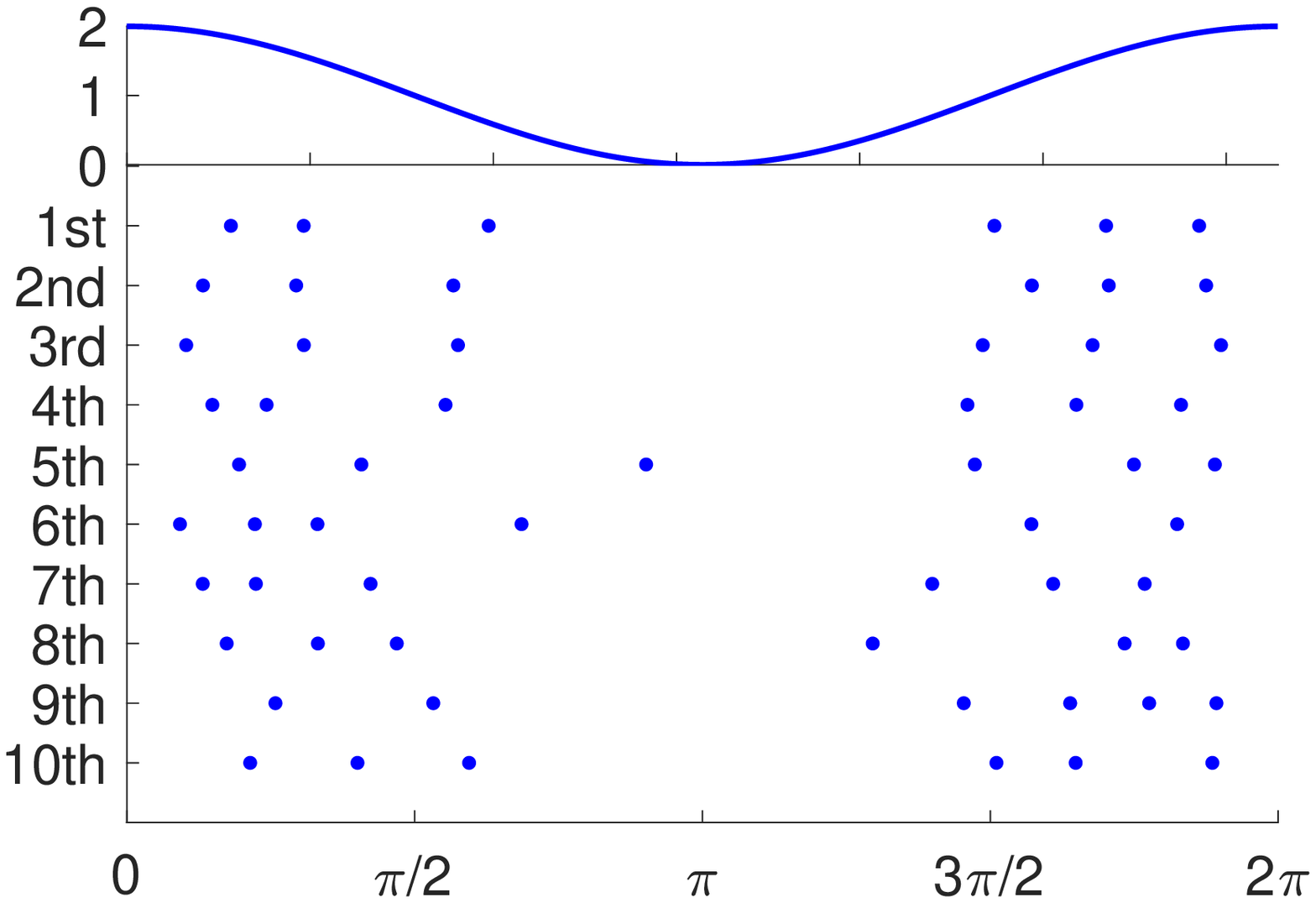}
		\end{minipage}
	}
	\subfigure[depth ranking if $r=0.1$]{
    		\begin{minipage}[b]{0.3\textwidth}
			\centering
   		 	\includegraphics[scale=0.2]{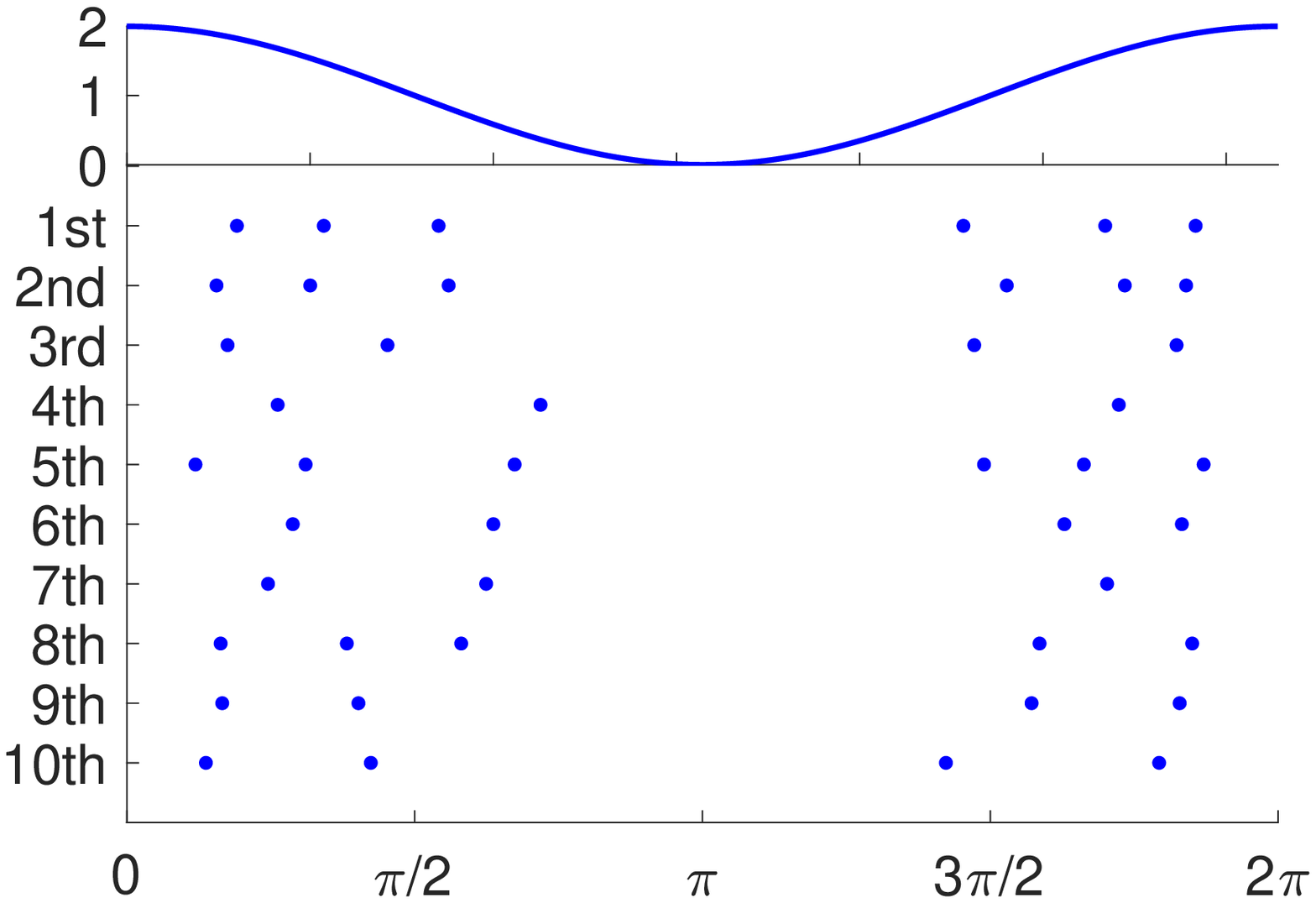}
    		\end{minipage}
    	}
	\caption{One dimensional depth and top $10$ depth-valued inhomogeneous point processes with $\lambda(t)=\cos(t)+1$. (a) One dimensional depth for different cardinalities in the simulated sample. (b) $\lambda(t)$ is displayed in the top panel.  In the bottom panel, each row is a realization of simulated IPP with the ranking of depth value shown in vertical axis. The depth is computed by Definition \ref{def:wholedef}, where the ILR depth combined with time-rescaling in Definition \ref{def:ilrgen} is the conditional depth and the hyper-parameter $r=1$. (c) Same as (b) except for $r=0.1$.}
	\label{fig:costop10}
\end{figure}

From Fig. \ref{fig:costop10}(a), one can conclude that the one dimensional depth obtains the maximum value when cardinality is $6$. 
When $r=1$, the one dimensional depth in Definition \ref{def:wholedef} is dominant, which makes realizations with cardinality $6$ rank at the top places.  This can be seen in Fig. \ref{fig:costop10}(b), where all realizations in top 10 places have $6$ time events. When $r=0.1$, the number of events is less dominant, and the distribution of event times becomes more important in the ranking process.  Fig. \ref{fig:costop10}(c) shows that the realizations with depth in top 10 places have cardinality different from 6. However, the locations of time events of top 10 depth realizations are close to the density pattern of the intensity function (i.e., there are more events when intensity is large, and fewer events when intensity is small).  



\subsubsection{ILR depth for non-Poisson process}
If the point process is non-Poisson process in time domain $[T_1,T_2]$, then there exists history dependence in conditional intensity function and the Histogram method is not applicable.  The estimation of the conditional intensity $\lambda(t|H_{t})$ is, in general, highly challenging.  A tractable simplification assumes the Markovian property in the following form \citep{kass2001spike}:
\begin{eqnarray} \label{eq:markov}
\lambda(t|H_{t})=\lambda(t,t-s_*(t))=\lambda_1(t)\lambda_2(t-s_{*}(t)),
\end{eqnarray}
where $\lambda_1(\cdot), \lambda_2(\cdot)$ are two deterministic intensity functions and $s_{*}(t)$ is the last time event preceding to $t$. If there is no time event before time $t$, then denote $s_*(t)=T_1$.  Point process with this simplified conditional intensity function is called an \textit{inhomogeneous Markov interval (IMI) process}. 

To obtain an estimation of Eqn. \eqref{eq:markov} from observed data, we adopt a non-parametric approach proposed by \citet{wojcik2009direct}. To simplify the notation, denote $\tau=t-s_*(t)$. Thus, Eqn. \eqref{eq:markov} can be rewritten as $\lambda(t,\tau)=\lambda_1(t)\lambda_2(\tau)$. In this case, $\lambda(t,\tau)$ can be considered as the product of two intensity functions, one depends on the current time $t$ and the other one is only corresponding to the inter-event time. The computation of the ILR depth with the IMI model to estimate $\lambda_1(t)$ and $\lambda_2(\tau)$ is shown in Algorithm \ref{alg:imi}.  

\begin{algorithm}[ht!]
\caption{Depth estimation for inhomogeneous Markov interval process}
\begin{algorithmic}
\STATE{\textbf{Input}: $n$ independent realizations of an IMI process on $[T_1,T_2]$; Number of bins $M_1$ to estimate $\lambda_1$; Number of bins $M_2$ to estimate $\lambda_2$.}
\STATE{- Calaulate all inter-event times in the data. Denote the largest one as $L$;}
\STATE{- Uniformly divide $L$ into $M_2$ bins and denote the bin width as $dt$;}
\FOR{each $i=1,2,\dots,n$}
	\STATE{- Denote $n_i$ as the number of events of the $i$-th realization; Denote $\bm{s_i}=(s_{i1},s_{i2},\dots,s_{in_i})$ as the $i$-th realization; Denote $s_{i0}=T_1$, $s_{i(n_i+1)}=T_2$; Denote $\bm{u_i}=(u_{i1},\dots,u_{in_i})=(s_{i1}-s_{i0},\dots,s_{in_i}-s_{i(n_i-1)})$;}
\ENDFOR
\FOR {each $i=1,2,\dots,M_2$}
	\STATE{- Denote the $i$-th bin as $B_i$;}
	\STATE{- For any $\tau\in B_i$, compute the density of IET $p(\tau)=\frac{\sum_{j=1}^{n}\sum_{k=1}^{n_j}I(u_{jk}\in B_i)}{dt\cdot\sum_{j=1}^{n}n_j}$;}
	\STATE{- Compute the conditional intensity function $\lambda_2(\tau)=\frac{p(\tau)}{1-\int_0^\tau p(\tau')d\tau'}$;}
\ENDFOR
\STATE{- Uniformly divide $[T_1,T_2]$ into $M_1$ bins and denote the bin width as $\Delta t$;}
\FOR{each $k=1,2,\dots,M_1$}
	\STATE{- Denote the $k$-th bin as $A_k$ and $t_k=(k-\frac{1}{2})\Delta t$}
	\STATE{- Denote $\tau_k^j=t_k-s_ *^j(t_k)$, where $s_ *^j(t_k)$ is the nearest events before $t_k$ in the $j$th relization, $j=1,2,\dots,n$}
	\STATE{- Probability of an event in the $k$-th bin $p_k=\frac{\sum_{i=1}^n\sum_{r=1}^{n_i}I(s_{ir}\in A_k)}{\sum_{i=1}^{n}n_i}$;}
	\STATE{- Compute $\lambda_1(t)=\frac{p_k\cdot n}{\Delta t\cdot\sum_{j=1}^n \lambda_2(\tau_k^j)}$ if $t\in A_k$;}

\ENDFOR
\FOR{each $i=1,2,\dots,n$}
	\STATE{- Compute $\hat{\Lambda}_S^{(n)}(s_{ij})=\int_{T_1}^{s_{ij}}\lambda_1(t)\lambda_2(t-s_*(t))dt$, $j=0,1,\dots,n_i+1$}
	\STATE{- 
	$\hat{D}_i=\frac{1}{1-\log\Big(\frac{(n_i+1)^{n_i+1}}{\hat{\Lambda}_S^{(n)}(T_2)^{n_i+1}}\prod_{j=1}^{n_i+1}(\hat{\Lambda}_S^{(n)}(s_{ij})-\hat{\Lambda}_S^{(n)}(s_{i(j-1)}))\Big)}$
	}
\ENDFOR
\STATE{\textbf{Output}: $\hat{D}_1,\dots,\hat{D}_n$ are the depth values.}
\end{algorithmic}
\label{alg:imi}
\end{algorithm}

We will use one example to illustrate the ILR depth for an IMI process based on Definition \ref{def:ilrgen}. Suppose the conditional intensity function is $\lambda(t|H_t)=(\sin(t)+1)\cdot(\sin(t-s_*(t)-\frac{\pi}{2})+1)$, where $s_*(t)$ is the last time event preceding to $t$.  $10000$ realizations are generated in the time domain $[0,2\pi]$.  We note that the conditional intensity function varies with respect to event history.  We use Algorithm \ref{alg:imi} to estimate the conditional intensity function as well as depth value for each realization.  The conditional intensity estimate on one typical realization is shown in Fig. \ref{fig:smooth_sin}.   We can see that the IMI model provides a much better estimate than the event-independent Histogram method.

\begin{figure} [h!]
\centering
\includegraphics[scale=0.4]{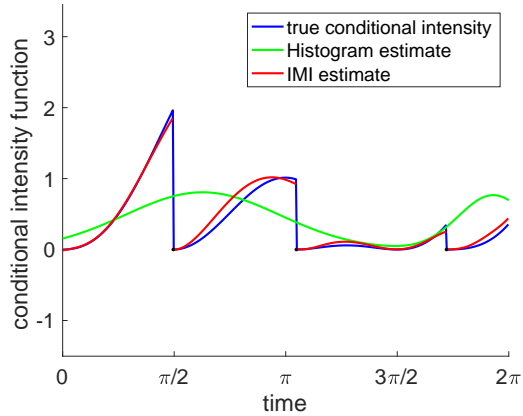}
\caption{Estimation of conditional intensity in a typical realization.  The three solid lines denote true conditional intensity (blue), Histogram estimate (green), and IMI estimate (red), respectively.}
\label{fig:smooth_sin}
\end{figure}

Similar to the IPP study, the ranking performance of the $10000$ realizations can be evaluated by selecting realizations with top 10 depth values.  The result is shown in Fig. \ref{fig:estimation}, where we display the indexed top 10 depth-valued processes using the true conditional intensity, estimated by the IMI method, and estimated by the Histogram method, respectively.   We can see the processes using true intensity (Panel (a)) and IMI method (Panel (b)) have a lot of overlaps. Indeed, 7 indices, $2200$, $7505$, $5099$, $4433$, $9423$, $9088$ and $6931$, appear in both Fig. \ref{fig:estimation}(a) and \ref{fig:estimation}(b) out of the top 10 realizations. However, in Fig. \ref{fig:estimation}(c), none of the realizations depth are ranked as top 10 in Fig. \ref{fig:estimation}(a). This result shows that the depth calculated by the IMI model is more accurate for general point process as compared to the Histogram method. 

\begin{figure} [h!]
	\centering
	\subfigure[true intensity]{
		\begin{minipage}[b]{0.3\textwidth}
			\centering
			\includegraphics[scale=0.2]{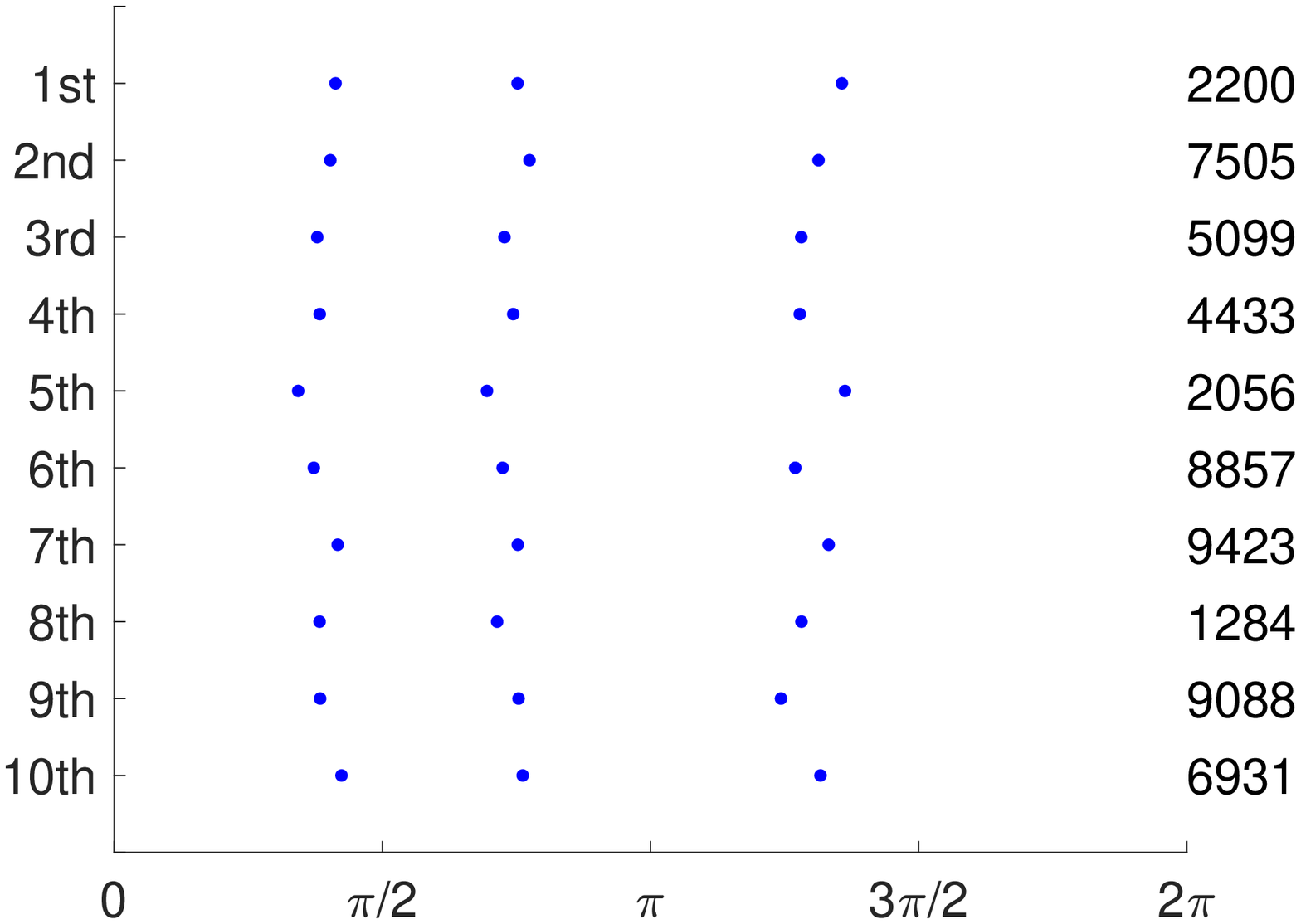}
		\end{minipage}
		\label{fig:sincom1}
	}
    	\subfigure[IMI model]{
    		\begin{minipage}[b]{0.3\textwidth}
			\centering
   		 	\includegraphics[scale=0.2]{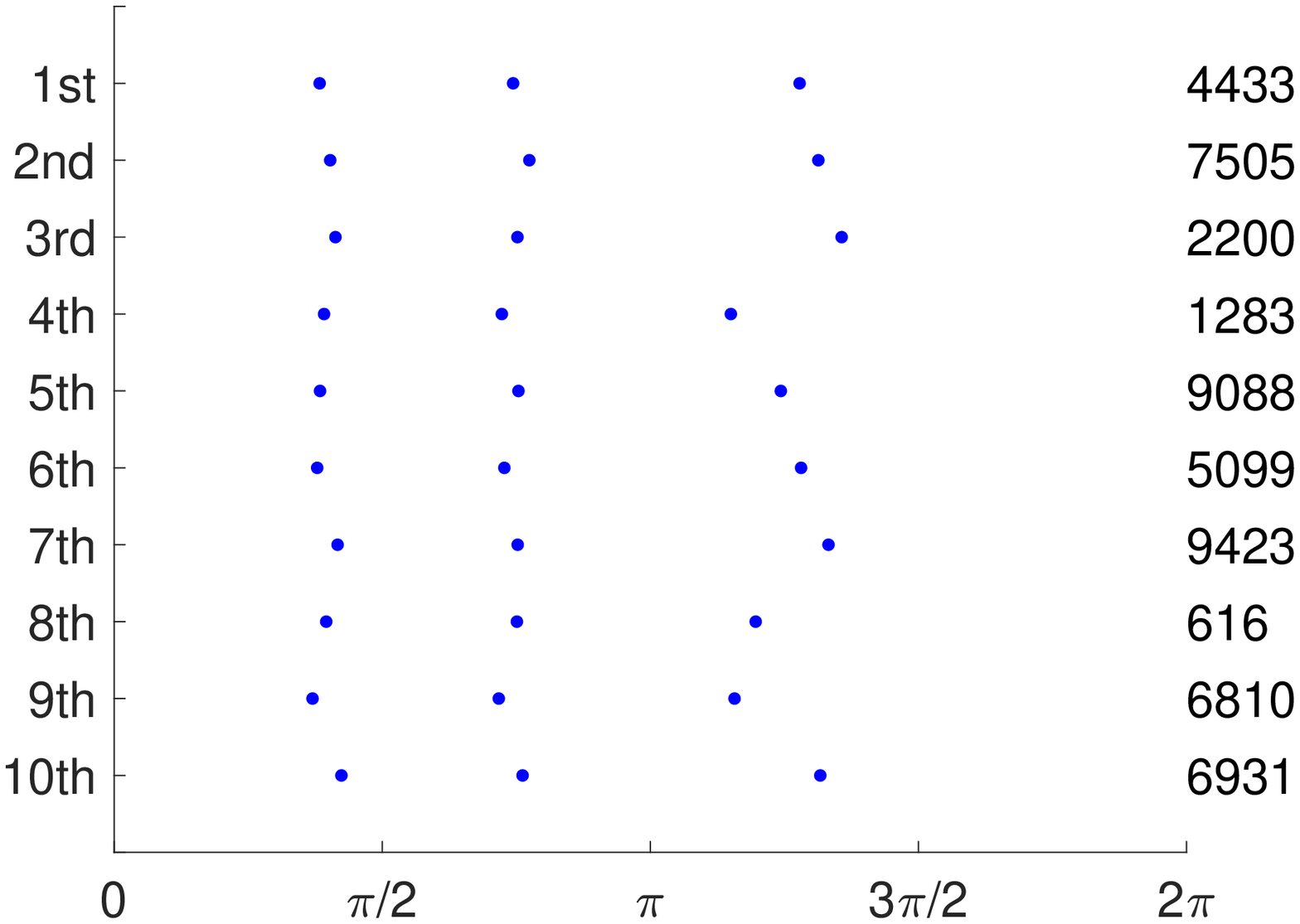}
    		\end{minipage}
		\label{fig:sincom2}
    	}
	\subfigure[Histogram method]{
    		\begin{minipage}[b]{0.3\textwidth}
			\centering
   		 	\includegraphics[scale=0.2]{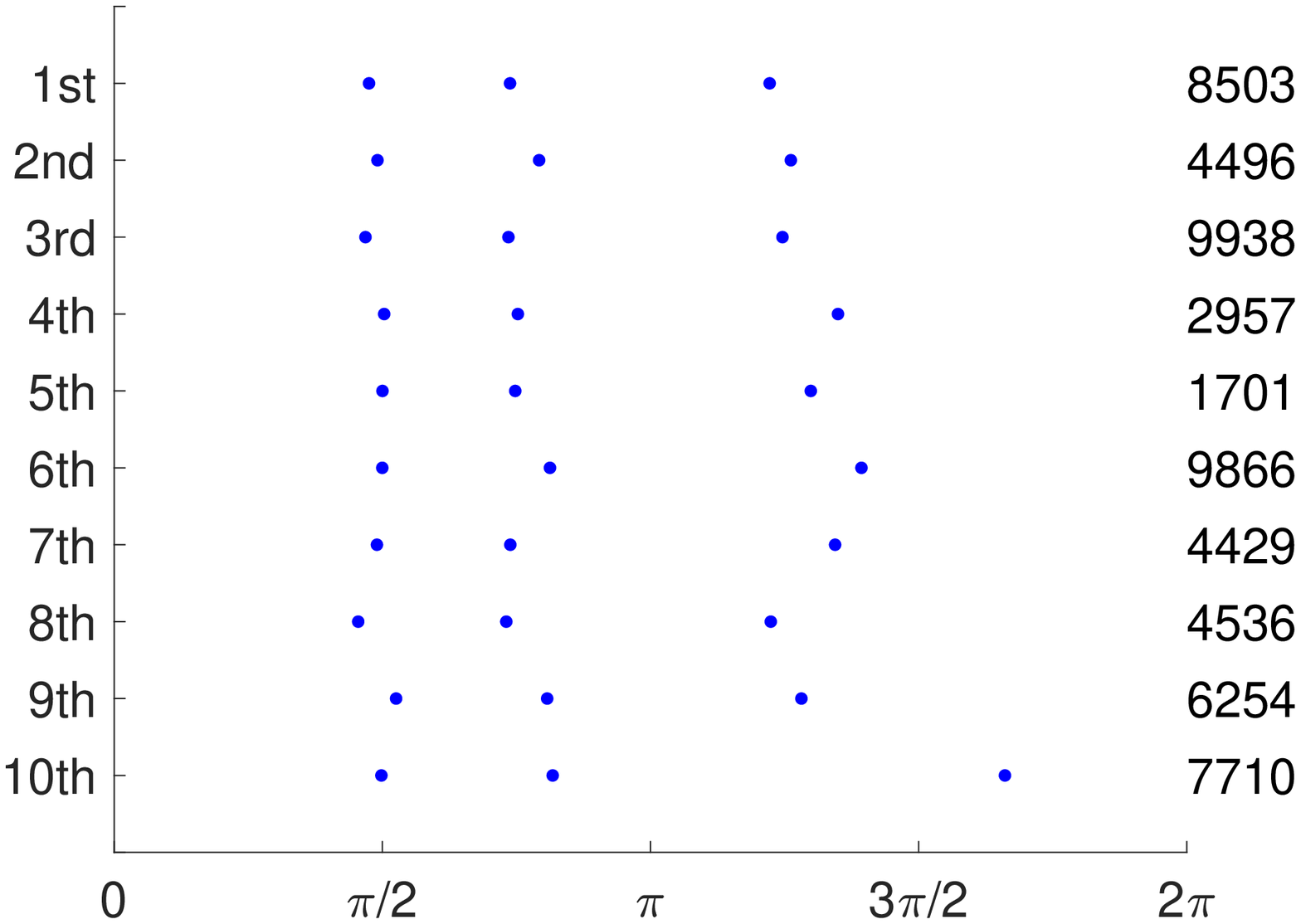}
    		\end{minipage}
		\label{fig:sincom3}
    	}
	\caption{Top $10$ depth-valued processes for different conditional intensity function estimations. (a) Each row is a realization of simulated point process with the rankings of depth value shown on left vertical axis and the index shown on right vertical axis. The depth is computed by Definition \ref{def:wholedef}, where the conditional ILR depth is estimated using true conditional intensity in the time-rescaling transformation and the hyper-parameter $r=1$. (b) Same as (a) except that the IMI model is used in the time-rescaling. (c) Same as (a) except that the Histogram method is  in the time-rescaling.}
	\label{fig:estimation}
\end{figure}

\section{Asymptotic Theory} \label{sec:asym}
In this section, we will study the asymptotics of the estimated depth values using Algorithm  \ref{alg:hist}.  We will prove that the sample depth value computed with estimated $\hat{\lambda}(t)$ converges to the population depth value computed by using the true intensity function $\lambda(t)$. Notice that Algorithm \ref{alg:hist} adopts the Histogram method, which is only applicable for Poisson process. We at first have a uniform convergence result about the estimated integrated intensity function in the following lemma, where the proof is given in Appendix \ref{app:uniformproof}. 
\begin{lemma} \label{lemma:uniform} 
Suppose $\lambda(t)$ is the true intensity function of a sample of IPP in $\mathbb{S}$ with sample size $n$. Assume $\sup_t\lambda(t)\leq R$ ($R$ is a positive finite number) and $\lambda(t)$ is $L$-Lipschitz continuous on $[T_1,T_2]$, i.e., for any $x,y\in [T_1,T_2]$, $|\lambda(x)-\lambda(y)|\leq L|x-y|$ for a finite $L$. Denote $\hat{\lambda}(t)$ as the estimated intensity function based on Algorithm \ref{alg:hist}. Let $M$ be the number of bins in Algorithm \ref{alg:hist}, and $\Lambda_S(x)=\int_{T_1}^x\lambda(t)dt$, $\hat{\Lambda}_S^{(n)}(x)=\int_{T_1}^x\hat{\lambda}(t)dt$, $T_1\leq x\leq T_2$.  Then the following uniform convergence rate holds: 
\begin{eqnarray*}
\sup_x|\hat{\Lambda}_S^{(n)}(x)-\Lambda_S(x)|=O\Big(\frac{1}{M}\Big)+O_P\Big(\sqrt{\frac{M^2}{n}}\Big)
\end{eqnarray*}
\end{lemma}

Based on Lemma \ref{lemma:uniform}, with some simple algebra, it is straightforward to conclude $M_{opt}=O(n^{\frac{1}{4}})$ is the optimal choice of $M$. Therefore, if $n\to\infty$ and $M\propto n^{\frac{1}{4}}$, $\hat{\Lambda}_S^{(n)}(x)$ uniformly converges to $\Lambda_S(x)$ in probability. 

Using this result, we can obtain the main conclusion on the convergence of the sample ILR depth in the following theorem. 

\begin{thm} \label{thm:final_convergence} 
Under the same assumptions as given in Lemma \ref{lemma:uniform}, denote $D_{c-TR}(s;P_{S||S|=k},\hat{\Lambda}_S^{(n)})$ as the sample ILR conditional depth computed by Algorithm \ref{alg:hist} with estimated intensity function. Also denote $D_{c-TR}(s;P_{S||S|=k},\Lambda_S)$ as the population ILR conditional depth computed by using the true conditional intensity function. Then we have,
\begin{eqnarray*}
\sup_{\bm{s}\in\mathbb{S}_k}|D_{c-TR}(\bm{s};P_{S||S|=k},\hat{\Lambda}_S^{(n)})-D_{c-TR}(\bm{s};P_{S||S|=k},\Lambda_S)|\to 0 
\end{eqnarray*}
in probability as $n\to\infty$.
\end{thm}
\begin{proof}
Since $\bm{s}\in\mathbb{S}_k$ and $\mathbb{S}_k=\{\bm{s}=(s_1,s_2,\dots,s_k)|T_1\leq s_1\leq s_2\leq\dots\leq s_k\leq T_2\}$, the time-rescaling result $\Lambda_S(\bm{s})$ belongs to a bounded and closed set $\Lambda_S(\mathbb{S}_k)=\{(s_1,s_2,\dots,s_k)|\Lambda_S(T_1)\leq s_1\leq s_2\leq\dots\leq s_k \leq \Lambda_S(T_2)\}$. Thus, $\Lambda_S(\mathbb{S}_k)$ is a bounded and closed subset of Euclidean space $\mathbb{R}^{k}$. From \textit{Heine–Borel theorem}, $\Lambda_S(\mathbb{S}_k)$ is a compact set and a continuous function defined on $\Lambda_S(\mathbb{S}_k)$ is a uniform continuous function. Consequently, $D_{c-TR}(s;P_{S||S|=k},\Lambda_S)$ is a uniform continuous function on $\Lambda_S(\mathbb{S}_k)$. Therefore, from the \textit{continuous mapping theorem}, $\sup_{s\in\mathbb{S}_k}|D_{c-TR}(s;P_{S||S|=k},\hat{\Lambda}_S^{(n)})-D_{c-TR}(s;P_{S||S|=k},\Lambda_S)|\to 0$ in probability as $n\to\infty$. 
\end{proof}

\section{Real Data Application} \label{sec:real_data_app}

In this section, we will illustrate the proposed depth method in a dataset from the real world. 
We consider the occurrence times of car accidents from 2016 to 2020 in Tallahassee, Florida. The data can be retrieved at the link \textit{https://www.kaggle.com/sobhanmoosavi/us-accidents} and was previously used in \citet{moosavi2019countrywide,moosavi2019accident}. Tallahassee is the capital city of Florida and there is a highway \textit{I-10} located at the northern region of the city.  Majority of people live in Tallahassee commute via local roads and the highway is mainly used by long-distance travelers. Our dataset includes accident occurrence times in two types of roads: highway (i.e., I-10) and local roads. For each type, the occurrence times are recorded in the time domain $[0,24]$ in the units of hours.

If at least one car accident was recorded for a specific day, the accident occurrence times in that day are treated as a realization of a point process. Since each accident can be assumed to be independent of each other, we will consider the point process as an inhomogeneous Poisson process and focus on the accident occurrence pattern for both local roads and the highway I-10. Hence, the histogram method introduced in Section \ref{sec:histogram_method} can be used to both data groups to estimate the intensity functions, and the result is shown in Fig. \ref{fig:car_inten_esti}.

\begin{figure} [h!]
	\centering
	\subfigure[local roads]{
		\begin{minipage}[b]{0.47\textwidth}
			\centering
			\includegraphics[scale=0.3]{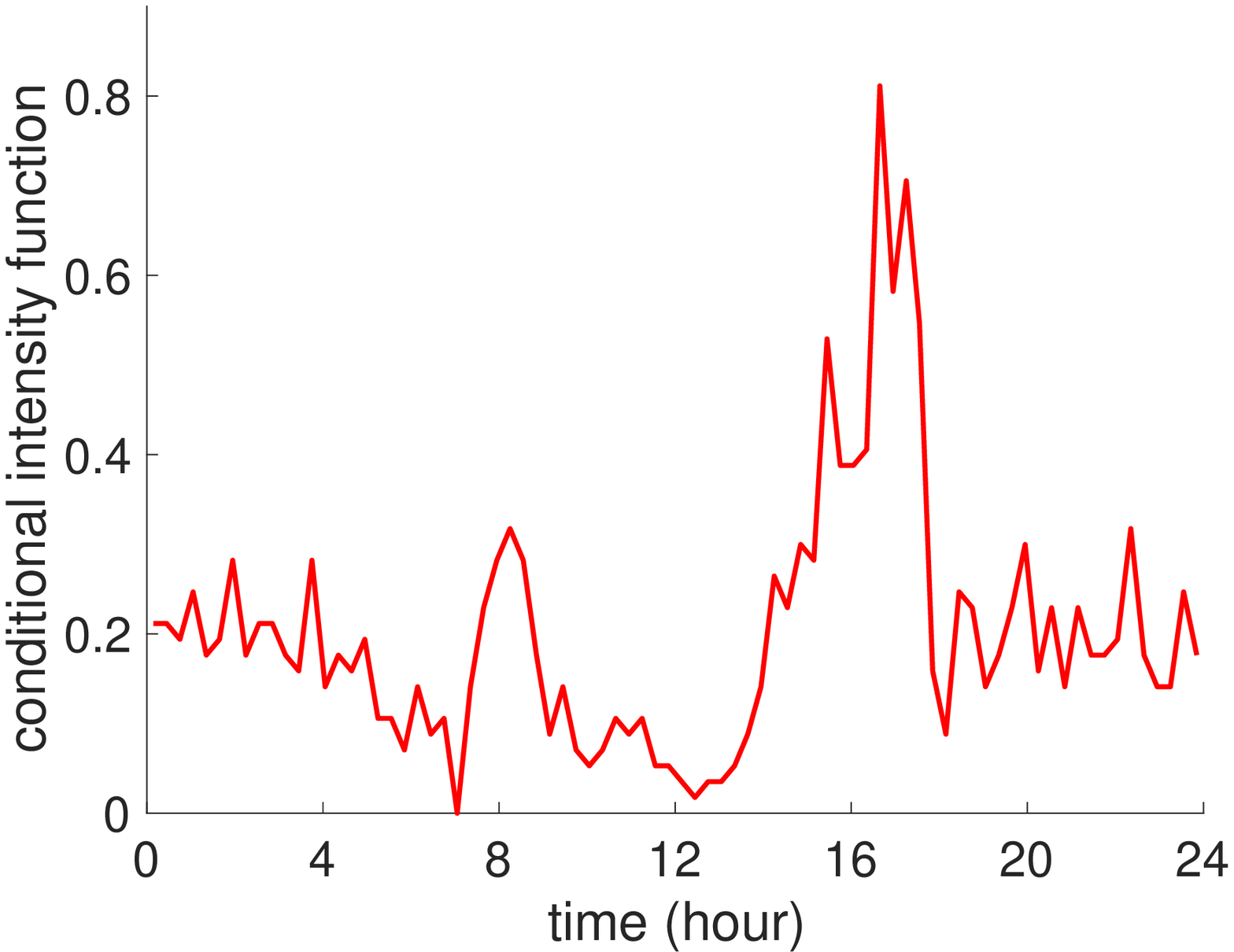}
		\end{minipage}
		\label{fig:local_inten}
	}
    	\subfigure[highway]{
    		\begin{minipage}[b]{0.47\textwidth}
			\centering
   		 	\includegraphics[scale=0.3]{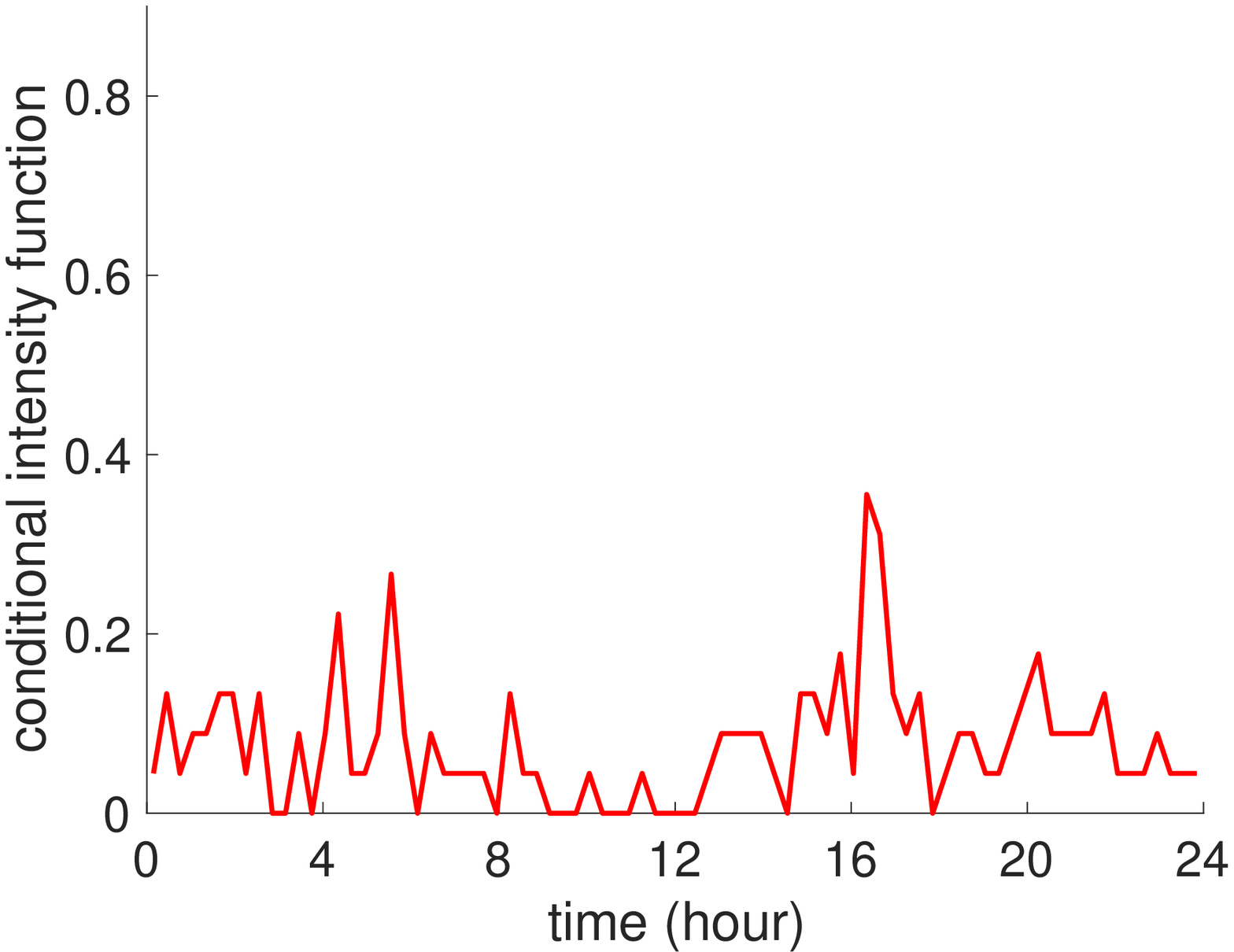}
    		\end{minipage}
		\label{fig:highway_inten}
    	}
	\caption{Estimation result. (a) Estimated intensity function of local roads. (b) Estimated intensity function of highway. }
	\label{fig:car_inten_esti}
\end{figure}

From Fig. \ref{fig:car_inten_esti}(a), we can see there exist two large peak regions for the estimated conditional intensity function for local roads: One is a global maximum region at around 5pm, which is the rush hour in the afternoon. The other is a local maximum region at around 8am, which is the rush hour in the morning.  It is also interesting to notice that the value of intensity function is even higher at night than at noon in local roads, which indicates noon may be the safest time to drive in local roads of Tallahassee. In contrast, the peak of the intensity function in the highway, shown in Fig. \ref{fig:car_inten_esti}(b), does not have any obvious pattern. There is a global maximum region around 4pm (starting time of rush hour in the afternoon), but there is no apparent peak in morning rush hours.  This result clearly shows that accidents in the local roads are mainly affected by the rush hour traffic, whereas those on the highway are not very related to it.  

\begin{figure} [h!]
	\centering
	\subfigure[one dimensional depth for local roads]{
		\begin{minipage}[b]{0.47\textwidth}
			\centering
			\includegraphics[scale=0.3]{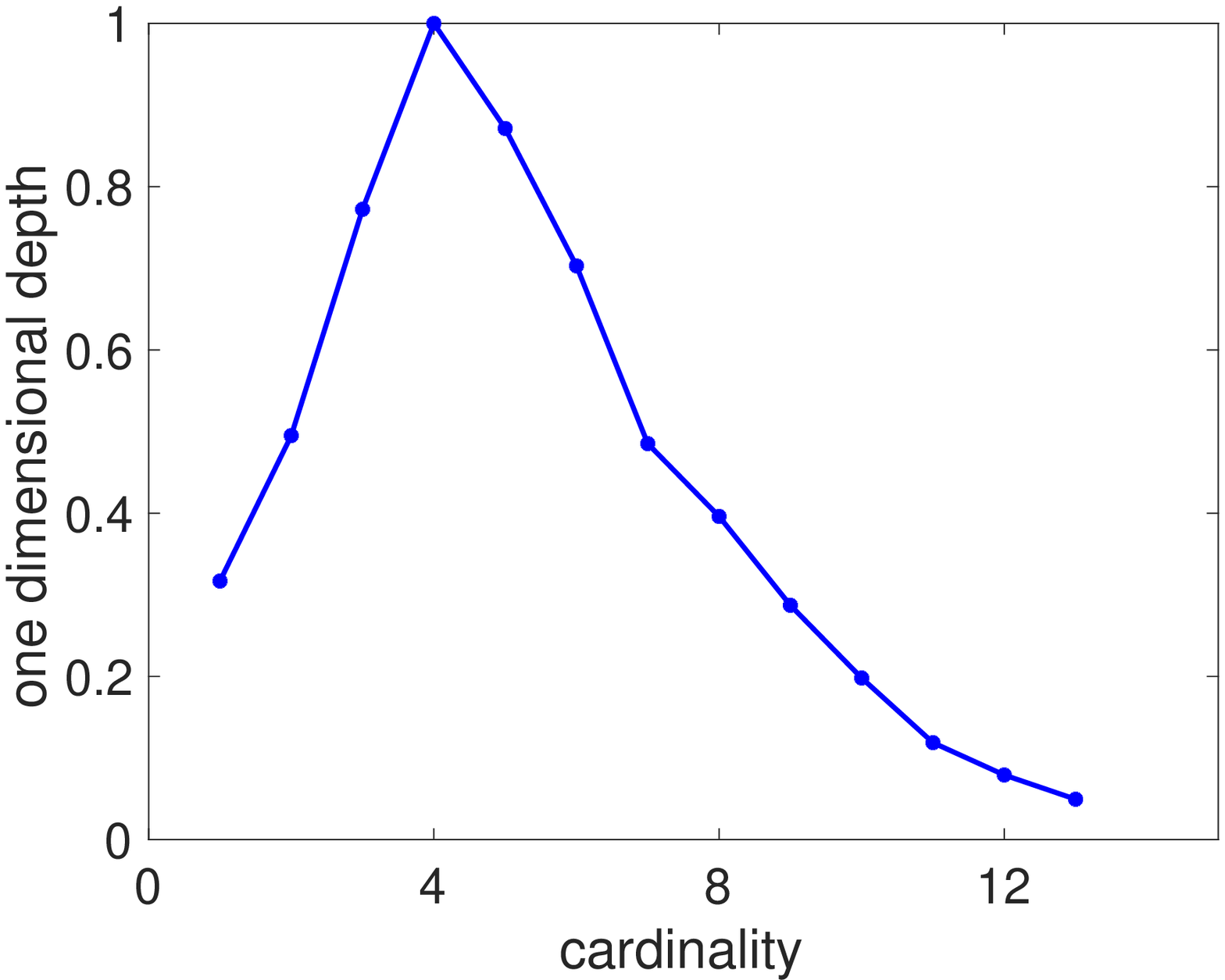}
		\end{minipage}
	}
    	\subfigure[one dimensional depth for highway]{
    		\begin{minipage}[b]{0.47\textwidth}
			\centering
   		 	\includegraphics[scale=0.3]{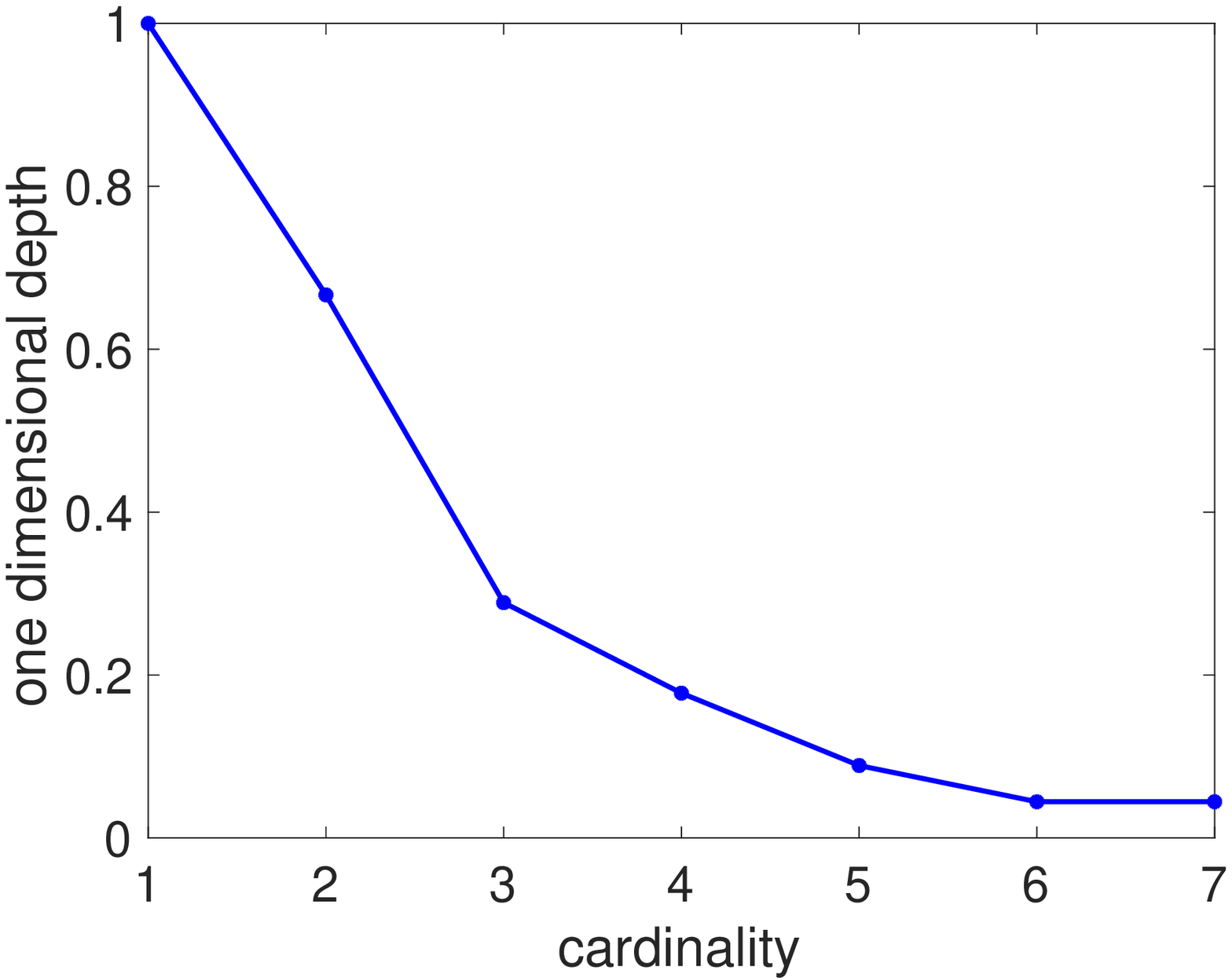}
    		\end{minipage}
		\label{fig:highway_top10}
    	}
	\\
	\subfigure[depth ranking of local roads]{
		\begin{minipage}[b]{0.47\textwidth}
			\centering
			\includegraphics[scale=0.3]{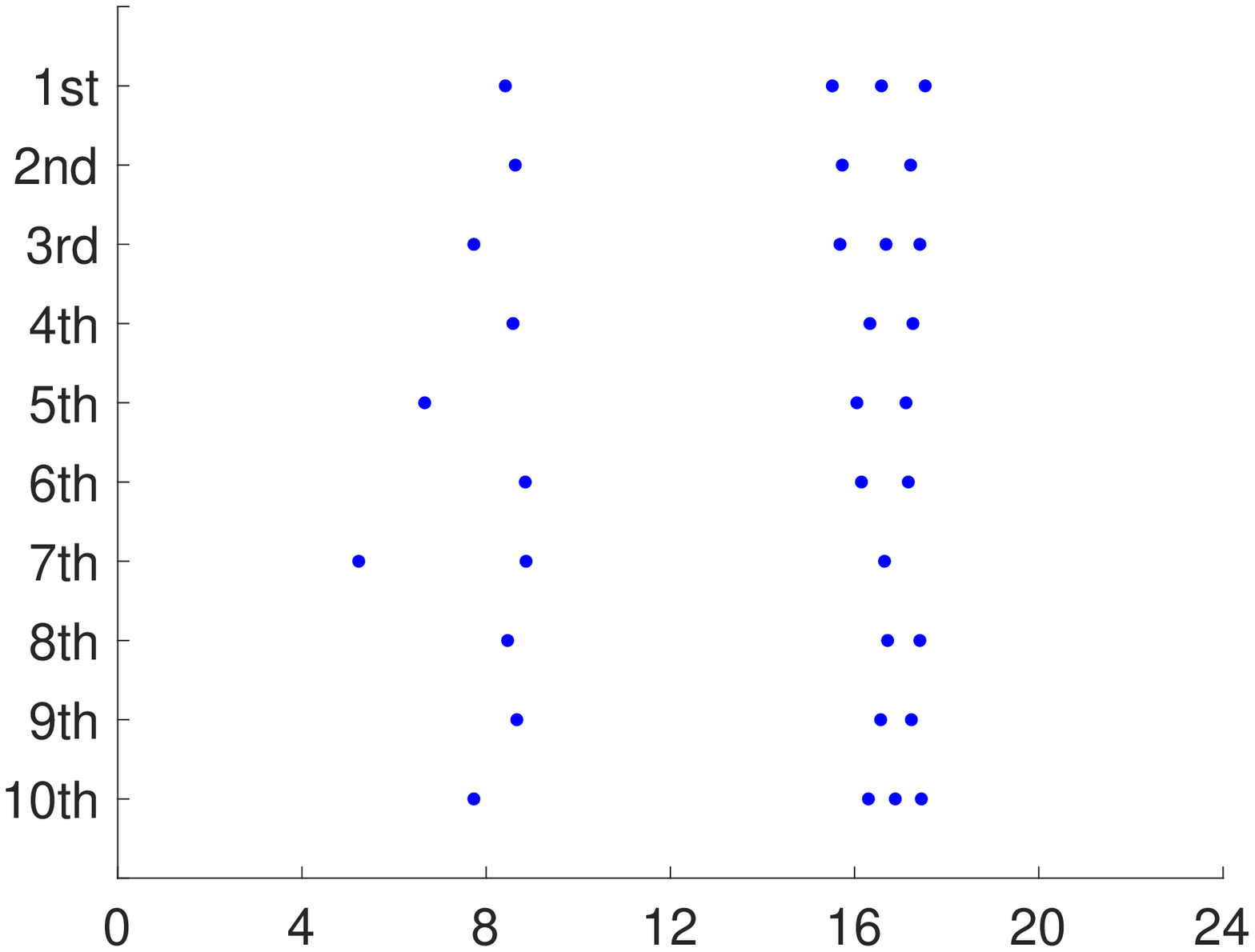}
		\end{minipage}
		\label{fig:local_top10}
	}
    	\subfigure[depth ranking of highway]{
    		\begin{minipage}[b]{0.47\textwidth}
			\centering
   		 	\includegraphics[scale=0.3]{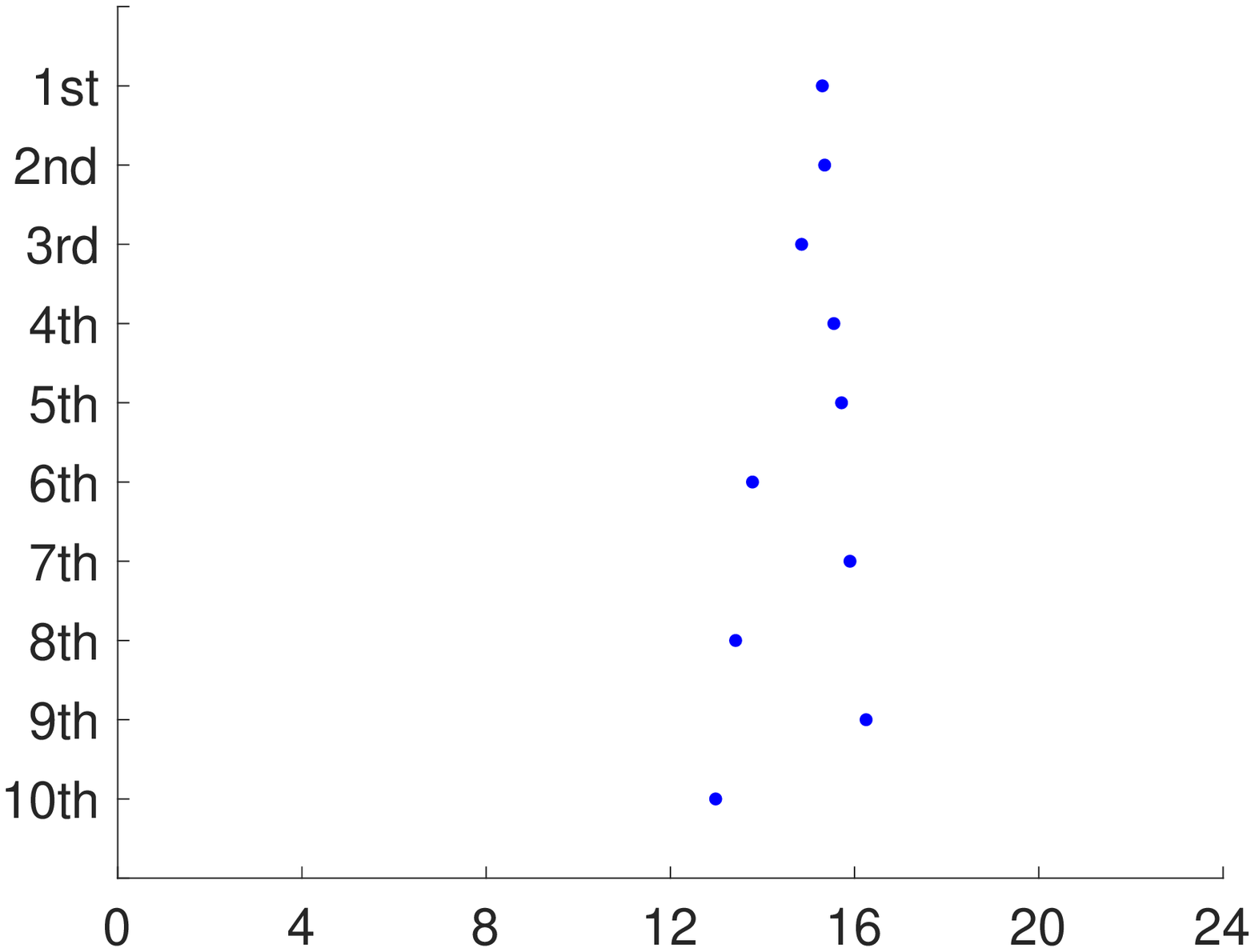}
    		\end{minipage}
		\label{fig:highway_top10}
    	}
	\caption{One dimensional depth and processes with top depth ranks. (a) One dimensional depth in Definition \ref{def:wholedef} for different number of time events for local roads sample. (b) Same as (a) except for highway sample. (c) Each row is a real case of one day car accidents occurred in local roads of Tallahassee with the ranking of depth value shown on vertical axis. The depth is computed by Definition \ref{def:wholedef}, where the ILR depth is used as the conditional depth and the hyper-parameter $r=1$. (d) Same as (c) except for accidents occurred in the highway. }
	\label{fig:car_top10}
\end{figure}

Similar to the previous simulation studies, the realizations with top 10 depth values can be collected for each group (with hyper-parameter $r=1$), and the result is shown in Fig. \ref{fig:car_top10}. Based on the estimated intensity function, the expected numbers of time events in a single realization are 4.25 and 1.38 for local road and highway, respectively, which means the frequency of realizations with 4 and 1 time events may be relatively larger in the local roads sample and highway sample, respectively. From Fig. \ref{fig:car_top10}(a), one can find the one dimensional depth in Definition \ref{def:wholedef} obtains larger values when the cardinality is around 3, 4 and 5 for local roads. This result can also be seen in Fig. \ref{fig:car_top10}(c) where the depths of realizations with 3 or 4 time events are ranked in the top 10. Similarly, from Fig. \ref{fig:car_top10}(b), one can find the one dimensional depth obtains a significant larger value when the cardinality is 1 for highway, which can also be seen in  \ref{fig:car_top10}(d) where  the depths of realizations with 1 time event are ranked in the top 10. Moreover, we can see that the patterns of the top 10 depth-valued realizations in local roads and highway perfectly match the estimated conditional intensity functions, respectively. That is, for local roads, the high rank processes have two or three events during the afternoon rush hour (4-6pm), and one event during the morning rush hour (8-9am).  For highway, there is only one event in top processes at around 4pm.  In summary, the ILR depth method provides a reasonable center-outward rank to summarize the given point process data.


\section{Summary and Future Work} \label{sec:future}
In this paper, we have introduced a new framework to define depth for point process. The definition can be divided into two parts: (1) One dimensional depth that characterizes the depth for point process cardinality. (2) The ILR depth that measures the center-outward ranks for location of time events conditioned on cardinality. In the formal definition, we adopt the approach proposed in \citet{qi2021dirichlet} to define the one dimensional depth and we focus on the derivation of the ILR depth. We at first develop the ILR depth for homogeneous Poisson process.  The approach is based on the ILR transformation of the inter-event time from a simplex space to a Euclidean space, and the depth is defined using the probability distribution of the ILR transformed data. We also conduct a thorough study on the mathematical properties of the ILR depth, in which we introduce a new notion called \textit{simplicial center} to satisfy the desirable properties. In addition, we introduce a simplified version of the ILR depth based on Gaussian approximation and examine its mathematical properties. We then extend the ILR depth to non-Poisson process via the time rescaling method. We examine all mathematical properties and conduct the asymptotic study when the process is a Poisson process. We also compare the ILR depth with two previous mehods in literature, the generalized Mahalanobis depth \citep{liu2017generalized} and the Dirichlet depth \citep{qi2021dirichlet}, and demonstrate the superiority of the new mehod. Finally, we use a real world dataset on car accidents to show the effectiveness of our new framework. 

The ILR depth is a novel and rigorous approach to define depth of point process conditioned on cardinality. The apparent advantage is that the ILR depth is a distribution-based depth that can capture the real data pattern and provide geometric interpretation in Euclidean space. There are clear topics for us to further investigate in the future.  At first, the density function in Eqn. \eqref{eq:density} is under-explored, and more work can be conducted to examine its shape in Euclidean space when the norm of $\bm{u}^*$ gets larger. Moreover, difference between the density in Eqn. \eqref{eq:density} and Gaussian density can be studied via the Kullback-Leibler divergence or Fisher-Rao distance \citep{srivastava2011registration} to obtain a more comprehensive understanding of the density function of $\bm{u}^*$. Finally, other than the time rescaling, we can explore more generalization methods of the ILR depth from homogeneous Poisson process to general point process. 

\newpage

\bibliographystyle{apalike}
\bibliography{references}

\newpage 

\appendix

\section{Derivation about the pdf of ILR transformation of IET} \label{app:der}
In Eqn. \eqref{eq:ilrpdf}, if the cardinality of a given homogeneous Poisson process is $k$, to figure out the kernel of the density function, $\det(J)$ need to be solved in terms of $\bm{u}^*$. For any $j, s=1,\dots,k$, the element of $J$ is: 
\begin{eqnarray*}
\frac{du_{j}}{du_{s}^{*}}=&(T_2-T_1)\cdot\frac{\Psi_{s,j}\exp(\sum_{i=1}^{k}u_{i}^{*}\Psi_{i,j})(\sum_{p=1}^{k+1}\exp(\sum_{i=1}^{k}u_{i}^{*}\Psi_{i,p}))}{(\sum_{p=1}^{k+1}\exp(\sum_{i=1}^{k}u_{i}^{*}\Psi_{i,p}))^{2}} \\
&-(T_2-T_1)\cdot\frac{\exp(\sum_{i=1}^{k}u_{i}^{*}\Psi_{i,j})(\sum_{p=1}^{k+1}\Psi_{s,p}\exp(\sum_{i=1}^{k}u_{i}^{*}\Psi_{i,p}))}{(\sum_{p=1}^{k+1}\exp(\sum_{i=1}^{k}u_{i}^{*}\Psi_{i,p}))^2}
\end{eqnarray*}
Therefore, $\det(J)$ can be expressed as follows. To save space, denote $A_{i,j,p}=\Psi_{i,j}e^{\sum_{i=1}^{k}u_{i}^{*}\Psi_{i,p}}$. 

\begin{eqnarray*}
\det(J)&=&\frac{(T_2-T_1)^k}{(\sum_{p=1}^{k+1}e^{\sum_{i=1}^{k}u_{i}^{*}\Psi_{i,p}})^{2k}}\times \\
    &\;&\begin{vmatrix}
e^{\sum_{i=1}^{k}u_{i}^{*}\Psi_{i,1}}\Big(\sum_{p=1}^{k+1}(A_{1,1,p}-A_{1,p,p})\Big) & \dots & e^{\sum_{i=1}^{k}u_{i}^{*}\Psi_{i,1}}\Big(\sum_{p=1}^{k+1}(A_{k,1,p}-A_{k,p,p})\Big)\\
e^{\sum_{i=1}^{k}u_{i}^{*}\Psi_{i,2}}\Big(\sum_{p=1}^{k+1}(A_{1,2,p}-A_{1,p,p})\Big)  & \dots & e^{\sum_{i=1}^{k}u_{i}^{*}\Psi_{i,2}}\Big(\sum_{p=1}^{k+1}(A_{k,2,p}-A_{k,p,p})\Big)\\ \vdots & \ddots & \vdots\\
e^{\sum_{i=1}^{k}u_{i}^{*}\Psi_{i,k}}\Big(\sum_{p=1}^{k+1}(A_{1,k,p}-A_{1,p,p})\Big)  & \dots & e^{\sum_{i=1}^{k}u_{i}^{*}\Psi_{i,k}}\Big(\sum_{p=1}^{k+1}(A_{k,k,p}-A_{k,p,p})\Big)

\end{vmatrix} \\
&=&\frac{(T_2-T_1)^k\prod_{p=1}^{k}e^{\sum_{i=1}^{k}u_{i}^{*}\Psi_{i,p}}}{(\sum_{p=1}^{k+1}e^{\sum_{i=1}^{k}u_{i}^{*}\Psi_{i,p}})^{2k}}
\begin{vmatrix}
\sum_{p=1}^{k+1}(A_{1,1,p}-A_{1,p,p}) & \dots &
\sum_{p=1}^{k+1}(A_{k,1,p}-A_{k,p,p})\\
\sum_{p=1}^{k+1}(A_{1,2,p}-A_{1,p,p}) & \dots &\sum_{p=1}^{k+1}(A_{k,2,p}-A_{k,p,p}) \\ \vdots & \ddots & \vdots\\
\sum_{p=1}^{k+1}(A_{1,k,p}-A_{1,p,p}) & \dots &
\sum_{p=1}^{k+1}(A_{k,k,p}-A_{k,p,p})
\end{vmatrix} \\
&=&\frac{(T_2-T_1)^k\prod_{p=1}^{k}e^{\sum_{i=1}^{k}u_{i}^{*}\Psi_{i,p}}}{(\sum_{p=1}^{k+1}e^{\sum_{i=1}^{k}u_{i}^{*}\Psi_{i,p}})^{2k}}
\begin{vmatrix}
\sum_{p=1}^{k+1}(A_{1,1,p}-A_{1,p,p}) & \dots &
\sum_{p=1}^{k+1}(A_{k,1,p}-A_{k,p,p})\\
\sum_{p=1}^{k+1}(A_{1,2,p}-A_{1,1,p}) & \dots &
\sum_{p=1}^{k+1}(A_{k,2,p}-A_{k,1,p})\\ \vdots & \ddots & \vdots\\
\sum_{p=1}^{k+1}(A_{1,k,p}-A_{1,1,p}) & \dots &
\sum_{p=1}^{k+1}(A_{k,k,p}-A_{k,1,p})
\end{vmatrix}\\
&=&\frac{(T_2-T_1)^k\prod_{p=1}^{k}e^{\sum_{i=1}^{k}u_{i}^{*}\Psi_{i,p}}}{(\sum_{p=1}^{k+1}e^{\sum_{i=1}^{k}u_{i}^{*}\Psi_{i,p}})^{k+1}}
\begin{vmatrix}
\sum_{p=1}^{k+1}(A_{1,1,p}-A_{1,p,p}) & \dots &
\sum_{p=1}^{k+1}(A_{k,1,p}-A_{k,p,p})\\
\Psi_{1,2}-\Psi_{1,1} & \dots &
\Psi_{k,2}-\Psi_{k,1}\\ \vdots & \ddots & \vdots\\
\Psi_{1,k}-\Psi_{1,1} & \dots &
\Psi_{k,k}-\Psi_{k,1}
\end{vmatrix} \\
&=&\frac{(T_2-T_1)^k\prod_{p=1}^{k}e^{\sum_{i=1}^{k}u_{i}^{*}\Psi_{i,p}}}{(\sum_{p=1}^{k+1}e^{\sum_{i=1}^{k}u_{i}^{*}\Psi_{i,p}})^{k+1}}\Bigg(\sum_{p=1}^{k+1}
\begin{vmatrix}
A_{1,1,p}-A_{1,p,p}& \dots &
A_{k,1,p}-A_{k,p,p}\\
\Psi_{1,2}-\Psi_{1,1} & \dots &
\Psi_{k,2}-\Psi_{k,1}\\ \vdots & \ddots & \vdots\\
\Psi_{1,k}-\Psi_{1,1} & \dots &
\Psi_{k,k}-\Psi_{k,1}
\end{vmatrix}\Bigg) \\
&=&\frac{(T_2-T_1)^k\prod_{p=1}^{k}e^{\sum_{i=1}^{k}u_{i}^{*}\Psi_{i,p}}}{(\sum_{p=1}^{k+1}e^{\sum_{i=1}^{k}u_{i}^{*}\Psi_{i,p}})^{k+1}}\Bigg(\sum_{p=1}^{k+1}e^{\sum_{i=1}^{k}u_{i}^{*}\Psi_{i,p}}
\begin{vmatrix}
\Psi_{1,1}-\Psi_{1,p} & \dots & 
\Psi_{k,1}-\Psi_{k,p}\\
\Psi_{1,2}-\Psi_{1,1} & \dots &
\Psi_{k,2}-\Psi_{k,1}\\ \vdots & \ddots & \vdots\\
\Psi_{1,k}-\Psi_{1,1} & \dots &
\Psi_{k,k}-\Psi_{k,1}
\end{vmatrix}\Bigg) \\
\end{eqnarray*}
Next, denote $D_{p}=\begin{vmatrix}
\Psi_{1,1}-\Psi_{1,p} & \dots & 
\Psi_{k,1}-\Psi_{k,p}\\
\Psi_{1,2}-\Psi_{1,1} & \dots &
\Psi_{k,2}-\Psi_{k,1}\\ \vdots & \ddots & \vdots\\
\Psi_{1,k}-\Psi_{1,1} & \dots &
\Psi_{k,k}-\Psi_{k,1}
\end{vmatrix}$ for $p=1,\dots,k$. If $p=1$, 
\begin{eqnarray*}
    D_{1}=\begin{vmatrix}
0 & \dots & 
0\\
\Psi_{1,2}-\Psi_{1,1} & \dots &
\Psi_{k,2}-\Psi_{k,1}\\ \vdots & \ddots & \vdots\\
\Psi_{1,k}-\Psi_{1,1} & \dots &
\Psi_{k,k}-\Psi_{k,1}
\end{vmatrix}=0
\end{eqnarray*}
If $p=2,\dots,k$, when calculating $D_{p}$, the first row can be added with the $p$-th row, then the first row will become $0$ and the determinant is unchanged. Thus, if $p=2,\dots,k$, 
\begin{eqnarray*}
D_{p}=\begin{vmatrix}
0 & \dots & 
0\\
\Psi_{1,2}-\Psi_{1,1} & \dots &
\Psi_{k,2}-\Psi_{k,1}\\ \vdots & \ddots & \vdots\\
\Psi_{1,k}-\Psi_{1,1} & \dots &
\Psi_{k,k}-\Psi_{k,1}
\end{vmatrix}=0
\end{eqnarray*}
Therefore, apply this result to the formula of $\det(J)$, 
\begin{eqnarray*}
\det(J)&=&\frac{(T_2-T_1)^k\prod_{p=1}^{k}e^{\sum_{i=1}^{k}u_{i}^{*}\Psi_{i,p}}}{(\sum_{p=1}^{k+1}e^{\sum_{i=1}^{k}u_{i}^{*}\Psi_{i,p}})^{k+1}}\Bigg(e^{\sum_{i=1}^{k}u_{i}^{*}\Psi_{i,k+1}}
\begin{vmatrix}
\Psi_{1,1}-\Psi_{1,k+1} & \dots & 
\Psi_{k,1}-\Psi_{k,k+1}\\
\Psi_{1,2}-\Psi_{1,1} & \dots &
\Psi_{k,2}-\Psi_{k,1}\\ \vdots & \ddots & \vdots\\
\Psi_{1,k}-\Psi_{1,1} & \dots &
\Psi_{k,k}-\Psi_{k,1}
\end{vmatrix}\Bigg) \\
&=&\frac{(T_2-T_1)^k\prod_{p=1}^{k+1}e^{\sum_{i=1}^{k}u_{i}^{*}\Psi_{i,p}}}{(\sum_{p=1}^{k+1}e^{\sum_{i=1}^{k}u_{i}^{*}\Psi_{i,p}})^{k+1}}
\begin{vmatrix}
\Psi_{1,1}-\Psi_{1,k+1} & \dots & 
\Psi_{k,1}-\Psi_{k,k+1}\\
\Psi_{1,2}-\Psi_{1,1} & \dots &
\Psi_{k,2}-\Psi_{k,1}\\ \vdots & \ddots & \vdots\\
\Psi_{1,k}-\Psi_{1,1} & \dots &
\Psi_{k,k}-\Psi_{k,1}
\end{vmatrix} \\
&=&\frac{(T_2-T_1)^ke^{\sum_{i=1}^{k}u_{i}^{*}(\sum_{p=1}^{k+1}\Psi_{i,p})}}{(\sum_{p=1}^{k+1}e^{\sum_{i=1}^{k}u_{i}^{*}\Psi_{i,p}})^{k+1}} 
\begin{vmatrix}
\Psi_{1,1}-\Psi_{1,k+1} & \dots & 
\Psi_{k,1}-\Psi_{k,k+1}\\
\Psi_{1,2}-\Psi_{1,1} & \dots &
\Psi_{k,2}-\Psi_{k,1}\\ \vdots & \ddots & \vdots\\
\Psi_{1,k}-\Psi_{1,1} & \dots &
\Psi_{k,k}-\Psi_{k,1}
\end{vmatrix} \\
&=&\frac{(T_2-T_1)^k}{(\sum_{p=1}^{k+1}e^{\sum_{i=1}^{k}u_{i}^{*}\Psi_{i,p}})^{k+1}}
\begin{vmatrix}
\Psi_{1,1}-\Psi_{1,k+1} & \dots & 
\Psi_{k,1}-\Psi_{k,k+1}\\
\Psi_{1,2}-\Psi_{1,1} & \dots &
\Psi_{k,2}-\Psi_{k,1}\\ \vdots & \ddots & \vdots\\
\Psi_{1,k}-\Psi_{1,1} & \dots &
\Psi_{k,k}-\Psi_{k,1}
\end{vmatrix}
\end{eqnarray*}
Therefore, the pdf of the ILR transformation for k-event HPP is: 
\begin{eqnarray*}
    f_{\bm{u^{*}}}(u_{1}^{*},\dots u_{k}^{*})&=&\frac{k!}{(\sum_{p=1}^{k+1}e^{\sum_{i=1}^{k}u_{i}^{*}\Psi_{i,p}})^{k+1}} 
    \begin{vmatrix}
\Psi_{1,1}-\Psi_{1,k+1} & \dots & 
\Psi_{k,1}-\Psi_{k,k+1}\\
\Psi_{1,2}-\Psi_{1,1} & \dots &
\Psi_{k,2}-\Psi_{k,1}\\ \vdots & \ddots & \vdots\\
\Psi_{1,k}-\Psi_{1,1} & \dots &
\Psi_{k,k}-\Psi_{k,1}
\end{vmatrix}
    \\
    &=& \frac{c}{(\sum_{p=1}^{k+1}e^{\sum_{i=1}^{k}u_{i}^{*}\Psi_{i,p}})^{k+1}}
\end{eqnarray*}
where $c$ is the positive constant that guarantees the integral of the density as $1$.

\section{Proof of log-concavity of the density in Eqn. \eqref{eq:density}} \label{app:logcon}
By the property of probability density function, if the density is log-concave, then it is a uni-modal shape curve. Therefore, the remaining task is to prove the density function in Eqn. \eqref{eq:density} is log-concave for any positive integer $k$.  

To prove $\log(f_{\bm{u^{*}}}(u_{1}^{*},\dots u_{k}^{*}))$ is a concave function, the Hessian matrix need to be found in closed form. In this way, the first and second order partial derivative of $\log(f_{\bm{u^{*}}}(u_{1}^{*},\dots u_{k}^{*}))$ can be computed as follows: 
\begin{eqnarray*}
\frac{\partial}{\partial u_{s}^{*}}\log(f_{\bm{u^{*}}}) &=&-(k+1)\cdot \frac{\sum_{p=1}^{k+1}\Psi_{s,p}e^{\sum_{i=1}^{k}u_{i}^{*}\Psi_{i,p}}}{\sum_{p=1}^{k+1}e^{\sum_{i=1}^{k}u_{i}^{*}\Psi_{i,p}}} \\
\frac{\partial^2}{\partial {u_{s}^*}^2}\log(f_{\bm{u^{*}}}) &=& -(k+1)\cdot \frac{\sum_{p=1}^{k}\sum_{q>p}(\Psi_{s,p}-\Psi_{s,q})^2e^{\sum_{i=1}^{k}u_{i}^{*}\Psi_{i,p}}e^{\sum_{i=1}^{k}u_{i}^{*}\Psi_{i,q}}}{\big(\sum_{p=1}^{k+1}e^{\sum_{i=1}^{k}u_{i}^{*}\Psi_{i,p}}\big)^2} \\
\frac{\partial^2}{\partial u_{s}^{*}du_{t}^{*}}\log(f_{\bm{u^{*}}})&=&-(k+1)\cdot \frac{\sum_{p=1}^{k}\sum_{q>p}(\Psi_{s,p}-\Psi_{s,q})(\Psi_{t,p}-\Psi_{t,q})e^{\sum_{i=1}^{k}u_{i}^{*}\Psi_{i,p}}e^{\sum_{i=1}^{k}u_{i}^{*}\Psi_{i,q}}}{\big(\sum_{p=1}^{k+1}e^{\sum_{i=1}^{k}u_{i}^{*}\Psi_{i,p}}\big)^2}
\end{eqnarray*}
Next, denote the Hessian matrix $H$ as: 
\begin{eqnarray*}
H=\begin{pmatrix}
\frac{\partial^2}{\partial {u_{1}^*}^2}\log(f_{\bm{u^{*}}}) & \frac{\partial^2}{\partial u_{1}^{*}\partial u_{2}^{*}}\log(f_{\bm{u^{*}}}) & \dots &\frac{\partial^2}{\partial u_{1}^{*}\partial u_{k}^{*}}\log(f_{\bm{u^{*}}}) \\
\frac{\partial^2}{\partial u_{2}^{*}\partial u_{1}^{*}}\log(f_{\bm{u^{*}}}) & \frac{\partial^2}{\partial {u_{2}^*}^2}\log(f_{\bm{u^{*}}}) & \dots & \frac{\partial^2}{\partial u_{2}^{*}\partial u_{k}^{*}}\log(f_{\bm{u^{*}}}) \\ \vdots & \vdots & \ddots & \vdots \\
\frac{\partial^2}{\partial u_{k}^{*}\partial u_{1}^{*}}\log(f_{\bm{u^{*}}}) & \frac{\partial^2}{\partial {u_{k}^*}\partial u_{2}^{*}}\log(f_{\bm{u^{*}}}) & \dots & \frac{\partial^2}{\partial {u_{k}^*}^2}\log(f_{\bm{u^{*}}})    
\end{pmatrix}
\end{eqnarray*}
Due to the property of concavity, a multivariate function is concave if and only if its Hessian matrix is negative definite. Denote $A_{s,t}=\sqrt{e^{\sum_{i=1}^{k}u_{i}^{*}\Psi_{i,s}}e^{\sum_{i=1}^{k}u_{i}^{*}\Psi_{i,t}}}$ and a $(k\times \frac{k(k+1)}{2})$ matrix $B$ as: 
\begin{eqnarray*}
    B=
    \begin{pmatrix}
    (\Psi_{1,1}-\Psi_{1,2})A_{1,2} & (\Psi_{1,1}-\Psi_{1,3})A_{1,3}
     & \dots & (\Psi_{1,k}-\Psi_{1,k+1})A_{k,k+1} \\ 
    (\Psi_{2,1}-\Psi_{2,2})A_{1,2}  &   (\Psi_{2,1}-\Psi_{2,3})A_{1,3} 
    &  \dots & (\Psi_{2,k}-\Psi_{2,k+1})A_{k,k+1} \\ \vdots &  \vdots & \ddots & \vdots \\
    (\Psi_{k,1}-\Psi_{k,2})A_{1,2} & (\Psi_{k,1}-\Psi_{k,3})A_{1,3}
    & \dots & (\Psi_{k,k}-\Psi_{k,k+1})A_{k,k+1}
    \end{pmatrix}
\end{eqnarray*}
After some algebra, $H$ can be expressed as: 
\begin{eqnarray*}
H=-\frac{k+1}{\big(\sum_{p=1}^{k+1}e^{\sum_{i=1}^{k}u_{i}^{*}\Psi_{i,p}}\big)^2}\cdot BB^{T}
\end{eqnarray*}
Since $BB^{T}$ is a positive definite matrix, $H$ is negative definite. Thus, the pdf in Eqn. \eqref{eq:density} is log-concave and uni-modal. Finally, when $\bm{u}^*=(0,0,\dots,0)^T$, take into account that the sum of each row of $\Psi$ is $0$, then, the first partial derivative $\frac{\partial}{\partial u_{s}^{*}}\log(f_{\bm{u^{*}}})$ equals to $0$ for each $s=1,2,\dots,k$. Therefore, the origin in Euclidean space $\mathbb{R}^k$ is the global maximum point of the density in Eqn. \eqref{eq:density}. 

\section{Proof of contour} \label{app:contour}
For the density in Eqn. \eqref{eq:density}, define the contour as $\sum_{p=1}^{k+1}e^{\sum_{i=1}^{k}u_{i}^{*}\Psi_{i,p}}=c$, where $c$ is a positive constant. Let $\bm{u}^*=(u_1^*,u_2^*,\dots,u_k^*)^T$ as a $k$ dimensional column vector, then the contour becomes $\sum_{p=1}^{k+1}e^{{\bm{u}^*}^T\Psi_{:,p}}=c$. If $\lVert\bm{u}^*\rVert$ is small, consider Taylor expansion up to the second order term, the contour has the approximated form: 
\begin{eqnarray*}
 & & \sum_{p=1}^{k+1}\big(1+{\bm{u}^*}^T\Psi_{:,p}+\frac{1}{2}({\bm{u}^*}^T\Psi_{:,p})^2\big) = c \\
\text{(Property of }\Psi)&\Longrightarrow&  \sum_{p=1}^{k+1}\big({\bm{u}^*}^T\Psi_{:,p}\big)^2 = c_1 \ \text{(another constant)} \\
&\Longrightarrow&  \sum_{p=1}^{k+1}\Big({\bm{u}^*}^T\Psi_{:,p}\Psi_{:,p}^T\bm{u}^*\Big) = c_1 \\
&\Longrightarrow&  {\bm{u}^*}^T\Big(\sum_{p=1}^{k+1}\Psi_{:,p}\Psi_{:,p}^T\Big)\bm{u}^*= c_1 \\
\text{(Property of }\Psi)&\Longrightarrow& {\bm{u}^*}^T\bm{u}^*=c_1
\end{eqnarray*}
Therefore, the contour becomes the formula of hyper-sphere in Euclidean space.

\section{Proof of normal approximation result} \label{app:gau}
According to Proposition \ref{prop:logconcativity}, the global maximum point of the density in Eqn. \eqref{eq:density} is the origin. Thus, with Taylor series expansion, $\log(f_{\bm{u}^*})$ can be rewritten as: 
\begin{eqnarray}
\log(f_{\bm{u}^*}(u_1^*,\dots,u_k^*))\approx\log(f_{\bm{u}^*}(0,0,\dots,0))+D^T\bm{u}^*-\frac{1}{2}{\bm{u}^*}^TH\bm{u}^* \label{eq:tay}
\end{eqnarray}
where $D$ is the first derivative of $\log(f_{\bm{u}^*})$ evaluated at origin and $H$ is the negative of Hessian matrix of $\log(f_{\bm{u}^*})$ evaluated at origin. Considering that origin is the global maximum point, $D$ is a column vector with all entries $0$ and the second term at the right hand side of Eqn. \eqref{eq:tay} can be omitted. The remaining task is to figure out the closed form of $H$. According to Appendix \ref{app:logcon}, given the properties of the matrix $\Psi$ in ILR transformation such that $\sum_{p=1}^{k}\Psi_{s,p}=0$, $\sum_{p=1}^{k}\Psi_{s,p}^2=1$ and $\sum_{p=1}^{k}\Psi_{s,p}\Psi_{t,p}=0$ for each $s,t=1,\dots,k$, the following result can be easily obtained for each $s,t=1,\dots,k$. 
\begin{eqnarray*}
-\frac{\partial^2}{\partial {u_{s}^*}^2}\log(f_{\bm{u^{*}}})(0,0,\dots,0)&=&(k+1)\frac{(k+1)\sum_{p=1}^{k+1}\Psi_{s,p}^2-(\sum_{p=1}^{k+1}\Psi_{s,p})^2}{(k+1)^2} \\
&=& 1 \\
-\frac{\partial^2}{\partial u_{s}^{*}du_{t}^{*}}\log(f_{\bm{u^{*}}})(0,0,\dots,0) &=& (k+1)\frac{(k+1)\sum_{p=1}^{k+1}\Psi_{s,p}\Psi_{t,p}-(\sum_{p=1}^{k+1}\Psi_{s,p})(\sum_{p=1}^{k+1}\Psi_{t,p})}{(k+1)^2} \\
&=& 0
\end{eqnarray*}
Therefore, $H$ is a $k\times k$ identity matrix. Finally, take exponential on both side of Eqn. \eqref{eq:tay}, the result is: 
\begin{eqnarray*}
f_{\bm{u}^*}(u_1^*,\dots,u_k^*) \approx C\cdot e^{-\frac{1}{2}{\bm{u}^*}^TH\bm{u}^*}
\end{eqnarray*}
where $C$ is a positive constant and $e^{-\frac{1}{2}{\bm{u}^*}^TH\bm{u}^*}$ is the kernel of a multivariate normal distribution with mean as the origin and covariance matrix as the inverse of $H$, which is the $k\times k$ identity matrix.

\section{Proof of mathematical properties of ILR depth} \label{app:mat}
\begin{enumerate}
\item Based on Definition \ref{def:formal}, this part is trivial. 
\item Based on Proposition \ref{prop:center}, this part is trivial. 
\item Denote the center of ILR depth as $\bm{u}_c^*$ in $\mathbb{R}^k$. Based on Proposition \ref{prop:logconcativity}, the density of $\bm{u}^*$ is log-concave in $\mathbb{R}^k$. Since the contour of ILR depth takes the same shape as the density of $\bm{u}^*$, denote $D_c(\bm{u}^*)$ as the ILR depth function for any $\bm{u}^*\in\mathbb{R}^k$, it is easy to verify that $D_c(\bm{u}^*)\leq D_c(\bm{u}_c^*+\alpha(\bm{u}^*-\bm{u}_c^*))$ for any $\bm{u}^*\in\mathbb{R}^k$ and $\alpha\in [0,1]$. 
\item According to Eqn. \eqref{eq:ilrdepth}, ILR depth remains invariant to scaling and translation. Thus, this part is verified. 
\end{enumerate}

\section{Proof of mathematical properties of the simplified version of ILR depth in Definition \ref{def:opt}} \label{app:gaumat}
\begin{enumerate}
\item Based on the definition, the depth value is continuous if $\bm{s}\notin\mathbb{B}_k$. Thus, the remaining task is to prove the depth function is continuous at boundary set $\mathbb{B}_k$, which is equivalent to prove the depth value approaches $0$ if the point process approaches $\mathbb{B}_k$. For a given $\bm{s}=(s_1,s_2,\dots,s_k)\in\mathbb{S}_k$, if there exists at least one $t=1,\dots,k+1$ such that $s_t-s_{t-1}\to 0$, one can find at least one $p=1,\dots,k+1$ such that $s_p-s_{p-1}\neq 0$ and is a finite positive number. In this case, the depth can be rewritten as: 
\begin{eqnarray*}
D_{c}(\bm{s};P_{S||S|=k}) &=&  \frac{1}{1+\frac{1}{2}\sum_{i=1}^{k+1}\Big(\log\frac{s_i-s_{i-1}}{(\prod_{j=1}^{k+1}(s_j-s_{j-1}))^{\frac{1}{k+1}}}\Big)^2} \\
& = & \frac{1}{1+\frac{1}{2}\sum_{i=1}^{k+1}\Big(\log\big((s_i-s_{i-1})^{\frac{k}{k+1}}\cdot(\prod_{j\neq i}(s_j-s_{j-1}))^{-\frac{1}{k+1}}\big)\Big)^2}
\end{eqnarray*}
According to the notation above, $(s_p-s_{p-1})^{\frac{k}{k+1}}\cdot(\prod_{j\neq p}(s_j-s_{j-1}))^{-\frac{1}{k+1}}\to\infty$ since $\prod_{j\neq p}(s_j-s_{j-1})\to 0$. Consider the fact that the denominator part is the sum of $1$ and $k+1$ positive terms, if one term approaches infinity, the denominator will approach infinity and this part is verified. 
\item Consider the contour of standard multivariate Gaussian density function, origin is the center based on all classical symmetries \citep{zuo2000general}. According to Proposition \ref{prop:center}, $\bm{s}=(T_1+\frac{T_2-T_1}{k+1},T_1+\frac{2(T_2-T_1)}{k+1},\dots,T_1+\frac{k(T_2-T_1)}{k+1})$ is the center. 
\item This part is trivial due to the shape of contour of normal distribution and Appendix \ref{app:mat}. 
\item This part is trivial. 
\end{enumerate}

\section{Proof of mathematical properties for IPP} \label{app:matipp}
\begin{enumerate}
\item Based on \citet{qi2021dirichlet}, $\Lambda_S(\cdot)$ is a continuous function for any general point process. Thus, the continuity will hold automatically. What is more, if there exists $i=1,2,\dots,k+1$ such that $u_i\to 0$, then, no matter whether the conditional intensity function $\lambda(\cdot)$ is deterministic or not,  $u_i'=\Lambda_S(s_i)-\Lambda_S(s_{i-1})=\int_{T_1}^{s_i}\lambda(t|H_t)dt-\int_{T_1}^{s_{i-1}}\lambda(t|H_t)dt=\int_{s_{i-1}}^{s_i}\lambda(t|H_t)dt \to 0$. Therefore, from the proof of Proposition \ref{prop:ilrdepmat}, the depth value will vanish at boundary. 
\item \label{p2} If the process is inhomogeneous Poisson process, the conditional intensity function can be considered as a positive deterministic function $\lambda(\cdot)$. Based on the definition of $\Lambda_S(\cdot)$, $\Lambda_S(\cdot)$ is a strict increasing function, and therefore is bijective function. Then, based on the proof of Proposition \ref{prop:ilrdepmat}, it is easy to verify this property. 
\item This part is similar to property \ref{p2} and can be omitted. 
\item This part is trivial. 
\end{enumerate}

\section{Proof of uniform convergence rate of Lemma \ref{lemma:uniform}} \label{app:uniformproof}
First, rewrite $\hat{\lambda}(t)-\lambda(t)=\big(\hat{\lambda}(t)-\mathbb{E}[\hat{\lambda}(t)]\big)+\big(\mathbb{E}[\hat{\lambda}(t)]-\lambda(t)\big)$. We will consider the second part first, denote $N_i(t)$ as the number of events occurred until time $t$ in realization $i$, then $\mathbb{E}[\hat{\lambda}(t)]-\lambda(t)$ can be rewritten as follows:  
\begin{eqnarray*}
\mathbb{E}[\hat{\lambda}(t)]-\lambda(t) &=& \frac{M}{n(T_2-T_1)}\sum_{i=1}^{n}\mathbb{E}\Big[\sum_{r=1}^{n_i}I(s_{ir}\in B_j)\Big]-\lambda(t) \\
&=& \frac{M}{n(T_2-T_1)}\sum_{i=1}^{n}\mathbb{E}\bigg[N_i\bigg(\frac{j(T_2-T_1)}{M}\bigg)-N_i\bigg(\frac{(j-1)(T_2-T_1)}{M}\bigg)\bigg]-\lambda(t) \\
&=& \frac{1}{n}\sum_{i=1}^{n}\frac{\mathbb{E}[N_i(\frac{j(T_2-T_1)}{M})-N_i(\frac{(j-1)(T_2-T_1)}{M})]}{\frac{j(T_2-T_1)}{M}-\frac{(j-1)(T_2-T_1)}{M}}-\lambda(t)
\end{eqnarray*}
From Mean value theorem and the definition of $\lambda(t)$, since all of the $n$ realizations have the same intensity function, thus, there exists $t^*\in \big[\frac{(j-1)(T_2-T_1)}{M},\frac{j(T_2-T_1)}{M}\big]$ such that for each $i=1,2,\dots,n$, $\lambda(t^*)=\frac{\mathbb{E}[N_i(\frac{j(T_2-T_1)}{M})-N_i(\frac{(j-1)(T_2-T_1)}{M})]}{\frac{j(T_2-T_1)}{M}-\frac{(j-1)(T_2-T_1)}{M}}$. Then, $\mathbb{E}[\hat{\lambda}(t)]-\lambda(t)=\frac{1}{n}\sum_{i=1}^{n}\lambda(t^*)-\lambda(t)=\lambda(t^*)-\lambda(t)$. Since $\lambda(t)$ is $L$-Lipschitz continuous, $\mathbb{E}[\hat{\lambda}(t)]-\lambda(t)=\lambda(t^*)-\lambda(t)\leq |\lambda(t^*)-\lambda(t)|\leq L|t^*-t|\leq \frac{L(T_2-T_1)}{M}$. This result can be generalized to every point $t$, then $\sup_t|\mathbb{E}[\hat{\lambda}(t)]-\lambda(t)|=O\Big(\frac{1}{M}\Big)$. 

Next, consider the variance of $\hat{\lambda}(t)$, which will be used in later proof. For any $t\in B_j$, $j=1,2,\dots,M$, $Var[\hat{\lambda}(t)]=\frac{M^2}{n^2(T_2-T_1)^2}\sum_{i=1}^{n}Var[\sum_{r=1}^{n_i}I(s_{ir}\in B_j)]$. Since $\sum_{r=1}^{n_i}I(s_{ir}\in B_j)$ denotes the total number of events in $B_j$ for the i-th realization, then $\sum_{r=1}^{n_i}I(s_{ir}\in B_j)\sim Poisson(\int_{\frac{(j-1)(T_2-T_1)}{M}}^{\frac{j(T_2-T_1)}{M}}\lambda(t)dt)$, Thus, 
\begin{eqnarray*}
Var[\hat{\lambda}(t)] &=& \frac{M^2}{n^2(T_2-T_1)^2}\sum_{i=1}^{n}\int_{\frac{(j-1)(T_2-T_1)}{M}}^{\frac{j(T_2-T_1)}{M}}\lambda(t)dt \\
&=& \frac{M^2}{n^2(T_2-T_1)^2}\sum_{i=1}^{n}\frac{\lambda(t^{**})}{M} \\
&\leq& R\cdot\frac{M}{n(T_2-T_1)^2}
\end{eqnarray*}
where $t^{**}$ is a point within $[\frac{(j-1)(T_2-T_1)}{M},\frac{j(T_2-T_1)}{M}]$. 

The remaining part is to focus on $\hat{\lambda}(t)-\mathbb{E}[\hat{\lambda}(t)]$. Denote $A_j=\frac{1}{n(T_2-T_1)}\sum_{i=1}^{n}\sum_{r=1}^{n_i}I(s_{ir}\in B_j)-\frac{1}{n(T_2-T_1)}\sum_{i=1}^{n}\mathbb{E}[N_i(\frac{j(T_2-T_1)}{M})-N_i(\frac{(j-1)(T_2-T_1)}{M})]$. Since, 
\begin{eqnarray*}
\sup_t|\hat{\lambda}(t)-\mathbb{E}[\hat{\lambda}(t)]| &=& \max_{j=1,2,\dots,M}|M\cdot A_j| \\
&=& M\cdot \max_{j=1,2,\dots,M}|A_j|
\end{eqnarray*}
Thus, for any $\epsilon >0$
\begin{eqnarray*}
\mathbb{P}\Big[\sup_t|\hat{\lambda}(t)-\mathbb{E}[\hat{\lambda}(t)]|>\epsilon\Big] &=& \mathbb{P}\Big[M\cdot \max_{j=1,2,\dots,M}|A_j|>\epsilon\Big] \\
&=& \mathbb{P}\Big[ \max_{j=1,2,\dots,M}|A_j|>\frac{\epsilon}{M}\Big] \\
&=& \mathbb{P}\Big[\bigcup_{j=1}^{M}\Big(|A_j|>\frac{\epsilon}{M}\Big)\Big] \\
&\leq& \sum_{j=1}^{M}\mathbb{P}\Big[|A_j|>\frac{\epsilon}{M}\Big] \\
&=& \sum_{j=1}^{M}\mathbb{P}\Big[A_j^2>\frac{\epsilon^2}{M^2}\Big] \\
\text{(Chebyshev's inequality)}&\leq& \sum_{j=1}^{M}\frac{Var(\frac{1}{n(T_2-T_1)}\sum_{i=1}^{n}\sum_{r=1}^{n_i}I(s_{ir}\in B_j))}{\epsilon^2/M^2}  \\
\text{(Previous result about variance)}&\leq& \sum_{j=1}^{M}\frac{R/(nM(T_2-T_1)^2)}{\epsilon^2/M^2} \\
&=& R\cdot \frac{M^2}{n(T_2-T_1)^2\epsilon^2}
\end{eqnarray*}
Therefore, 
\begin{eqnarray*}
\sup_t|\hat{\lambda}(t)-\mathbb{E}[\hat{\lambda}(t)]|=O_P\Big(\sqrt{\frac{M^2}{n}}\Big)
\end{eqnarray*}
Combine with the previous result about $\mathbb{E}[\hat{\lambda}(t)]-\lambda(t)$, the uniform convergence rate of $\hat{\lambda}(t)$ is: 
\begin{eqnarray*}
\sup_t|\hat{\lambda}(t)-\lambda(t)|=O\Big(\frac{1}{M}\Big)+O_P\Big(\sqrt{\frac{M^2}{n}}\Big)
\end{eqnarray*}
Finally, the uniform convergence rate about $\hat{\Lambda}_S^{(n)}(x)$ can be derived as follows: 
\begin{eqnarray*}
\sup_x|\hat{\Lambda}_S^{(n)}(x)-\Lambda_S(x)| &=& \sup_x\Big|\int_{T_1}^{x}\hat{\lambda}(t)dt-\int_{T_1}^{x}\lambda(t)dt\Big| \\
&=& \sup_x\Big|\int_{T_1}^{x}(\hat{\lambda}(t)-\lambda(t))dt\Big| \\
&\leq& \sup_x\int_{T_1}^{x}|\hat{\lambda}(t)-\lambda(t)|dt \\
&\leq& \sup_x\int_{T_1}^{x}\sup_t|\hat{\lambda}(t)-\lambda(t)|ds \\
\text{(The integrand is non-negative)} &\leq&\int_{T_1}^{T_2}\sup_t|\hat{\lambda}(t)-\lambda(t)|ds  \\
&=& O\Big(\frac{1}{M}\Big)+O_P\Big(\sqrt{\frac{M^2}{n}}\Big)
\end{eqnarray*}


\end{document}